\documentclass[11pt,twoside,a4paper]{article}

\usepackage{times}
\usepackage{hyperref}
\usepackage{url}
\usepackage{mathtools} % mathtools loads the amsmath package automatically
\usepackage{cases}
\usepackage{subcaption}
\usepackage{amssymb,algorithm,cite,caption,cases,float,graphicx,url,color}
\usepackage{enumerate}
\usepackage{amsthm}
\usepackage{epstopdf}

\definecolor{darkred}{RGB}{250,0,0}
\definecolor{darkgreen}{RGB}{0,150,0}
\definecolor{myblue}{RGB}{0,0,250}
\definecolor{darkblue}{RGB}{0,0,200}
\hypersetup{colorlinks=true, linkcolor=darkred, citecolor=myblue, urlcolor=darkblue}

\usepackage{microtype} % Slightly tweak font spacing for aesthetics

\usepackage[hmarginratio=1:1,top=32mm,columnsep=20pt]{geometry} % Document margins
\newtheorem*{assumption*}{\assumptionnumber}
\providecommand{\assumptionnumber}{}
\makeatletter

\newboolean{showcomments}
\setboolean{showcomments}{true}
\newcommand{\ehsan}[1]{  \ifthenelse{\boolean{showcomments}}
{ \textcolor{blue}{(Ehsan says:  #1)}} {}  }
\newcommand{\chris}[1]{\ifthenelse{\boolean{showcomments}}
{ \textcolor{magenta}{(Chris says: #1)} } {} }
\newcommand{\babak}[1]{\ifthenelse{\boolean{showcomments}}
{ \textcolor{green}{(Babak says:  #1)}}{}}

\newcommand{\distS}[2]{\mathrm{dist}^2_{#1}\left(#2\right)}

\newcommand{\dist}[2]{\mathrm{dist}_{#1}\left(#2\right)}

\newcommand{\Cauchy}{Cauchy-Schwarz }
\newcommand{\simiid}{\stackrel{\text{iid}}{\sim}}

\newcommand{\Dcall}{\Dc(\alpha,\taug,\beta,\tauh)}

\newcommand{\DKP}{D_{\Kc}}
\newcommand{\DKPo}{\overline{D}_{\Kc}}

\newcommand{\Pro}{\mathbb{P}}
\newcommand{\psiubw}{\psi(\w,\ub)}
\newcommand{\Scc}{{\Sc^c}}
\newcommand{\tauh}{{\tau_h}}
\newcommand{\taug}{{\tau_g}}

\newcommand{\env}[3]{\mathrm{e}_{{#1}}\left({#2};{#3}\right)}
\newcommand{\prox}[3]{\mathrm{prox}_{{#1}}\left({#2};{#3}\right)}
\newcommand{\chii}{\chi}

\newcommand{\cost}{\mathrm{C}_*}

\newcommand{\GZ}{\al G + Z}

\newcommand{\Lm}[2]{L\left({#1},{#2}\right)}

\newcommand{\FFm}[2]{F\left({#1},{#2}\right)}

\newcommand{\Rcn}{\Rc_n}

\newcommand{\vaG}{\hat{v}_{\al G+Z,\tau}}

\newcommand{\vaGar}[2]{\hat{v}_{\al #1+#2,\tau}}
\newcommand{\vaxGar}[2]{\hat{v}_{(\al+x)#1+#2,\tau+y}}

\newcommand{\ellaG}{\ell'_{\al G+Z,\tau}}
\newcommand{\ellaxG}{\ell'_{(\al+x)G+Z,\tau+y}}

\newcommand{\ellaGar}[2]{\ell'_{\al #1+#2,\tau}}
\newcommand{\ellaxGar}[2]{\ell'_{(\al+x)#1+#2,\tau+y}}

\newcommand{\xio}{\overline{\phi}}
\newcommand{\phio}{\overline{\phi}}

\theoremstyle{theorem}

\newtheorem{thm}{Theorem}[section]

\newtheorem{lem}{Lemma}[section]

\newtheorem{cor}{Corollary}[section]
\newtheorem{ass}{Assumption}%[section]

\theoremstyle{remark}
\newtheorem{remark}{Remark}[subsection]

\theoremstyle{definition}

\newcommand{\consist}{\cite[Thm.~2.7]{NF36} }
\newcommand{\uniform}{\cite[Cor..~II.1]{AG1982} }

\newcommand{\al}{\alpha}

\newcommand{\eps}{\epsilon}

\newcommand{\sign}{\mathrm{sign}}

\newcommand{\I}{\mathbf{I}}

\newcommand{\xx}{\overline{\x}}

\newcommand{\mn}{{}}
\newcommand{\Exp}{\mathbb{E}}               % expectation
\newcommand{\E}{\mathbb{E}}                    % expectation
\newcommand{\la}{{\lambda}}                     % lambda

\newcommand{\sig}{\sigma}

\newcommand{\nn}{\notag}
\newcommand{\R}{\mathbb{R}}

\newcommand{\Z}{{Z}}

\newcommand{\G}{\mathbf{G}}

\newcommand{\X}{\mathbf{X}}

\newcommand{\A}{\mathbf{A}}

\newcommand{\x}{\mathbf{x}}
\newcommand{\w}{\mathbf{w}}
\newcommand{\ub}{\mathbf{u}}
\newcommand{\g}{\mathbf{g}}
\newcommand{\vb}{\mathbf{v}}

\newcommand{\tb}{\mathbf{t}}
\newcommand{\y}{\mathbf{y}}
\newcommand{\s}{\mathbf{s}}
\newcommand{\z}{\mathbf{z}}

\newcommand{\ab}{\mathbf{a}}

\newcommand{\h}{\mathbf{h}}

\newcommand{\deltab}{\boldsymbol{\delta}}

\newcommand{\Sc}{{\mathcal{S}}}

\newcommand{\Yc}{\mathcal{Y}}
\newcommand{\Dc}{\mathcal{D}}
\newcommand{\Xc}{\mathcal{X}}

\newcommand{\Kc}{\mathcal{K}}
\newcommand{\Nc}{\mathcal{N}}
\newcommand{\Rc}{\mathcal{R}}
\newcommand{\Nn}{\mathcal{N}}

\newcommand{\Cc}{\mathcal{C}}
\newcommand{\Gc}{{\mathcal{G}}}
\newcommand{\Ac}{\mathcal{A}}

\newcommand{\Ec}{\mathcal{E}}
\newcommand{\Hc}{\mathcal{H}}

\newcommand{\Jc}{\mathcal{J}}
\newcommand{\Ic}{\mathcal{I}}

\newcommand{\vh}{\hat{v}}
\newcommand{\beq}{\begin{equation}}
\newcommand{\eeq}{\end{equation}}
\newcommand{\bea}{\begin{align}}
\newcommand{\eea}{\end{align}}

\newcommand{\vp}{\vspace{4pt}}

\newcommand{\rP}{\xrightarrow{P}}

\newcommand{\loss}{\mathcal{L}}

\title{
Precise Error Analysis of Regularized $M$-estimators \\ in High-dimensions
}

\author{\vspace{10pt}
Christos Thrampoulidis,\hspace{10pt}Ehsan Abbasi,\hspace{10pt}Babak Hassibi\vspace{10pt}%\\Department of Electrical Engineering \\ Caltech, Pasadena -- 91125
\thanks{Department of Electrical Engineering, Caltech, Pasadena -- 91125, 
emails: {\tt(cthrampo, eabbasi, hassibi)@caltech.edu}. }
}

\begin{document}

\maketitle

\begin{abstract}

A popular approach for estimating an unknown signal $\x_0\in \R^n$ from noisy, linear measurements $\y=\A\x_0+\z\in \R^m$ is via solving a so called \emph{regularized} M-estimator: $\hat\x:=\arg\min_\x \loss(\y-\A\x)+\lambda f(\x)$. Here, $\loss$ is a convex loss function, $f$ is a convex (typically, non-smooth) regularizer, and, $\lambda> 0$ is a regularizer parameter. We analyze the squared error performance $\|\hat\x-\x_0\|_2^2$ of such estimators in the \emph{high-dimensional proportional regime} where $m,n\rightarrow\infty$ and $m/n\rightarrow\delta$. The design matrix $\A$ is assumed to have entries iid Gaussian; only minimal and rather mild regularity conditions are imposed on the loss function, the regularizer, and on the noise and signal distributions. 
We show that the squared error converges in probability to a nontrivial limit that is given as the solution to a minimax convex-concave optimization problem on four scalar optimization variables. We identify a new summary parameter, termed the Expected Moreau envelope to play a central role in the error characterization. The \emph{precise} nature of the results permits an accurate performance comparison
between different instances of regularized M-estimators  and allows to optimally tune the involved parameters (e.g. regularizer parameter, number of measurements). 
The key ingredient of our proof is the \emph{Convex Gaussian Min-max Theorem} (CGMT) which is a tight and strengthened version of a classical Gaussian comparison inequality that was proved by Gordon in 1988.
\end{abstract}

%\newpage

%%\vspace{-5pt}
\section{Introduction}
%%\vspace{-3pt}
%\vp

\subsection{Motivation}

\vp
\noindent\textbf{Structured signals in high-dimensions }.
%\noindent\textbf{Noisy Linear Measurements}.
We consider the standard problem of recovering an unknown signal $\x_0\in\R^n$ from a vector $\y\in\R^m$ of $m$ noisy, linear observations given by
$
\y = \A \x_0 + \z\in \R^m.
$
Here, $\A\in\R^{m\times n}$ is the (known) measurement matrix, and, $\z\in\R^m$ is the noise vector; the latter is generated from some distribution density in $\R^m$, say $p_\z$. 
%\vp
%\noindent\textbf{Structured signals in high-dimensions }.
Our focus is on the \emph{high-dimensional regime}  where both the dimensions of the ambient space $n$ and the number of measurements $m$ are large \cite{serdobolskii2013multivariate,donoho2000high}. This is  different than the  classical one, where $n$ is small and fixed and only $m$ is assumed large. Of special interest is  
%In the high-dimensional regime, 
the scenario of \emph{compressed} measurements, in which  $m<n$.
%, is of special interest. 
In principle, such inverse problems are ill-posed, unless the unknown vector is somehow structurally constrained to only have very few degrees of freedom relative to its ambient space. Such signals are called \emph{structured} signals; popular examples include sparsity, block-sparsity, low-rankness, etc. \cite{bach2010structured,Cha}. We model such structural information on $\x_0$ by assuming that it is sampled from an $n$-dimensional probability density $p_{\x_0}$. %For an illustration, we might consider modeling a \emph{sparse} signal $\x_0$ as $\x_{0,i}\simiid \rho\delta_0 + (1-\rho)p_x$, where $\rho$ is the (average) sparsity level, $\delta_0$ is the Dirac delta function at zero and $p_x$ is a density on $\R$. 

\vp
\noindent\textbf{Regularized M-estimators.}
The most widely used  approach to obtain an estimate $\hat\x$ of the unknown $\x_0$ from the vector $\y$ of observations is via solving the  \emph{convex}  program
\begin{align}\label{eq:genM}
\hat\x:=\arg\min_\x~\loss(\y-\A\x) + \la f(\x).
\end{align}
%for appropriate choices of the
The \emph{loss function} $\loss:\R^m\rightarrow\R$ measures the  deviation of $\A\hat\x$ from the observations $\y$,  the \emph{regularizer} $f:\R^n\rightarrow\R$ aims to promote the particular structure of $\x_0$, and, the regularizer parameter $\la>0$ balances between the two. Henceforth, both $\loss$ and $f$ are assumed to be convex. Also, $f$ will typically be non-smooth.
%Typical examples of loss functions include 
%\subsection{Regularized M-Estimators}
We refer to the minimization problems of the form in \eqref{eq:genM} as  \emph{regularized M-estimators}. Different choices of the loss function and of the regularizer give rise to a number of well-known estimators. 
%Among popular choices of the loss function are of course the least-squares, the least-absolute deviation, the Huber loss and so on. As a concrete example, 
A few concrete examples might suffice:
%\begin{itemize}
%\item 
(i) 
the \emph{LASSO} \cite{TibLASSO} corresponds to \eqref{eq:genM} with $\loss(\vb)=\frac{1}{2}\|\vb\|_2^2$ and $f(\x) = \|\x\|_1$. General choices of the regularizer for the same loss function lead to the Generalized LASSO \cite{OTH13,plan2015generalized} 
%\item 
(ii) The \emph{regularized-LAD} \cite{wang2013} minimizes an $\ell_1$-loss function. 
%\item 
(iii)
The (generalized) \emph{square-root LASSO} \cite{Belloni} solves \eqref{eq:genM} for $\loss(\vb)=\|\vb\|_2$.  
%\end{itemize}
%We are interested in the estimation performance of such regularized M-estimators, when measured as the th
In the first two examples  the loss function is \emph{separable} over its entries, i.e. $\loss(\vb) = \sum_{j=1}^m\ell(\vb_j)$ for convex $\ell:\R\rightarrow\R$; in contrast, the square-root LASSO does not belong to this category. Accordingly, the regularizer function might be separable (e.g. $\ell_1$-norm) or not (e.g. nuclear-norm). 

\vp\noindent
\noindent\textbf{Challenge}.
A popular way to compare performance among different instances of \eqref{eq:genM} is by the \emph{squared-error} $\|\hat\x-\x_0\|_2^2$. In the absence of the regularizer function $f$, the family of estimators in \eqref{eq:genM} corresponds to the ``plain-vanilla"
regression M-estimators and there is a complete, practical and elegant theory developed in the statistics
literature that analyzes its asymptotic performance.
%In the absence of the regularizer function f, the family of estimators in (1) corresponds to the plain-vanilla
%1
%regression M-estimators, for which there is a complete, practical and elegant theory developed in the statistics
%literature. 
This theory includes some of the most popular notions and results in statistics, such as conditions on
the optimality of Maximum Likelihood (ML) estimators, the theory of robust statistics \cite{huber2011robust}, etc.. Unfortunately,
it only holds under an assumption of many observations (large m) of only a few well-chosen variables to be
estimated (small $n$), and thus, it fails to capture the following prevailing features of modern applications: (a)
large number of variables to be estimated (large $n$); (b) (often) fewer observations than variables ($m < n$); (c)
the unknown signal $\x_0$ is structured.
%
%
%A general and powerful theory that analyzes its  asymptotic behavior in the regime of large $m$ and fixed $n$, and, which answers a number of critical optimality questions regarding the  choice of the loss function in different contexts, e.g. in the presence of outliers, has been developed in the statistics literature, e.g. \cite{huber2011robust,NF36,liese2008statistical}.
%%
%%The asymptotic behavior of this quantity in the traditional regime of large $m$ and fixed $n$ is a classical subject in the statistics literature, e.g. \cite{huber2011robust,NF36,liese2008statistical},  The general and powerful theory that has been developed has provided answers to a number of critical optimality questions such as the optimal choice of the loss function under different settings, e.g. in the presence of outliers \cite{huber2011robust,serdobolskii2013multivariate}. 
%Yet, this traditional model is insufficient to capture the dimensionality explosion that occurs in modern applications. 
Therefore, an \emph{extension of the theory to the high-dimensional regime is of interest}. In fact, the roots of such a question are quite old and date back to the works of Huber, Kolmogorov, and others (see \cite{huber2011robust,serdobolskii2013multivariate} and references therein). Nonetheless, and despite several remarkable recent advances, we still lack a general and clear theory that would resemble that of the traditional regime.

%\vp
%\noindent\textbf{Recent Developments.}

\subsection{Contribution}\label{sec:contr}
%\vp
%\noindent\textbf{Contribution.}

\vp
\noindent\textbf{Error Prediction}. 
In this work, we characterize the (mean) squared-error performance of the generalized M-estimator in \eqref{eq:genM} under the following setting:
%\begin{itemize}
%\item 

--~ \emph{high-dimensional proportional regime}: $m,n\rightarrow\infty$ with $m/n\rightarrow\delta\in(0,\infty)$,
%\item 

--~ \emph{Gaussian design}: $\A$ has entries iid Gaussian,
%\item 

--~ \emph{Regularity conditions}:  only \emph{minimal} and \emph{generic} conditions are imposed on the loss function, the regularizer, and the noise and signal statistics.
%\end{itemize}

We show that  the squared error converges in  probability to a nontrivial limit  which is given as the unique minimizer to a deterministic convex optimization problem that only involves four scalar optimization variables.
%
% is given as the \emph{unique} minimizer to a Deterministic Optimization (DO) problem that involves three more scalar variables. 
 The normalized number of measurements $\delta$ and the regularizer parameter $\la$ appear in the objective function of the optimization explicitly. In contrast, the loss function $\loss$ and the noise distribution $p_\z$ appear through a summary functional, which we call the \emph{Expected  Moreau Envelope}. The same holds for the regularizer $f$ and the distribution of the signal $p_{\x_0}$.

\vp
\noindent\textbf{Expected Moreau Envelope}. The Expected Moreau Envelope $L(c,\tau) := L_{\ell,p_\z}(c,\tau)$ is defined for all $c\in\R, \tau>0$ as the converging limit of $\frac{1}{n}\env{\loss}{c\g+\z}{\tau}$,
%\begin{align}\nn%\label{eq:expM}
%L(c,\tau) := L_{\ell,p_\z}(c,\tau):= \Exp_{\substack{\z\sim p_\z\\ \g\sim\Nn(0,\I_m) }}\left[\env{\loss}{c\g+\z}{\tau}\right],
%\end{align}
where $\env{\loss}{\ub}{\tau}:=\min_{\vb}\frac{1}{2\tau}\|\ub-\vb\|_2^2+\loss(\vb)$ denotes the Moreau-envelope approximation of $\loss$ at $\ub$ with parameter $\tau$ and  $\g$ is a vector with entries iid standard normal. This is  the critical parameter that determines the role of the loss function and that of the regularizer in the error performance of \eqref{eq:genM}. It has some key properties: it is smooth irrespective of any smoothness assumptions on $\loss$, and, is strictly convex under mild assumptions. Also, it is insightful to view it as a \emph{generalization} of  corresponding summary parameters, such as the ``gaussian width"\cite{Cha} and the ``statistical dimension"\cite{TroppEdge},  that are geometric in nature, and, which play a fundamental role in the study of phase transitions in noiseless linear inverse problems.

\vp
\noindent\textbf{Generality}.~ A key feature of our result is that it holds under very general settings. 
%This includes non-smooth, non-separable and not necessarily strongly convex loss and regularizer functions. Also, it allows for noise distributions with unbounded moments. For example, our main theorem can be applied to predict the error performance of M-estimators with Huber- or $\ell_1$- loss functions under heavy-tailed noise (e.g. $\mathrm{Cauchy}(0,1)$). Also, it predicts the performance of the square-root LASSO, in which the loss function is non-separable. 
%
All existing results in the literature on the performance of specific instances of M-estimators can be seen as special cases of the main theorem of this work (Theorem \ref{thm:master}). Beyond those, the theorem can be used to derive a wide range of novel results, including instances where the loss function and the regularizer may be non-smooth or non-separable, and where the noise distribution may have unbounded moments. %An example, is an M-estimator with Huber- or $\ell_1$- loss functions under heavy-tailed noise. 

\vp
\noindent\textbf{Opportunities}.~ 
The \emph{precise} characterization of the squared error permits an accurate performance comparison
between different instances of \eqref{eq:genM}. Hence, the main theorem of this work lays the groundwork towards developing a complete theory of regularized M-estimators in the high-dimensional regime. This involves providing rigorous answers to  optimality questions regarding the choice of the involved parameters:
%
%
%Being able to predict the error performance as a function of the involved parameters $\loss,f,p_\z,p_{\x_0},\la$ and $\delta$, and, under generic settings, is the first fundamental step towards developing a complete theory regarding regularized M-estimators in the high-dimensional regime.  Part of this theory would be providing answers to optimality questions regarding the choice of the involved parameters:
\begin{itemize}
\item[--] What is the optimal loss function and regularizer, under different settings, e.g. in the presence of outliers, particular structure of $\x_0$, etc.?
\item[--] What is the minimum achievable squared error in each one of those scenarios? Do there exist consistent M-estimators, i.e. instances for which $\|\hat\x-\x_0\|_2\rightarrow0$?
\item[--] How to optimally tune the regularizer parameter $\lambda$?
\item[--] How does the sampling ratio $\delta=m/n$ affect the error?
%\item 
\end{itemize}
Given the popularity of M-estimators, the questions above are clearly of both theoretical and practical interest. 
Only partial answers that apply to special cases and to only few of them are known in the literature, while most remain open and challenging. We envision that the main theorem of this work gets us a step closer to overcoming the challenge and to exploring phenomena that are new when compared to what is known in the classical statistics regime. Although, this goes beyond the scope of the current paper, we have included some preliminary results and discussions to illustrate those potentials.

%%%\vspace{-20pt}
%\noindent\emph{Seperable functions}: We specialize our results to the popular case of separable loss functions and regularizers, and, of iid noise and signal distributions. The prediction is true for non-differentiable functions and for distributions with unbounded second moments, e.g. regularized LAD under iid Cauchy noise. The error is predicted by solving a system of four nonlinear equations in four unknowns, which corresponds to the first-order optimality conditions of the corresponding (DO) problem.  Empirical observations suggest that this can be done fast by a simple iterative scheme.
% converges fast to the u those equations for all the instances considered in this paper. 
%This proves to be very useful for the simulation results presented later in the main body of the paper.

%%\vspace{-20pt}
\vp
\noindent\textbf{Convex Gaussian Min-max Theorem}.~ The main ingredient of the proof of our main result is the Convex Gaussian Min-max Theorem (CGMT). The CGMT is a generalization and a strengthened version of a classical Gaussian comparison inequality due to Gordon, which dates back to 1988 \cite{GorLem,GorThm}. While Gordon's original result only provides lower bounds, the CGMT shows that the results become \emph{tight} when additional convexity assumptions are imposed. 
The idea of combining Gordon's inequality with convexity is attributed to Stojnic, who used it to analyze the high-SNR performance of the constrained LASSO \cite{StoLASSO}. The CGMT solidifies and adds upon this initial idea. The final result leads to a transparent and readily applicable framework which is  powerful enough to be  useful under  the general framework of the current paper. In fact, the CGMT in the generality that it appears here\footnote{An early version appears by subset of the authors in \cite{COLT15}.}, might be of independent interest and may have applications that go beyond the scope of our work. Finally, it should be noted, that the successful application of the CGMT to the analysis of regularized M-estimators involves a number of new ideas that are introduced as part of this work.

%The tool that underlies our analysis is the recently developed Convex Gaussian Min-max Theorem (CGMT) framework of \cite{COLT15}, which is a generalization and stronger version of a classical Gaussian comparison inequality of Gordon that dates back to 1988 \cite{GorLem,GorThm}. The CGMT associates with a Primary Optimization (PO) problem a simplified Auxiliary Optimization
%(AO), from which we can tightly infer properties of the original (PO). Roughly speaking, if the optimal minimizer of the (AO) belongs to some set $\Sc$ with probability approaching one in the limit of problem dimensions, then, the same is true for the optimal minimizer of the (PO).  We manage to write the general $M$-estimator in \eqref{eq:genM} as a (PO). Thus we analyze the corresponding (AO) problem instead of the original (PO).

% Next, we analyze the error of the (AO) and translate the result to the (PO) thanks to the CGMT.

\subsection{Related Work}\label{sec:rel}
With the advent of Compressed Sensing there is a very large number of theoretical results  that have appeared in recent years  in place for various types of regularized M-estimators. The vast majority of those results  hold under standard incoherence or restricted eigenvalue conditions on the measurement matrix $\A$ \footnote{Such conditions have been shown to be satisfied by a wide class of randomly designed measurement matrices, (e.g. \cite{foucart2013mathematical,eldar2012compressed,davenport2011introduction} and references therein). A more recent line
of works obtains similar order-wise bounds under even weaker assumptions on the randomness properties of $\A$ \cite{lecue2014sparse,tropp2014convex,sivakumar2015beyond}.},
but they are \emph{order-wise} in nature, i.e., they characterize the error performance only up to {loose} constants. While this line of work includes unifying frameworks for the analysis of general instances of \eqref{eq:genM}, the loose constants involved in the error bounds do \emph{not} permit any accurate comparisons among the different instances (e.g. \cite{negahban2012unified, wainwright2014structured, banerjee2014estimation,li2015geometric} and references therein); therefore, they cannot be used to answer optimality questions of the nature discussed in Section \ref{sec:contr}.

This paper derives precise characterizations of the error behavior (ones that do not involve unknown constants). Results of this nature have appeared in the literature under the additional assumption of an iid Gaussian distribution imposed on the entries of  the matrix $\A$. The inspiration behind these studies can be traced back to the seminal work of Donoho \cite{donoho2006high,donoho2009counting} on the phase-transition of $\ell_1$-minimization in the Compressed Sensing problem. This and the extensive follow-up literature mostly focused on the \emph{noiseless} signal recovery problem. More recently, researchers have initiated the study of the exact reconstruction error of instances of \eqref{eq:genM} in the presence of \emph{noise}.
 Unfortunately, no unifying treatment that holds for general instances has hitherto been available. To the best of our knowledge, our work is the first to obtain \emph{precise} characterizations of the error performance of \eqref{eq:genM} for \emph{general} convex loss functions, convex regularizers, and noise and signal distributions under a Gaussian assumption on the random measurement matrix $\A$. In the rest of this section, we briefly outline the relevant literature. A far more detailed discussion appears  in Section \ref{sec:prior}. 

%
%More recently, 
%
% We review the most relevant of those results here and include a further detailed discussion in Section \ref{sec:}. Before that, we remark that no unifying treatment of \eqref{eq:genM} that holds for general instances has hitherto been available; our work is the first to obtain \emph{precise} characterizations of the error performance of \eqref{eq:genM} for \emph{general} convex loss functions, convex regularizers, and noise and signal distributions under a Gaussian assumption on the random measurement matrix $\A$.
%
%analysis, we preserve the Gaussianity assumption on $\A$  and obtain \emph{precise} characterizations of the error performance of \eqref{eq:genM} for \emph{general} convex loss functions, convex regularizers, and noise and signal distributions\footnote{\chris{Comment on Gaussian}}. 
%In this section, we include a short discussion on the relevant literature on existing precise results regarding the performance of convex regularized M-estimators. A more detailed 

% review the relevant literature on existing precise results on the performance of convex regularized M-estimators. 

The first precise results on the performance of non-smooth convex optimization methods appear in the literature in the context of \emph{noiseless} linear inverse problems that arise in Compressed Sensing. Here, the vector of measurements of the unknown structured signal $\x_0$ takes the form $\y=\A\x_0\in\R^m$ and recovery is attempted via solving $\min_{\y=\A\x} f(\x)$, for an appropriately chosen convex regularizer $f$. In the absence of noise, the standard measure of performance becomes that of the minimum number of measurements required for \emph{exact recovery} of $\x_0$. By now, there is an elegant and complete theory that precisely characterizes this number when $\A$ has entries iid Gaussian. The theory was built in a series of recent papers \cite{donoho2009counting, Sto, Cha, bayati2015universality, TroppEdge, stojnic2013upper, OTH13}. See Section \ref{sec:prior} for details on the contribution of each reference.
Our work extends the analysis to the \emph{noisy} setting. In the presence of noise, the analysis is inherently more challenging since: (a) one
needs to characterize the precise value of the estimation error, rather than just discriminating between exact
recovery or not; (b) the performance depends not only on the number of measurements but also on the noise and signal statistics. Also, it naturally includes the results of the noiseless case as special instances. However, many of the ideas, analytical tools and concepts developed in the works \cite{Sto,Cha,TroppEdge,stojnic2013upper} have proved to be useful in extending the results to the noisy setting. % the analysis of the more general setting of the current work.

In the noisy setting, the first precise results analyzed the error performance of regularized least-squared (a.k.a. generalized-LASSO) under an iid gaussianity assumption on the noise distribution \cite{DMM,montanariLASSO,StoLASSO,OTH13,ICASSP15,ell22,IRO}. 
%It has  been only very recently that \emph{precise} predictions on the error performance of instances of \eqref{eq:genM} have appeared in the literature \cite{DMM,montanariLASSO,StoLASSO,OTH13,ICASSP15,ell22,IRO,Allerton14}. Still,  those references study only very specific instances, mainly different variations of the LASSO under iid gaussian noise distribution. 
%More recently
It has been only very recently, that El Karoui \cite{karoui13,karoui15}, and, Donoho and Montanari \cite{montanari13,montanari15} were able to rigorously\footnote{The study of high-dimensional M-estimators has been previously considered in \cite{el2013robust,bean2013optimal}. However, those results are only based on heuristic arguments and simulations.}
predict the error performance of M-estimators under more general assumptions on the loss function and on the noise distribution. However, the papers by Donoho and Montanari assume \emph{no} ragularization and El Karoui considers the special case of ridge regularization. Finally, the very recent paper \cite{bradic2015robustness} builds upon \cite{montanari13} and extends the study to the case of $\ell_1$-regularization. 
%  in the absence of a regularizer\footnote{To be precise, \cite{karoui15} also considers ridge-regression. This is only a special instance of our general result (cf. Section \ref{sec:ridge}).}.
%More generally, we allow for non-smooth regularizers, which are typically used in the recovery of structured signals.}. 
%Their results are the first\footnote{} to identify some of the key new phenomena that appear in the high-dimensional analysis compared to the traditional regime. 
We compare these results to ours in several places in the manuscript; also, see Section \ref{sec:prior} for a very detailed discussion. 
In short, our work achieves by several means a  more complete and transparent treatment of the subject, overcoming the limitations of previous endeavors as follows:  %significantly extends on the results above in several directions with the most notable being as follows: 
(i) We consider arbitrary convex regularizers, (ii)  We identify minimal and generic assumptions under which the general result holds, (iii) We remove any smoothness and strong-convexity assumptions on the loss function, which are required in all previous works. Also, the loss function (and regularizer) need not be separable (e.g., we allow $\loss(\vb) = \|\vb\|_2$ or $\|\vb\|_\infty$), and, the distributions need not be iid. (iv) We remove boundedness assumptions on the moments of the noise distribution. %\footnote{This is also the case in \cite{karoui15}, but not in \cite{montanari13}.}.
%\begin{itemize}
%\item add a regularization term as in \eqref{eq:genM}. %Hence, our results 
%\item identify minimal and generic assumptions under which our results hold.
%\item remove any smoothness and strong-convexity assumptions on the loss function. Also, the loss function (and regularizer) need not be separable and the distributions need not be iid.
%\item remove boundedness assumptions on the moments of the noise distribution\footnote{This is also the case in \cite{ElKarou15}, but not in \cite{Montanari13}.}.
%\end{itemize}
Notably, our proof technique is fundamentally different than that of  \cite{karoui15} and \cite{montanari13}, and, it appears to be more direct and insightful in several ways.

\subsection{Organization}\label{sec:rel}
 The rest of the paper is organized as follows. In Section \ref{sec:pre}, we introduce some basic notions and set-up the problem. The main theorem (Theorem \ref{thm:master}) is presented next in Section \ref{sec:res}, where its features and implications are also discussed. 
%The rest of the paper is organized as follows. 
%In Section \ref{sec:res}, we state our main result in its full generality.
% A detailed discussion on the interpretation of the accompanying assumptions and on the implications of the theorem is included is the same section. 
 Theorem \ref{thm:master} is specialized to instances of M-estimators with separable loss and regularizer functions in Section \ref{sec:sep}. A number of examples of M-estimators and relevant numerical simulations are included in Section \ref{sec:sim} to illustrate the applicability and the premises of the result. In Section \ref{sec:outline}, we introduce the mechanics that lead to the proof of Theorem \ref{thm:master}; this includes the statement of the Convex Gaussian Min-max Theorem in Section \ref{sec:CGMT}. Section \ref{sec:rel} discusses the relevant literature in some detail. Finally, the paper concludes in Section \ref{sec:conc} with a discussion on several promising directions of future research. The proofs of the results of all the sections are deferred to Appendices \ref{sec:proof}-\ref{sec:proofs_4}

\section{Preliminaries}\label{sec:pre}

\subsection{Notation}
We gather here the basic notation that is used throughout the work. 

\vp
\noindent\emph{Convex Analysis:}~ For a convex function $f:\R^n\rightarrow\R$, we let $\partial f(\x)$ denote the subdifferential of $f$ at $\x$ and $f^*(\y)=\sup_{\x}\y^T\x-f(\x)$  its Fenchel conjugate. The Moreau envelope function of $f$ at $\x$ with parameter $\tau$ is defined by 
$$\env{f}{\x}{\tau}:=\min_{\vb}\frac{1}{2\tau}\|\x-\vb\|_2^2 + f(\vb).$$
The optimal value in the minimization above is denoted by $\prox{f}{\x}{\tau}$. When writing $\x_* = \arg\min_\x f(\x),$ we let the  operator $\arg\min$ return any one of
the possible minimizers of $f$.
%
%Using a slight abuse of notation, we let the  operator $\arg\min$ return any one of
%those optimal values.

\vp
\noindent\emph{Limits and Derivatives:}~ For a real-valued (not necessarily differentiable) convex function $f$ on $\R$ denote 
%its right derivative functions as $f^\prime_+$, and, 
$$
f^\prime_+(v) := \sup_{s\in\partial f(v)} |s|.
$$
Also, write  $\lim_{x\rightarrow c^+}f(x)$ for the one-sided limit of $f$ at $c$, as $x$ approaches from above. For a  function $g(x,\tau)$ that is continuously differentiable on $\R^2$, we write $g^\prime(x,\tau)$ or $g_1(x,\tau)$ for the derivative with respect to the first variable, and, $g_2(x,\tau)$ for the derivative with respect to the second variable.

\vp
\noindent\emph{Probability:}~ The symbols $\Pro\left(\cdot\right)$ and $\Exp\left[\cdot\right]$ denote the probability of an event and the expectation of a random variable, respectively. For a sequence of random variables $\{\mathcal{X}^{(n)}\}_{n\in\mathbb{N}}$ and a constant $c\in \R$ (independent of $n$), we write $\{\mathcal{X}^{(n)}\}_{n\in\mathbb{N}}\rP c$, to denote convergence \emph{in probability}, i.e. $\forall\eps>0$, $\lim_{n\rightarrow\infty}\Pro\left(|\mathcal{X}^{(n)}-c|>\eps\right)=0.$ We write $\mathcal{X}\sim p_X$ to denote that the random variable $\mathcal{X}$ has a density $p_X$. If $\mathcal{X}$ is a vector random variable with entries iid, then we use $\simiid$. Also, $\mathcal{X}\sim\Nn(\mu,\sigma^2)$  denotes a Gaussian random variable with mean $\mu$ and variance $\sigma^2$.

We reserve the letters $\g$ and $\h$ to denote standard Gaussian vectors (with iid entries $\Nn(0,1)$) of dimensions $m$ and $n$, respectively. Similarly, $G$ and $H$ are reserved to denote (scalar) standard normal random variables.

%\vp
%\noindent\emph{Constants:}~ We say $C\in\R$ is a constant if it is independent of the problem 

\subsection{Setup} \label{sec:model}
\indent \emph{Linear Asymptotic Regime}:  Our study falls into the linear asymptotic regime in which the problem dimensions $m$ and $n$ grow proportionally to infinity with 
$${m}/{n}\rightarrow\delta\in(0,\infty).$$

%In particular, we consider a 
%sequence of problem instances $\{\xx_0^\mn, \A^\mn, f^\mn, m^\mn\}_{n\in\mathbb{N}}$ indexed by $n$ such that 
%%$p_{\xx_0}^\mn:\R^n\rightarrow[0,1]$ is a probability density, $\A^\mn\in\R^{m\times n}$, $f^\mn:\R^n\rightarrow\R$ , and,  
% $\A^\mn\in\R^{m\times n}$ has entries i.i.d. $\Nn(0,1)$, $f^\mn:\R^n\rightarrow\R$ is proper convex, and,
%$m:=m^\mn$ with $m= {\delta n}, \delta\in(0,\infty)$. 

\vspace{5pt}
\noindent\emph{Measurement matrix:} The entries of $\A\in\R^{m\times n}$ are i.i.d. $\Nn(0,\frac{1}{n})$. The normalization of the variance
ensures that the rows of $\A$ are approximately unit-norm; this is necessary in order
to properly define a signal-to-noise ratio.

\vp
\noindent\emph{{Unknown (structured) signal:}} Let $\x_0\in\R^n$ represent the unknown signal vector that is sampled from a probability density $p_{\x_0}\in\R^n$ with one dimensional marginals that are independent of $n$. Note, that we do \emph{not} necessarily require that the entries of $\x_0$ be iid. The signal $\x_0$ is assumed independent of $\A$. 

Information about the structure of  $\x_0$ is encoded in $p_{\xx_0}$. For instance, to study an $\x_0$ which is sparse, it is typical to assume that its entries are i.i.d. $\x_{0,i}\sim(1-\rho)\delta_0+\rho q_{\X_0}$, where $\rho\in(0,1)$ becomes the normalized sparsity level, $q_{\X_0}$ is a scalar p.d.f. and $\delta_0$ is the Dirac delta function\footnote{Such models in place for studying structured signals have been widely used in the relevant literature, e.g. \cite{DJ94,DMM,donoho2013accurate}. In fact, the results here continue to hold as long as the marginal distribution of $\x_0$ converges to a given distribution (as in \cite{montanariLASSO}).}.

\vspace{5pt}
\noindent\emph{{Regularizer:}} We consider regularizers $f:\R^n\rightarrow\R$ that are proper continuous \emph{convex}  functions. 

%On the technical side, we require that for any $n\in\mathbb{N}$ and any $\|\ub\|_2\leq C$, it holds $f(\ub)\leq c\sqrt{n}$, for constants $c,C$. This is no more than a normalization condition on $f$.

\vspace{5pt}
\noindent\emph{{Loss function:}} The loss function $\loss:\R^n\rightarrow\R$ is  proper continuous and convex. 
% It is common in the context of M-estimators that $\loss$ is chosen separable over its arguments, in which case $\loss(\vb) = \sum_{i=1}^m{\ell(\vb_i)}$ for some $\ell:\R\rightarrow\R$. Note that our result in its most general form goes beyond separable functions, e.g. $\loss(\vb)=\|\vb\|_2$ or $\loss(\vb)=\|\vb\|_\infty$.
Without loss of generality, we assume for simplicity that
$
\min_{\vb}\loss(\vb) = 0.
$
Finally, we impose a  natural normalization condition as follows:
for all $n\in\mathbb{N}$ and all constants  $c>0$ there exists constant $C>0$, such that $\|\vb\|_2\leq c\sqrt{n} \implies \sup_{\s\in\partial\loss(\vb)}\|\s\|_2\leq C\sqrt{n}$. %For instance, this is of course the case for all the examples of loss functions discussed above.
%\begin{assumption}{(L)}{1}\label{L1}
%\end{assumption}

%\underline{\emph{Technical Assumption {\textbf{(L)}}}.} At least one of the following two holds true:
%\begin{enumerate}
%\item  there exists convex interval $\Ic\in\R$
%\end{enumerate}

\vspace{5pt}
\noindent\emph{Noise vector:} The noise vector $\z\in\R^m$ follows a probability distribution $p_{\z}\in\R^m$ with one dimensional marginals that are independent of $n$. Also, it is independent of the measurement matrix $\A$.
%The entries of $\z$ are nor necessarily iid. Furthermore, we do not require 

%\vspace{5pt}
%\textbf{{Loss function}}. We consider regularizers $\loss:\R^n\rightarrow\R$ that are closed \emph{convex}  functions. It is common in the context of M-estimators that $\loss$ is chosen separable over its arguments, i.e. $\loss(\vb) = \sum_{i=1}^n{\ell(\vb_i)}$ and $\ell:\R\rightarrow\R$; note that our result in its most general form is not restricted to this case.  On the technical side, we require that for any $n\in\mathbb{N}$ and any $\|\vb\|_2\leq C\sqrt{n}$, it holds $\sqrt{n}\loss(\vb)\leq c$, for constants $c,C$. This is no more than  a normalization condition on $\loss$. 

\vspace{5pt}
\noindent\emph{Sequence of problem instances:} 
Formally, our result applies on a sequence of problem 
instances \\ $\{\x_0^\mn, \A^\mn, \z^\mn, \loss^\mn, f^\mn,  m^\mn\}_{n\in\mathbb{N}}$ indexed by $n$ such that 
%$p_{\xx_0}^\mn:\R^n\rightarrow[0,1]$ is a probability density, $\A^\mn\in\R^{m\times n}$, $f^\mn:\R^n\rightarrow\R$ , and,  
the properties listed above hold for all members of the sequence and for all $n\in\mathbb{N}$. (We do not write out the subscripts $n$ for arguments of the sequence to not overload notation).
Every such sequence generates a sequence $\{\y^\mn,\hat\x^\mn\}_{n\in\mathbb{N}}$ where  $\y^\mn := \A^\mn\x_0^\mn + \z^\mn$, and,
\begin{align}\label{eq:genLASSO}
\hat\x^\mn := \arg\min_{\x} ~\loss^\mn\left(\y^\mn - \A^\mn\x\right) + \la f^\mn(\x) .
\end{align}
Here, $\la>0$ is a fixed regularizer parameter.
% When clear from the context, we will often drop the superscripts $(n)$ to enlighten notation.

%Some remarks regarding the technical conditions of this section are in place.
%First, condition 2. translates to a condition on the ``structural sparsity" of the signal growing linearly with $n$;
%%. In the examples of Section \ref{sec:model} this means that
%e.g. the normalized sparsity $k/n$ should be constant independent on $n$. Conditions 3. and 4. are mild  conditions on the distributions of $\xx_0$ and $g$ that essentially guarantee that they comply with the weak law of large numbers. 

\vspace{5pt}
\noindent\emph{Estimation error:} Solving \eqref{eq:genLASSO} aims to recovering $\x_0^\mn$. 
We assess the quality of the estimator $\hat\x^\mn$ with the ``empirical squared error" (or simply, ``squared-error") defined as:
%\begin{align}\label{eq:NSE}
$\frac{1}{n}\|\hat\x^\mn-\x_0^\mn\|_2^2.$
%\end{align}
Note, that this is a random quantity owing to the randomness of $\A^\mn,\z^\mn$ and $\x_0^\mn$. Our main theorem precisely evaluates its high probability limit as $n\rightarrow\infty$.

\section{General Result}\label{sec:res}

\subsection{Key Assumption}
As already hinted in the introduction the functions $\loss$, $f$ and the distributions $p_\z$ and $p_{\x_0}$ determine the error performance indirectly through ``summary functionals" related to the Moreau-envelope approximations. The assumption below is an in-probability convergence requirement on the sequence of Moreau-envelopes, and defines those summary functionals. It also involves a rather natural  growth restriction on the loss function in the presence of noise to handle instances where the noise may have unbounded moments.

% with unbounded moments.

%The very first assumption is an in-probability convergence requirement on the sequence of Moreau-envelope functions of $\loss$ and $f$ when evaluated at

\vspace{5pt}
\begin{ass}[Summary functionals $L$ and $F$]
\label{ass:tech}
We say that Assumption \ref{ass:tech} holds if:
\begin{enumerate}[(a)]
\item For all $c\in\R$ and $\tau>0$, there exist continuous functions $L:\R\times\R_{>0}\rightarrow\R$ and $F:\R\times\R_{>0}\rightarrow\R$ such that\footnote{The convergence above is in probability over $\z\sim p_\z, \x_0\sim p_{\x_0}, \g\sim\Nn(0,\I_m)$ and $\h\sim\Nn(0,\I_n)$.
}
$$
\frac{1}{m}\left\{\env{\loss}{c\g+\z}{\tau}-\loss(\z)\right\}\rP \Lm{c}{\tau} \qquad\text{ and }\qquad\frac{1}{n}\left\{\env{f}{c\h + \x_0}{{\tau}}-f(\x_0)\right\}\rP\FFm{c}{\tau},
$$
%where $L:\R\times\R_{>0}$ and $F:\R\times\R_{>0}$ are continuous functions. 

\item At least one of the following holds. There exists constant $C>0$ such that $\frac{\|\z\|_2}{\sqrt{m}}\leq C$ with probability approaching 1 (w.p.a.1), or,  $\sup_{\vb\in\R^m}\sup_{\s\in\partial\loss(\vb)}\|\s\|_2<\infty$ for all $m\in\mathbb{N}$.
\end{enumerate}
\end{ass}

Assumption \ref{ass:tech} is rather mild: as discussed later in Section \ref{sec:Ass1_rem}, it holds naturally under very generic settings. Yet, it is of key importance since it defines the functionals $L$ and $F$, which are necessary ingredients involved in the error prediction of \eqref{eq:genLASSO}. The main theorem in its most general form will require some extra  (continuity and growth) properties on the functionals $L$ and $F$. Those will most often be naturally inherited from corresponding easy-to-verify and in cases well-studied properties of the Moreau envelope functions.

\subsection{Theorem}

Assumption \ref{ass:tech} provides us with the basic terminology needed for the statement of the main theorem. Technically, a few additional mild constraint qualifications are required. We present those immediately after the statement of the main result (see Assumption \ref{ass:prop}). The proof of the theorem is deferred to Appendix \ref{sec:proof}. An outline is given earlier in Section \ref{sec:outline}.

\begin{thm}[Master Theorem]\label{thm:master}
Let $\hat\x$ be a minimizer of the Generalized $M$-estimator in \eqref{eq:genLASSO} for fixed $\la>0$. Further let Assumptions \ref{ass:tech} and \ref{ass:prop} hold. If the following convex-concave minimax scalar optimization
\begin{align}\label{eq:AO_det_thm}
\inf_{\substack{\alpha\geq0 \\ \taug>0}}~\sup_{\substack{\beta\geq0 \\ \tauh> 0}}~~\Dc(\alpha,\taug,\beta,\tauh):= \frac{\beta\taug}{2} + 
\delta\cdot\Lm{\alpha}{\frac{\taug}{\beta}}
- \frac{\alpha\tauh}{2} - \frac{\alpha\beta^2}{2\tauh} + \la\cdot \FFm{\frac{\alpha\beta}{\tauh}}{\frac{\alpha\la}{\tauh}}.
\end{align}
has a unique minimizer $\alpha_*$,
%
% For the Deterministic Optimization in \eqref{eq:AO_det_thm} assume that it attains a \emph{unique} optimal point $\alpha_*\geq 0$ and that there exists $\beta_*\geq0$ such that $(\alpha_*,\beta_*)$ is optimal.
then, it holds in probability that
$$
\lim_{n\rightarrow\infty}\frac{1}{n}\|\hat\x-\x_0\|_2^2 = \alpha_*^2.  
$$
\end{thm}
We will often refer to the optimization problem in \eqref{eq:AO_det_thm} as the \emph{Scalar Performance Optimization} (SPO) problem.

A few important remarks are in place here (a detailed discussion follows in Section \ref{sec:det_disc}): (i) The convergence in the theorem is over the randomness of the design matrix $\A$, of the noise vector $\z$ and of the unknown signal $\x_0$. (ii) As was discussed in Section \ref{sec:model} the result applies to a properly defined sequence of M-Estimators of growing dimensions $m$ and $n$ such that $m/n\rightarrow\delta\in(0,\infty)$. (We have dropped the dependence of $\hat\x$ and $\x_0$ on $n$ to simplify notation.) (iii) The terms involving division by $\alpha$ and $\beta$ are understood as taking their limiting values when  $\alpha=0$ and $\beta=0$,  i.e.
$
\Dc(0,\taug,\beta,\tauh) = \lim_{\alpha\rightarrow 0^+}\Dc(\alpha,\taug,\beta,\tauh)$ and $ \Dc(\alpha,\taug,0,\tauh) = \lim_{\beta\rightarrow 0^+}\Dc(\alpha,\taug,\beta,\tauh).
$

\vp
%\subsection{An Additional Assumption}
Before proceeding with a further discussion of the result, let us state Assumption \ref{ass:prop} on the functionals $L$ and $F$ as required by Theorem \ref{thm:master}.

\begin{ass}[Properties of $L$ and $F$]\label{ass:prop}
We say that Assumption \ref{ass:prop} holds if all the following are true.
\begin{enumerate}[(a)]
\item $\lim_{\tau\rightarrow0^+}F(\tau,\tau) = 0$~~ and ~~$\lim_{c\rightarrow+\infty}\left\{\frac{c^2}{2\tau}-F(c,\tau)\right\}=+\infty$ for all $\tau>0$.
%
%\item For all $c\geq0,\tau>0$: $\frac{c^2}{2\tau}\geq F(c,\tau)$.  Furthermore, $\lim_{c\rightarrow\infty}\frac{c^2}{2\tau}-F(c,\tau)=+\infty$.
%%If $c>0$, then strict inequality holds. 
%
\item $\lim_{\tau\rightarrow0^+}{L(\alpha,\tau)}<+\infty$, $\lim_{\tau\rightarrow0^+}{L(0,\tau)}=0$, and , $-\infty<L_{2,+}(0,0):=\lim_{\tau\rightarrow0^+}L_{2,-}(0,\tau)\leq 0.$
\item $\frac{1}{m}\loss(\z)\rP L_0\in [0,\infty]$.
%If $\frac{1}{m}\loss(\z)$ converges in probability to some deterministic quactity, denote that by  $L_0$. Otherwise, set $L_0=+\infty$. Assume that 
Also, $L_0 = -\lim_{\tau\rightarrow+\infty} L(c,\tau)\geq -L(c,\tau^\prime)$ for all $c\in\R$,  $\tau^\prime>0$.
\item If $L_0=+\infty$, then $\lim_{\tau\rightarrow+\infty}\frac{L(c,\tau)}{\tau} = 0$, for all $c\in\R$.
%
%\item $-\infty<L_{2,+}(0,0):=
%
%%\item At least one of the following holds. There exists constant $C>0$ such that $\frac{\|\z\|_2}{\sqrt{m}}\leq C$ with probability approaching 1 (w.p.a.1), or/and,  $\sup_{\vb\in\R^m}\sup_{\s\in\partial\loss(\vb)}\|\s\|_2<\infty$ for all $m\in\mathbb{N}$.
%	\begin{itemize}
%		\item $\sup_{\vb\in\R^m}\|\partial\loss(\vb)\|_2<\infty$ for all $m\in\mathbb{N}$, 
%		\item . Also, for all $c>0$ there exists $C>0$, such that $\|\vb\|_2\leq c\sqrt{n} \implies \frac{1}{\sqrt{n}}\|\partial\loss(\vb)\|_2\leq C$ for all $n\in\mathbb{N}$.
%	\end{itemize}
%

\end{enumerate}
\end{ass}

A few remarks on the notation used in Assumption \ref{ass:prop} are as follows. In (b), $L_{2,-}(0,\tau)$ denotes the left derivative of $L$ with respect to its second argument evaluated at $(0,\tau)$. In (d), $L_0$ can take the value $+\infty$. For a sequence of  random variables $\{\mathcal{X}^{(n)}\}_{n\in\mathbb{N}}$, we write $\mathcal{X}^{(n)}\rP+\infty$, iff for all $M>0$, $\lim_{n\rightarrow\infty}\Pro\left(\mathcal{X}^{(n)}>M\right)=1$. 

\subsection{Separable M-estimators}
A special yet popular family of M-estimators involves separable loss/regularizer functions and iid noise/signal distributions. We refer to such instances as ``\emph{separable} M-estimators". To be concrete, consider solving
\begin{align}\label{eq:sep_inf}
\min_{\x}~ \sum_{j=1}^{m}\ell\left(\y_j-\ab_j^T \x\right)  + \la \sum_{i=1}^n f(\x_i),
\end{align}
where additionally, $\z_j\simiid p_z$ and ${\x_0}_i\simiid p_x.$
%
%\begin{align}\label{eq:sep_case}
%\loss(\vb) = \sum_{j=1}^{m}\ell(\vb_j),
%\quad
%\z_j\simiid p_z,
%\qquad\text{ and }\qquad
%f(\x) = \sum_{i=1}^{n}f(x_i)
%\quad
%{\x_0}_i\simiid p_x.
%\end{align}
Popular choices for the (scalar) loss function $\ell(v)$ above, include $v^2$, $|v|$ , Huber-loss, etc.. 
In the separable case, the generic Assumptions \ref{ass:tech} and \ref{ass:prop}  translate to very primitive and naturally interpretable conditions. Also, the functionals $L$ and $F$ take here an explicit form, which we call the ``{Expected Moreau envelope}". The \emph{Expected Moreau envelope} associated with the loss function is given by 
$$
L(c,\tau) = \Exp_{\substack{G\sim\Nn(0,1)\\Z\sim p_Z}}\left[ \env{\ell}{c G + Z}{\tau} - \ell(Z) \right].
$$ 
The function $L$, above, has the following remarkable properties: (i) it is smooth regardless of the smoothness of $\ell$, and, (ii) it is strictly convex regardless of whether $\ell$ is itself strictly convex or not. In particular, the second property can be used to show that the uniqueness condition of Theorem \ref{thm:master} regarding the minimizer $\alpha_*$ of \eqref{eq:AO_det_thm} is satisfied.

 In order to get a better understanding of those issues before discussing Theorem \ref{thm:master} in its most generality, we state below a summary of the main result regarding separable M-estimators. (The formal statement will be given later in Section \ref{sec:sep}, which includes a detailed treatment of separable M-estimators.) 
% slightly informal version of the main result for separable M-estimators. %\eqref{eq:sep_case} holds. 
%A detailed treatment with several interesting additional information is included 
%Section \ref{sec:sep}. 

\vp
\noindent\textbf{Summary of result for separable M-estimators }.~\emph{
 Let $\ell,f:\R\rightarrow\R$ be convex non-negative functions, and, $Z\sim p_Z$, $X_0\sim p_x$ such that for all $c\in\R$:
\begin{align}\label{eq:ass_main_inf}
\E\left[ 
%\sup_{s\in\partial\ell(c G + Z)}|s|^2
|\ell_+^\prime( c G + Z)|^2
\right]<\infty
\quad\text{ and }\quad
\E\left[ 
|f_+^\prime( c H + X_0)|^2
\right]<\infty.
\end{align}
Further assume $\E X_0^2<\infty$, and, that  either $\Exp Z^2<\infty$~ or ~
%$\sup_{v\in\R}|\ell_+^\prime(v)|
%%\sup_{\s\in\partial\ell(v)}|s|
%<\infty.$
$\sup_v\frac{|\ell(v)|}{|v|}<\infty$.
 Then, any minimizer $\hat\x$ of \eqref{eq:sep_inf} 
satisfies in probability, 
$$
\lim_{n\rightarrow\infty}\frac{1}{n}\|\hat\x-\x_0\|_2^2 = \alpha_*^2,
$$
where $\alpha_*$ is the unique minimizer to the (SPO) problem in \eqref{eq:AO_det_thm} with 
\begin{align*}%\label{eq:sep_L}
L(c,\tau) = \Exp\left[ \env{\ell}{c G + Z}{\tau} - \ell(Z) \right] \quad\text{ and }\quad F(c,\tau) = \Exp\left[ \env{f}{c H + X_0}{\tau} - \ell(X_0) \right].
\end{align*}
}

%\begin{thm}[Separable M-estimators (informal)]\label{thm:sep_inf}
% Let $\ell,f:\R\rightarrow\R$ be proper closed convex non-negative functions, and, $Z\sim p_Z$, $X_0\sim p_x$ such that for all $c\in\R$:
%\begin{align}\label{eq:ass_main_inf}
%\E\left[ 
%%\sup_{s\in\partial\ell(c G + Z)}|s|^2
%|\ell_+^\prime( c G + Z)|^2
%\right]<\infty
%\quad\text{ and }\quad
%\E\left[ 
%|f_+^\prime( c H + X_0)|^2
%\right]<\infty.
%\end{align}
%Further assume that $\ell(0)=f(0)=0$, $\E X_0^2<\infty$, and, that  either $\Exp Z^2<\infty$~ or ~
%%$\sup_{v\in\R}|\ell_+^\prime(v)|
%%%\sup_{\s\in\partial\ell(v)}|s|
%%<\infty.$
%$\sup_v\frac{|\ell(v)|}{|v|}<\infty$.
% Then, any minimizer $\hat\x$ of \eqref{eq:sep_inf} 
%satisfies with probability one, 
%$$
%\lim_{n\rightarrow\infty}\frac{1}{n}\|\hat\x-\x_0\|_2^2 = \alpha_*^2,
%$$
%where $\alpha_*$ is the unique minimizer to the (SPO) problem in \eqref{eq:AO_det_thm} with 
%\begin{align*}%\label{eq:sep_L}
%L(c,\tau) = \Exp\left[ \env{\ell}{c G + Z}{\tau} - \ell(Z) \right] \quad\text{ and }\quad F(c,\tau) = \Exp\left[ \env{f}{c H + X_0}{\tau} - \ell(X_0) \right].
%\end{align*}
%\end{thm}

We defer most of the discussions to Section \ref{sec:sep}. We only note here that there is \emph{no} smoothness or strict convexity assumption imposed on $\ell$ or $f$. Neither is the noise distribution required to have bounded moments. For example, $\ell(v)=|v|$ with $z$ distributed iid Cauchy satisfies all the conditions.  The main condition of the theorem is the one in \eqref{eq:ass_main_inf}, which is very primitive, and,  easy to check. It essentially guarantees that $\env{\ell}{c G + Z}{\tau} - \ell(Z)$ is absolutely integrable, thus $L$ is well-defined . It turns out that this also suffices for all requirements of Assumption \ref{ass:prop} to be satisfied.

\subsection{Remarks}\label{sec:det_disc}

\subsubsection{On Assumption \ref{ass:tech}}\label{sec:Ass1_rem}

We have made an effort to identify technical assumptions required for the statement of Theorem \ref{thm:master} which are as  \emph{generic} and \emph{minimal} as possible. Assumption \ref{ass:tech}
 summarizes  those technical conditions that are essential for our result to hold in its most general form. 
In later sections, when we discuss special cases (e.g. separable M-estimators in Section \ref{sec:sep}), we show that these conditions translate to more \emph{primitive} sufficient conditions that are often easier to check. 

\begin{remark}[WLLN and Robust Statistics]
The most natural setting where Assumption \ref{ass:tech}(a) can be easily interpreted is that of separable functions. For instance, if $\loss(\vb)=\sum_{j=1}^m{\ell(\vb_j)}$ and $\z_j\simiid p_Z(Z)$, then, in view of the WLLN, the natural candidate for $L(c,\tau)$ is $\Exp[\env{\ell}{cG+Z}{\tau}-\ell(Z)]$. Of course, this requires the argument under the expectation be aboslutely integrable. This is naturally satisfied for most loss functions in the case of noise distributions with bounded moments. On the other hand, when the noise is (say) heavy-tailed, some extra caution is required on the choice of the loss function; this leads to \eqref{eq:ass_main_inf}. As a warning to this discussion, Assumption \ref{ass:tech} does \emph{not} require separability. For example, we use  Theorem  \ref{thm:master}  to analyze the error performance of the square-root lasso (for which $\loss(\vb)=\|\vb\|_2$) in Section \ref{sec:sq_LASSO}, and, that of another instance with a non-separable regularizer function  in Section \ref{sec:cone}.
\end{remark}

\begin{remark}[Convexity of $L$ and $F$]\label{rem:LF_convex}
We remark that if Assumption \ref{ass:tech} holds, then both the functions $F$ and $L$ defined therein are \emph{jointly convex} in their arguments. This follows from the facts that (a) the Moureau envelope of a convex function is jointly convex in its arguments (cf. Lemma \ref{lem:Mconvex}(ii)), (b) taking limits preserves convexity. In that sense, the continuity requirement of the assumption on $L$ and $F$ is rather mild, since convex functions are continuous on the interior of their domain \cite[Thm.~10.1]{Roc70}.
\end{remark}

\begin{remark}[Robust Statistics]
Assumption \ref{ass:tech}(b) is tailored to scenarios in which the noise distribution has unbounded moments (e.g. mean, variance); in this case $\|\z\|_2/\sqrt{n}$ is not bounded with high probability. It is not hard to see that condition \ref{ass:tech}(b) implies $\sup_{\vb}\frac{\|\loss(\vb)\|_2}{\|\vb\|_2}<\infty$; such a requirement that $\loss$ grows at most linearly at infinity is natural in the context of robust
statistics. 
%Essentially this translates requirement translates to the Moreau-envelope function of $\loss$.
%Conditions (a) and (b) impose continuity requirements on $L$ and $F$.  Those are rather naturally inherited by corresponding well-known properties of the Moreau-envelope functions (e.g.,  Example \ref{} below).
\end{remark}

\subsubsection{On Assumption \ref{ass:prop}}

\begin{remark}[Continuity]
Conditions (a), (b) and (c) impose continuity and growth requirements on $L$ and $F$.  Those are rather naturally inherited by corresponding properties of the Moreau-envelope functions. In Appendix \ref{sec:M_prop} we have gathered such relevant and useful properties of Moreau-envelopes, which we use extensively throughout the text. For an illustration, it is not hard to see\footnote{Formally, this is a well-known continuity result on Moreau-envelopes. see Lemma \ref{lem:Mconvex}(ix)} that $\lim_{\tau\rightarrow0^+}\env{\loss}{\z}{\tau} = \loss(\z)$. This, of course is in line with Assumption \ref{ass:prop}(b) that $\lim_{\tau\rightarrow0^+}L(0,\tau)=0$. 
\end{remark}

\begin{remark}[Robust Statistics]
Assumption \ref{ass:prop}(d) is meant to deal with cases of noise with unbounded moments (this will often translate to $L_0=+\infty$). In such cases, we   require that $L(c,\tau)$ grows sub-linearly in $\tau$. Once more, this property is essentially inherited without any extra effort  by corresponding property of the Moreau-envelope.
\end{remark}

\subsubsection{On the theorem}

\begin{remark}[Limits]
In evaluating the objective function $\Dc$ of the (SPO) at $\alpha=0$ and $\beta=0$, Assumptions \ref{ass:prop}(a)-(b) turn out to be useful, giving
$$
\lim_{\beta\rightarrow 0^+} \Lm{\alpha}{\frac{\taug}{\beta}} = -L_0 \quad\text{ and }\quad\lim_{\alpha\rightarrow 0^+} \FFm{\frac{\alpha\beta}{\tauh}}{\frac{\alpha\la}{\tauh}} = 0.
$$
\end{remark}

\begin{remark}[Convexity]
An important property of the (SPO) is that it is \emph{convex}: its objective function $\Dcall$ is (jointly) convex in $\alpha,\taug$ and concave in $\beta,\tauh$.  As is well known, convexity translates to the ability to efficiently solve the optimization; see also Remark \ref{rem:alt_cha} below.
\end{remark}

\begin{remark}[Uniqueness of $\alpha_*$]
Theorem \ref{thm:master} assumes that the (SPO) problem has a unique minimizer $\alpha_*$. In most cases discussed in this paper, the uniqueness property is a consequence of the fact that the function $L(c,\tau)$ turns out to be (jointly) \emph{stricly} convex in its arguments.  
%Then, it can be shown that the objective of the (DO) is itself \emph{stricly} convex-concave.  See Example \ref{ex:uni} for an initial illustration.
In the separable case, this translates to the strict convexity of the expected Moreau envelope function $\E[\env{\ell}{cG+Z}{\tau}-\ell(Z)]$, cf Remark \ref{rem:uni}.
% We show this to hold for a wide class of loss functions and noise distributions in Section \ref{sec:}.
%Later in Lemma \ref{lem:strict}, we will see that for a non-differentiable, continuous, convex function $\psi:\R\rightarrow\R$ and a continuous (not everywhere zero) density $p$, the function $\E_{\substack{X\sim p\\ G\sim\Nn(0,1)}}[\env{\psi}{cG+X}{\tau}]$ is jointly strictly convex in $c$ and $\tau$.
\end{remark}

\subsubsection{Further Discussions}

%~~~
%\vp\noindent{\textbf{Generality and assumptions}.~ 
%
%\begin{remark}The measurement matrix $\A$ is drawn from a Gaussian distribution of zero mean and of appropriately normalized variance $1/n$. The assumptions on the distributions of $\z$ and $\x_0$ are pretty generic and are implicit in the two assumptions of Section \ref{sec:ass}. The assumptions also imposes implicit conditions on the loss function $\loss$ and on the regularizer $f$, other than the standard convexity and continuity assumption. However, such conditions are rather mild. For an illustration of the generality of the assumptions, it was shown in Examples \ref{ex:tech} and \ref{ex:uni} that they hold for the $\ell_1$-regularized LAD with iid Cauchy-distributed noise and a sparse-gaussian model on $\x_0$. Besides, our result is general enough to capture other interesting instances of M-estimators that go beyond separable functions, e.g. the square-root lasso with nuclear-norm regularization. We explore such special cases in Section \ref{sec:Examples}. 
%The generality of Assumption \ref{ass:tech} serves to identify the essential ingredients required in the proof of Theorem \ref{thm:master}.

%~~~
%\vp\noindent{\textbf{The role of the parameters}.~
\begin{remark}[The role of the parameters]
The role of the normalized number of measurement $m/n\rightarrow\delta$ and that of the regularizer parameter $\la$ are explicit in \eqref{eq:AO_det_thm}. On the other hand, the structure of $\x_0$ and the choice of the regularizer $f$ are implicit through $F$. Similarly, any prior knowledge on the noise vector $\z$ and the effect of the loss function $\loss$ are also implicit in \eqref{eq:AO_det_thm} through $L$. In the separable case, the role of those summary parameters is played by the Expected Moreau envelope function.
\end{remark}

% Later we will have to work a little to verify the assumptions in special cases,
%but the general result serves to identify the essential ingredients in the proof.

%~~~
\begin{remark}[An alternative characterization]\label{rem:alt_cha}
The (SPO) problem in \eqref{eq:AO_det_thm} is \emph{convex-concave} and only involves four \emph{scalar} variables. Thus, the optimal $\alpha_*$ can, in principle, be efficiently numerically computed. Equivalently, $\alpha_*$ can be expressed as the solution to the corresponding first-order optimality conditions, which offers an alternative to the current statement of Theorem \ref{thm:master}. In Section \ref{sec:sys_eq} we explicitly derive the system of stationary equations for the case of separable M-estimators. It is often possible to solve the stationary equations by means of simple iterative schemes (cf. Remark \ref{rem:num}). Furthermore, this alternative formulation might be easier to work with when deriving analytic properties of  $\alpha_*$. As an example, in Sections \ref{sec:no_reg}--\ref{sec:cone} for specific instances of M-estimators,  we start from the stationary equations, combine them in an appropriate way, and, derive insightful and practically useful properties, such as lower bounds  on $\alpha_*$, necessary conditions on the problem parameters such that $\alpha_*$ (correspondingly, the equated error) be bounded, etc..

\end{remark}

%For instance, when $\loss$ and/or $f$ are separable, it turns out that $L$ and $F$ are always continuously differentiable irrespective of whether the former are differentiable or not. This property is tied to the well-established and important smoothness properties of the Moureau envelopes of convex functions (e.g. \cite[]{RocVar}). Please refer to Section \ref{sec:separable} for a further detailed discussion.

%In the popular scenario, where these functions are separable, we can further massage the equations in \eqref{eq:stationary} and obtain simpler characterizations in Section \ref{sec:separable}. 

%~~~
\begin{remark}[Optimal cost]
Although not stated as part of our main result, the analysis that leads to Theorem \ref{thm:master} further characterizes the limiting behavior of the optimal cost, say $\cost^\mn$, of the M-Estimator in \eqref{eq:genLASSO}.  Let $\mathrm{\Gamma}_*$  be the optimal cost of the (SPO), then
% and assume that  $\frac{1}{n} f(\x_0)\rP F_{\x_0}$. Then, it holds with probability one,
\begin{align}\label{eq:opt_cost}
\frac{1}{n} \min_{\x}\left\{ \loss(\y-\A\x)-\loss(\z) + \la( f(\x)-f(\x_0) ) \right\} \rP {\Gamma}_* 
\end{align}
\end{remark}
%~~~
\begin{remark}[Asymptotics]
The statement of the theorem holds under an asymptotic setup in which the problem dimensions $m$ and $n$ grow to infinity. In Section \ref{sec:sim} we examine via simulations the validity of the prediction for finite values of $m$ and $n$. The results indicate that the asymptotic prediction becomes accurate for values of the problem parameters ranging on a few hundreds, and, in cases even on a few tens. 
\end{remark} 
%~~~

%~~~
\begin{remark}[Beyond Gaussian designs]~ Theorem \ref{thm:master} assumes that the entries of the design matrix $\A$ are iid Gaussian. In the proof of the result this assumption is crucial since the proof itself heavily relies on the CGMT, for which the gaussianity assumption is implicit. Yet, a few important remarks apply regarding the potential use of the results and the analysis of this work to cases beyond the gaussian design. Some examples include the cases of Elliptical Distributions \cite{karoui13} and Isotropically Random Orthogonal matrices to which the CGMT framework is still applicable. We discuss those in Section \ref{sec:conc}.
\end{remark}

\begin{remark}[Proof]\label{rem:CGMT}
%\vp\noindent{\textbf{Convex Gaussian Min-max Theorem}.~
The fundamental tool behind  our analysis is the \emph{Convex Gaussian Min-max Theorem (CGMT)}. The CGMT is a tight and extended version of a Gaussian comparison inequality due to Gordon \cite{GorLem}\footnote{Gordon's original result is often referred to as the Gaussian Min-max Theorem (GMT). It is classically used to establish non-asymptotic probabilistic lower bounds on the minimum singular value of Gaussian matrices (e.g. \cite{Ver}), and has a number of other applications in high-dimensional convex geometry (e.g. \cite{GorLem,Ledoux,Milman}). More recently,
 Vershynin and Rudelson  introduced  the idea of using the GMT (more specifically, a corollary of it known as the ``escape through the mesh" Lemma of \cite{GorLem}) to study the phase transition of $\ell_1$-minimization in compressed sensing\cite{rudelson2006sparse}. This idea was refined, clarified and extended to
general settings in the papers \cite{Cha,TroppEdge,stojnic2013upper}. See Section \ref{sec:prior} for details. 
% 
% 
% Stojnic \cite{Sto} was able to tighten the analysis of \cite{rudelson2006sparse}, and, derived thresholds that were shown to be exact through simulations. \cite{Cha} generalized Stojnic's result to arbitrary structure inducing norms. The thresholds obtained via the GMT analysis were recently shown to be asymptotically exact by Amelunxen et. al. \cite{TroppEdge}, and independently, by Stojnic \cite{stojnic2013upper}. 
  }, under additional convexity assumptions that arise in many practical applications. 
It associates with a primary optimization (PO) problem a simplified auxiliary optimization
(AO) problem from which we can tightly infer properties of the original (PO), such as the
optimal cost, the optimal solution, etc.. We manage to write the general $M$-estimator in \eqref{eq:genLASSO} as a (PO) problem so that CGMT is applicable. This leads to a corresponding (AO) problem. Next, we analyze the error of the (AO) and translate the result to the (PO) thanks to the CGMT. These ideas form the basic mechanics of the proof and are rather simple to explain; see Section \ref{sec:outline} for an outline.
The idea of combining Gordon's original result with convexity is attributed to Stojnic \cite{StoLASSO,stojnic2013meshes,stojnic2013upper}. Thrampoulidis, Oymak and Hassibi built and significantly extended on this idea arriving at the CGMT as it appears in \cite[Thm.~3]{COLT15}.
%\footnote{We refer the interested reader to \cite{COLT15} for further discussions on the CGMT, and a more detailed comparison/connections with the previous works of Gordon and of Stojnic.}. 
The version of the CGMT presented here in Theorem \ref{thm:CGMT} includes a further generalization which can significantly extend the scope and applicability of the Theorem beyond the squared error analysis of $M$-estimators (see Section \ref{sec:conc}).
%We repeat a statement of the CGMT here. This is done in part for the ease of reference of the reader, but primarily because the version that we present here includes a generalization that can significantly extend the scope and usefulness of the Theorem beyond the error analysis of $M$-estimators (e.g. \cite{upcoming}). In the next section we show how the CGMT relates to our specific problem of interest.
\end{remark}

\begin{remark}[Why ``Master"?]
All existing results in the literature on the performance of specific instances of M-estimators can be seen as special cases of Theorem \ref{thm:master}. Beyond those, the theorem can be used to derive a wide range of novel results, including instances where the loss function and the regularizer may be non-smooth and non-separable, and where, the noise distribution may have unbounded moments. We discuss several examples in Section \ref{sec:sim}.
\end{remark}

%~~~
%\subsubsection{Premises/Opportunities}

\begin{remark}[Premises/Opportunities]
%\vp\noindent{\textbf{Premises/Opportunities}.~
 Theorem \ref{thm:master} paves the way to answering optimality questions regarding the performance of M-estimators under different scenarios.
%  We might ask:
%\begin{itemize}
%\item[-] what is the optimal loss function and regularizer, under different settings, e.g. in the presence of outliers, particular structure of $\x_0$, etc.?
%\item[-] what is the minimum achievable squared error in each one of those scenarios? Do there exist consistent M-estimators, i.e. instances for which $\|\hat\x-\x_0\|_2^2/n\rP0$?
%\item[-] how to optimally tune the regularizer parameter $\lambda$?
%\item[-] how does the sampling ratio $\delta=m/n$ affect the error?
%%\item 
%\end{itemize}
The first fundamental step in answering such optimality questions (see Section \ref{sec:contr} ) is characterizing the squared error in terms of the problem design parameters, i.e. $f,\ell,\la$ and $\delta$. And, of course, this is exactly what Theorem \ref{thm:master} achieves. Since the characterization differs from the corresponding results of classical statistics (where $n$ is considered fixed), the questions will not in general admit the same answers. In the high-dimensional regime, our knowledge on those issues is rather limited and there is an exciting potential for exploring new phenomena  and providing answers that are both of theoretical and of practical interest. We  provide a few preliminary results towards this direction in Section \ref{sec:sim}. %including the following: (i) in Section \ref{sec:no_reg} we show that in the absence of regularization any loss function is inconsistent, (ii) an example of an M-estimator which is consistent for the right choice of $\la$ is explored in Section \ref{sec:cone}, (iii) for a specific instant, termed cone-constrained M-estimators, we derive in Section \ref{sec:cone} necessary conditions on the number of measurements that guarantee bounded error performance.

\end{remark}
% and, regularizer parameter $\lambda$
%
% A few interesting examples might include the following:
%\begin{itemize}
%\item 
%\item How to optimally choose the loss function?
%\end{itemize}

%Finally, if $\cost$ denotes the optimal cost of the optimization in \eqref{eq:genLASSO}, then
%$$
%\lim_{n\rightarrow\infty}\frac{1}{n}\cost = \alpha_*^2.  
%$$

%\begin{remark}
%\end{remark}

%\subsection{Proof Outline}
\section{Separable M-estimators}\label{sec:sep}

%\subsection{Setup}

We specialize the general result of Section \ref{sec:res} to the popular case where the loss function $\loss$ and the regularizer $f$ are both separable, and, the noise vector and signal $\x_0$ both have entries iid. To make things concrete, assume\footnote{Note the slight abuse of notation here in using $f$ to denote both the vector-valued and scalar regularizer function.} 
$$\loss(\vb) = \sum_{j=1}^{m}\ell(\vb_j)
\qquad\text{ and }\qquad
\z_j\simiid p_Z,~ j=1,\ldots, m.
$$
$$
f(\x) = \sum_{i=1}^{n}f(x_i)
\qquad\text{ and }\qquad
{\x_0}_i\simiid p_x,~ i=1,\ldots, n.
$$
Henceforth, both $\ell$ and $f$ are proper closed convex functions. Also, it is further assumed 
\begin{align}
\ell(0)=0=\min_v\ell(v) \quad\text{ and }\quad f(0)=0.
\end{align}

%\chris{We might need to say that $\mathrm{dom}\ell=\R$}

%
%We consider the popular case where the loss function $\loss$ and the regularizer $f$ are separable over their arguments, i.e. 
%$$\loss(\vb) = \sum_{j=1}^{m}\ell(\vb_j)\qquad\text{ and }\qquad\f(\x) = \sum_{i=1}^{n}f(\x_i),$$ for  proper, continuous, convex functions $\ell:\R\rightarrow\R$ and $f:\R\rightarrow\R$.  Furthermore, we let the entries of $\z$ and $\x_0$ be iid, and write $p_z$ and $p_x$ for their probability densities, i.e.
%$$
%\z_j\sim p_z,~ j=1,\ldots, m\qquad\text{ and }\qquad \x_{0,i}\sim p_x,~ i=1,\ldots, n.
%$$
%Onwards, we assume the above and specialize Theorem \ref{thm:master} to this  setting.

\subsection{Satisfying Assumptions \ref{ass:tech} and \ref{ass:prop}}\label{sec:ass_sep}
%First, we need to satisfy the two generic Assumption \ref{ass:tech} and \ref{ass:uni}, namely the pointwise convergence and the uniqueness of solution of the (DO).

To apply Theorem \ref{thm:master}, we first need to verify that Assumptions \ref{ass:tech} and \ref{ass:prop} hold for both the loss function and the noise distribution, and, for the regularizer and the signal distribution.

\subsubsection{Loss function and Noise Distribution}
In the separable case Assumptions \ref{ass:tech} and \ref{ass:prop} essentially translate to the following requirement on $\ell$ and $p_Z$:
\begin{align}\label{eq:ass_main}
\E\left[ 
%\sup_{s\in\partial\ell(c G + Z)}|s|^2
|\ell_+^\prime( c G + Z)|^2
\right]<\infty,\quad \text{ for all } c\in\R.
%\E\left[|\ell^\prime(\alpha G+Z)|^2\right]<\infty, \quad\text{where} \ell'(v) = \sup_{s\in\partial\ell(v)}s
\end{align}
where the expectation is over $Z\sim p_Z$ and $G\sim\Nn(0,1)$. This is shown in Lemma \ref{lem:ass_sep} below.

\begin{lem}[Expected Moreau envelope--Loss fcn]\label{lem:ass_sep}
If $\ell$ and $p_Z$ satisfy \eqref{eq:ass_main}, then, Assumptions \ref{ass:tech}(a) and  \ref{ass:prop}(b)-(d) hold with
\begin{align}\label{eq:sep_L}
L(c,\tau) = \Exp\left[ \env{\ell}{c G + Z}{\tau} - \ell(Z) \right].
\end{align}
\end{lem}

The condition in \eqref{eq:ass_main} is very primitive and is, in general, easy to check. It essentially guarantees that $\env{\ell}{c G + Z}{\tau} - \ell(Z)$ is absolutely integrable (for a proof see Appendix \ref{app:ass_sep}). Hence, $L$ in Lemma \ref{lem:ass_sep} is well-defined and it satisfies Assumption \ref{ass:tech}(a) as a result of applying the WLLN.  A few examples for which \eqref{eq:ass_main} is satisfied include:
\begin{enumerate}
\item $\ell(v)=v^2$ and $\E Z^2<\infty$,
\item \eqref{eq:ass_main} is trivially
%\footnote{this it to be compared to the work in Example \ref{ex:tech}.} 
satisfied if $\ell(v)=|v|$ for any noise distribution $p_Z$,
\item Huber-loss and $Z\sim \mathrm{Cauchy}(0,1)$. %\chris{is it?}.
\end{enumerate}

Apart from \eqref{eq:ass_main}, we also need to satisfy Assumption \ref{ass:tech}(b), which here translates to the following requirement:
\begin{align}\label{eq:ass_main1}
\Exp Z^2<\infty \qquad \text{{ or }} \qquad \sup_{v\in\R}|\ell_+^\prime(v)|
%\sup_{\s\in\partial\ell(v)}|s|
<\infty.
\end{align}
The second condition above on boundedness of the sub-differential is equivalent to $\sup_{v}\frac{|\ell(v)|}{|v|}<\infty$. That is, if $Z$ has unbounded second moments then $\ell$ needs to grow to infinity at most linearly, e.g. $|\cdot|$, Huber-loss, etc.

%The rest, i.e. Assumptions \ref{ass:prop}(b)-(c) are  inherited by corresponding properties of the Moreau envelope.

\subsubsection{Regularizer and Signal Distribution}
Not surprisingly, following the results of Section  \ref{sec:ass_sep}, the required condition on $f$ and $p_x$ becomes
\begin{align}\label{eq:ass_main2}
\E\left[ 
%\sup_{s\in\partial f(c H + X_0)}|s|^2
|f_+^\prime( c H + X_0)|^2
\right]<\infty,\quad \text{ for all } c\in\R.
%\E\left[|\ell^\prime(\alpha G+Z)|^2\right]<\infty, \quad\text{where} \ell'(v) = \sup_{s\in\partial\ell(v)}s
\end{align}
where the expectation is over $X_0\sim p_x$ and $H\sim\Nn(0,1)$.  Additionally,   the following mild assumptions are required:
 \begin{align}\label{eq:ass_main21}
\exists~ x_+>0,~x_-<0~ \text{ such that } ~0\leq f(x_{\pm})<\infty \qquad\text{ and }\qquad \E X_0^2 <\infty.
 \end{align}

%Lemma \ref{lem:ass_sep2} below summarizes the result.

\begin{lem}[Expected Moreau Envelope--Regularizer fcn]\label{lem:ass_sep2}
If $f$ and $p_x$ satisfy \eqref{eq:ass_main2} and \eqref{eq:ass_main21}, then, Assumptions \ref{ass:tech}(a) and \ref{ass:prop}(a) hold with
\begin{align}\label{eq:sep_F}
F(c,\tau) = \Exp\left[ \env{f}{c H + X_0}{\tau} - f(X_0) \right].
\end{align}
\end{lem}

\subsection{The Expected Moreau Envelope}\label{sec:EME}

%\subsection{Strict Convexity of the Expected Moreau Envelope}
If conditions \eqref{eq:ass_main}, \eqref{eq:ass_main1} and \eqref{eq:ass_main2} are satisfied, then Theorem \ref{thm:master} is applicable with $L$ and $F$ given as in \eqref{eq:sep_L} and \eqref{eq:sep_F}, respectively. We call those functions, the \emph{Expected Moreau envelopes}. The important role they play in determining the error performance of the corresponding M-estimator is apparent from Theorem \ref{thm:master}. In this section, we discuss two key features that they possess, namely, \emph{smoothness} and \emph{strict convexity}. 

\begin{lem}[Smoothness]\label{lem:smooth_ell}
Suppose $\ell$ is a closed proper convex function and $p_Z$ a noise density such that   \eqref{eq:ass_main} holds.
Then, the function $
L(c,\tau):=
\E\left[\env{\ell}{c G+ Z}{\tau} - \ell(Z)\right]
$
%Let Assumption \ref{ass:sep}(i) hold and consider the function $L:\R\times\R_{>0}\rightarrow\R$, 
%%\begin{align*}
is differentiable in $\R\times\R_{>0}$ with 
$$
\frac{\partial{L}}{{\partial}{c}} = \E\left[ \mathrm{e}_\ell^\prime\left(c G + Z;{\tau}\right) G \right] \quad\text{ and }\quad\frac{\partial{L}}{{\partial}{\tau}} = -\frac{1}{2}\E\left[ \left(\mathrm{e}_\ell^\prime\left(c G + Z;{\tau}\right) \right)^2 \right].
$$
\end{lem}

\begin{remark}
Note that $L$ is smooth, regardless of any non-smoothness of $\ell$. This is a well-known fact about Moreau envelope approximations, and also, one of the primal reasons behind the important role those functions play  in convex analysis \cite{RocVar}.  The property is naturally inherited to the \emph{Expected} Moreau envelopes as revealed by the lemma above.
\end{remark}

\begin{lem}[Strict Convexity]
\label{lem:strict_ell}
Suppose $\ell$ is a closed proper convex function and $p_Z$ a noise density such that   \eqref{eq:ass_main} holds and the following are satisfied:
\begin{enumerate}[(a)]
\item Either there exists $x\in\R$ at which $\ell$ is \emph{not} differentiable, or, there exists interval $\Ic\subset\R$ where $\ell$ is differentiable with a strictly increasing derivative,
\item $\mathrm{Var}(Z)\neq0$ \footnote{We require that there exist at least two values of $z\in\R$ for which $p_Z(z)>0$. In particular, there is \emph{no} requirement that $\mathrm{Var}(Z)$ be defined, e.g. Cauchy distribution is allowed.}, and, at each $z\in\R$, $p_Z(z)$ is either a Dirac delta function or it is continuous.
\end{enumerate}
Then,  $
L(c,\tau):=
\E\left[\env{\ell}{c G+ Z}{\tau} - \ell(Z)\right]
$ is jointly strictly convex in $\R_{>0}\times\R_{>0}$.
\end{lem}

%\begin{ass}[Separable Functions]\label{ass:sep}
%Let $\ell,f:\R\rightarrow\R$ be convex, continuous functions; $G,H\sim\Nn(0,1)$, $Z\sim p_z$ and $X_0\sim p_x$. We say that Assumption \ref{ass:sep} holds, if:
%\begin{enumerate}[(i)]
%%\item for all $c\in\R$,  \chris{Fix those}
%%$$\E\left[| \ell(cG+Z) - \ell(Z) |\right]<\infty \qquad \text{ and } \qquad \E\left[| f(cH+X_0) - f(X_0) |\right]<\infty,$$
%%\item for all $c\in\R$ and $\tau>0$, 
%%$$\E\left[| (cG+Z) - \prox{\ell}{cG+Z}{\tau}|^2\right]<\infty \qquad \text{ and } \qquad \E\left[| (cH+X_0) - \prox{f}{cH+X_0}{\tau} |^2\right]<\infty,$$
%\item each one of the functions $\ell$ and $f$ is such that at least one of the following holds:
%	\begin{enumerate}[(a)]
%		\item there exists $\x\in\R$ at which the function is \emph{not} differentiable,
%		\item there exists interval $\Ic\subset\R$ where the function is differentiable with a strictly increasing derivative,
%		% they are strictly convex in some ,
%	\end{enumerate}
%\item each one of the densities $p_z$ and $p_x$ of $Z$ and $X_0$, respectively, are such that both the following hold:
%	\begin{enumerate}[(a)]
%%		\item the density takes at least two strictly positive values in $\R$,
%		\item $Var(Z) > 0$.
%		\item at every $x\in\R$ the density function is either a Dirac delta function or it  is continuous.
%		% they are strictly convex in some ,
%	\end{enumerate}
%\end{enumerate}
%%Then, we say that Assumption \ref{ass:sep} holds.
%\end{ass}

\begin{remark}
The function $L$ is strictly convex, without requiring any strong or strict convexity assumption on $\ell$. Interestingly, this property is not in general true for Moreau envelope approximations, but, it turns out to be the case for the \emph{Expected} Moreau envelope $L$. The fact that the latter further involves taking an expectation over $cG+Z$, with $G$ having a nonzero density on the entire real line, turns out to be critical. 
\end{remark}

\begin{remark}[Strict convexity$\implies$Uniqueness of $\alpha_*$]\label{rem:uni}
The strict convexity property of $L$ is critical because it guarantees uniqueness of the minimizer $\alpha_*$ of the (SPO) problem in Theorem \ref{thm:master}. This implication is proved in Lemma \ref{lem:uni} in Appendix \ref{app:uni}.
\end{remark}

\subsection{Error Prediction}
We are now ready to state the main result of this section which characterizes the squared error of separable M-estimators. This is essentially a corollary  of Theorem \ref{thm:master}.

\begin{thm}[Separable M-estimators]\label{thm:sep}
Suppose $\ell$ and $p_Z$ satisfy \eqref{eq:ass_main}, \eqref{eq:ass_main1}, and, the two conditions of Lemma \ref{lem:strict_ell}. Further assume that $f,p_x$ satisfy \eqref{eq:ass_main2} and \eqref{eq:ass_main21}. Let $\hat\x$ be any minimizer of the separable M-estimator and consider the (SPO) problem in \eqref{eq:AO_det_thm} with $L$ and $F$ given as in \eqref{eq:sep_L} and \eqref{eq:sep_F}, respectively. If the set of minimizers of the (SPO) over $\alpha$ is bounded, then there is a unique such minimizer $\alpha_*$ for which it holds in probability that
$$
\lim_{n\rightarrow\infty}\frac{1}{n}\|\hat\x-\x_0\|_2^2 = \alpha_*^2.
$$
%%where $\alpha_*$ is the unique minimizer to the optimization in \eqref{eq:AO_det} with 
%%\begin{align*}%\label{eq:sep_L}
%%L(c,\tau) = \Exp\left[ \env{\ell}{c G + Z}{\tau} - \ell(Z) \right] \quad\text{ and }\quad F(c,\tau) = \Exp\left[ \env{f}{c H + X_0}{\tau} - \ell(X_0) \right].
%\end{align*}
\end{thm}

\vp
\begin{remark}[Boundedness]
Applying Theorem \ref{thm:sep} requires a few primitive and easy to check assumptions on $\ell,Z,f$ and $X_0$. In contrast to the general case in Theorem \ref{thm:master}, here, the uniqueness of $\alpha_*$ is guaranteed if the set of minimizers of the (SPO) over $\alpha$ is bounded.  The boundedness condition is essentially in one to one correspondence with the squared error of the M-estimator being (stochastically) bounded or not. We expect the boundedness assumption, which is generic in nature, to translate to necessary and sufficient primitive conditions on $\ell,f, p_Z, p_x$ and $\delta$. For example, in Remark \ref{rem:stable_f=0} we show that in the case of un-regularized M-estimators, a necessary such condition is that the normalized number of measurements be larger than 1, i.e. $\delta>1$ 
\footnote{ Besides, in Remark \ref{rem:stable_cone}, we show that with appropriate regularization, the necessary condition on the number of measurements becomes $\delta>\overline{D}_{f,\x_0}$, where $\overline{D}_{f,\x_0}$ is a (normalized) summary functional of $f$ and $\x_0$, which is geometric in nature and can in general be strictly less than one. In particular, this means that with an appropriate regularizer the signal $\x_0$ can be robustly estimated with a number of measurements that is less than the dimension of the signal. Of course, this  is one of the fundamental results in the compressive sensing literature. In particular, $\overline{D}_{f,\x_0}$ coincides with the phase-transition threshold of noiseless compressed sensing \cite{Cha,TroppEdge}. }. 
Identifying such conditions  that would guarantee bounded error is an important design issue, since it provides guarantees and guidelines on how the loss function, the regularizer and the number of measurements ought to be chosen. In the general case, this remains an open question. We expect that Theorem \ref{thm:sep} itself and the proof ideas behind it (in particular, see Lemma \ref{lem:convergence}(b)) can be used  to answer this question. Since this is not the main focus of the paper, we leave the rest for future work.
%
%It remains an open subject to characterize necessary and sufficient conditions that would guarantee bounded error in the general case. This is one of the 
%
% Please also refer to Remark \ref{} for a more involved example in the presence of set indicator-type of regularizer. Identifying necessary/sufficient such conditions for the general case, is a   
%
%Assuming boundedness of $\alpha_*$, it is possible to identify necessary primitive conditions 
\end{remark}

\subsubsection{As a system of nonlinear equations}\label{sec:sys_eq}
Theorem \ref{thm:sep} predicts the error of the M-estimator as the optimizer $\alpha_*$ to a convex-concave optimization problem with four optimization variables. Equivalently, $\alpha_*$ can be expressed via the first-order optimality conditions (stationary equations) corresponding to this optimization. Recall from Lemma \ref{lem:smooth_ell} that $L$ and $F$ are differentiable (irrespective of smoothness of $\ell$ and $f$). The error of the M-estimator is then the unique $\alpha_*\geq 0$ for which there exist $\taug_*\geq0,\beta_*\geq 0$ and $\tauh_*\geq0$ satisfying
\begin{align}\label{eq:4eq_0}
\frac{\partial\Dc}{\partial\alpha}\Big|_{p_*}(\alpha-\alpha_*)\geq 0, \qquad \frac{\partial\Dc}{\partial\taug}\Big|_{p_*}(\taug-\taug_*)\geq 0, \qquad
\frac{\partial\Dc}{\partial\beta}\Big|_{p_*}(\beta-\beta_*)\leq 0, \qquad \frac{\partial\Dc}{\partial\tauh}\Big|_{p_*}(\tauh-\tauh_*)\leq 0,
\end{align}
for all $\alpha,\beta\geq 0,\taug,\tauh>0$ and $p_*=(\alpha_*,\taug_*,\beta_*,\tauh_*)$.
A similar remark as the one that follows Theorem \ref{thm:master} is in place regarding the values $\alpha=0$ and $\beta=0$. At these, the derivatives above should be interpreted as the corresponding (upper) limits as $\alpha\rightarrow0^+$ and $\beta\rightarrow0^+$. The continuity properties of the Moreau envelope (see Lemma \ref{lem:Mconvex}) guarantee that those limits are well-defined
%
%Note that $\env{\ell}{\chi}{\tau}$ and $\env{f}{\chi}{\tau}$ are continuously differentiable with respect to $\chi$, irrespective of nonsmothness of $\ell$ and $f$, respectively (cf. Lemma \ref{lem:Mconvex}). Denote,
%$$
%\mathrm{e}_\ell^\prime\left({\chi};{\tau}\right) = \frac{\partial{\env{\ell}{\chi}{\tau}}}{\partial\chi}\qquad\text{ and }\qquad 
%\mathrm{e}_f^\prime\left({\chi};{\tau}\right) = \frac{\partial{\env{f}{\chi}{\tau}}}{\partial\chi}.
%$$

When $\alpha_*>0$ and there also exist optimal values $\beta,\taug,\tauh$, all of them strictly positive, then \eqref{eq:4eq_0} holds with equalities. In this case, a little bit of algebra, and, an appropriate change of variables from $\taug,\tauh$ to $\kappa,\nu$, shows that the optimality conditions can be expressed as follows:
\begin{equation}\label{eq:4eq}
 \left\{
 \begin{aligned}
        \al^2&=\Exp\left[\left(\frac{\la}{\nu} \cdot{\mathrm{e}^\prime_f}\left(\frac{\beta}{\nu}H + X_0;\frac{\la}{\nu}\right)-\frac{\beta}{\nu}H\right)^2\right],\\
        \beta^2&=\delta\cdot\Exp\left[ \left(\mathrm{e}_\ell^\prime(\GZ,\kappa)\right)^2\right],\\
        \nu\al &= \delta\cdot \Exp\left[ \mathrm{e}_\ell^\prime(\GZ,\kappa) \cdot G\right],\\
        \kappa\beta &= \frac{\beta}{\nu} - \frac{\la}{\nu}\cdot \Exp\left[ \mathrm{e}_f^\prime\left(\frac{\beta}{\nu}H + X_0;\frac{\la}{\nu}\right) \cdot H\right].
       \end{aligned}
 \right.
 \end{equation}
Here, $e_f^\prime$ and $e_\ell^\prime$, denote the first derivatives of the Moureau envelopes with respect to their first argument. 

\subsubsection{Remarks}
%\begin{remark}[Reformulations] 
\begin{remark}[Reformulations]
%\vp\noindent\textbf{Reformulations}:
The system of equations in \eqref{eq:4eq} can be easily reformulated in terms of the proximal operator of $f$ and $\ell$, using
$$
\mathrm{e}_\ell^\prime(\chi,\tau) = \frac{1}{\tau}(\chi-\prox{\ell}{\chi}{\tau}),
$$
and similar for $f$ (see Lemma \ref{lem:Mconvex}(iii)). In the case of additional smoothness assumptions on the loss function and/or the regularizer, further reformulations are possible. For example, if $\ell$ is two times differentiable, then using Stein's formula for Normal random variables we can make the following substitution in \eqref{eq:4eq}:
\begin{align}\label{eq:Stein's}
\Exp\left[ \mathrm{e}_\ell^\prime(\GZ,\kappa) \cdot G\right] = \alpha\cdot\Exp\left[ \mathrm{e}_\ell^{\prime\prime}(\GZ,\kappa)\right],
\end{align}
where the double-prime superscript denotes the second derivative with respect to the first argument. Such reformulations, are often convenient for analysis purposes; see for example Remark \ref{rem:Stein_no}.
\end{remark}

\vp
\begin{remark}[Numerical Evaluations]\label{rem:num}
%\vp\noindent\textbf{Numerical Evaluations}:
The system of equations in \eqref{eq:4eq} comprises of four nonlinear equations in four unknowns. Setting $\tb = (\alpha,\beta,\nu,\kappa)$ for the vector of unknowns, the system of equations in \eqref{eq:4eq} can be written as $\tb=S(\tb)$, for appropriately defined $S:\R^4\rightarrow\R^4$. We have empirically observed that a simple recursion $\tb_{k+1} = S(\tb_k), k=0,1,\ldots$ converges to a solution $\tb_*$ satisfying $\tb_*=S(\tb_*)$. This observation is particularly useful since it allows for efficient numerical experimentations, cf. Section \ref{sec:sim}. It is certainly an interesting and practically useful subject of future work to identify analytic conditions under which such simple recursive schemes provide efficient means of solving \eqref{eq:4eq}.
\end{remark}

\vp
\begin{remark}[Extensions]
%\vp\noindent\textbf{Extensions}:
The results of this section extend naturally, and without any extra effort, to the case of ``block-seperable" loss functions and/or regularizers. A popular example that falls in this category is $\ell_{1,2}$-regularization, which is typically used for the recovery of block-sparse signals. In such a case $f(\x)= \sum_{i=1}^{b}\|[\x]_i\|_2$, where $[\x]_i=[\x_{(i-1)t+1}, \x_{(i-1)t+2},\ldots,\x_{(i-1)t+t}],~i=1,\ldots,b$ is the $i^{\text{th}}$ block of $\x$. Here, $b$  is the number of blocks and $t$ is the length of each block. In the proportional high-dimensional regime, one would assume $b$ growing linearly with $n$ with a constant ratio of $1/t$.
\end{remark}

\section{Examples and Numerical Simulations}\label{sec:sim}

% ++++++++++++++++++++++++++++++++++++++++++++++++++++++++++ %
% ++++++++++++++++++++++++++++++++++++++++++++++++++++++++++ %

\subsection{No Regularization}\label{sec:no_reg}
%We specialize the general result of Theorem \ref{thm:master} to 

Consider an M-estimator  without regularization, i.e.,
\begin{align}\label{eq:f=0_1}
\hat\x := \arg\min_{\x}\sum_{j=1}^m\ell(\y_j-\ab_j^T\x_j).
\end{align}
For simplicity, we consider $\z_j\simiid p_Z$ and a separable loss function. Assuming that $\ell$ and $p_Z$ satisfy the assumptions of Theorem \ref{thm:sep}, and, noting that $f=0\implies F(c,\tau)=0$, the squared error of \eqref{eq:f=0_1} is predicted by the minimizer $\alpha_*$ of the following (SPO) problem
\begin{align}\label{eq:AO_f=0}
\inf_{\substack{\alpha\geq 0\\\taug>0}}\sup_{\beta\geq 0} \frac{\beta\taug}{2} + \delta L(\alpha,\frac{\taug}{\beta}) - \alpha\beta,
\end{align}
where we have performed the (straightforward) optimization over $\tauh$: $\inf_{\tauh>0}\frac{\tauh}{2}+\frac{\beta^2}{2\tauh} = \beta.$ We may equivalently express $\alpha_*$ as the solution to the first-order optimality conditions of \eqref{eq:AO_f=0}. In particular, the stationary equations (see \eqref{eq:4eq}) simplify in this case to the following system of two equations in two unknowns:
\begin{equation}\label{eq:f=0}
 \left\{
 \begin{aligned}
        \al^2&=\delta \kappa^2 \Exp\left[\left(\mathrm{e}_\ell^\prime(\GZ,\kappa)\right)^2\right],\\
        \al &= \delta\kappa\cdot \Exp\left[ \mathrm{e}_\ell^\prime(\GZ,\kappa) \cdot G\right].
       \end{aligned}
 \right.
 \end{equation}

Starting from \eqref{eq:f=0}, some interesting conclusions can be drawn regarding the performance of M-estimators without regularization, which we gather in the following remarks.

\begin{remark}[Stable recovery]\label{rem:stable_f=0}
It follows from \eqref{eq:f=0} that in the absence of regularization, it is required that the number of measurements $m$ is  at least as large as the dimension of the ambient space $n$ ($\delta\geq 1$), in order for the recovery to be \emph{stable}, i.e. the error be finite. To see this, assume stable recovery, then  there exists  $(\al_*,\kappa_*)$ satisfying \eqref{eq:f=0}. Starting from the second equation, applying the \Cauchy inequality and substituting back the first equation we find:
\begin{align}\nn
\al_* = \delta\kappa_*\cdot \Exp\left[ \mathrm{e}_\ell^\prime(\al_*G+Z,\kappa_*) \cdot G\right] \leq \delta\kappa_*\cdot \sqrt{\Exp\left[ \left( \mathrm{e}_\ell^\prime(\al_*G+Z,\kappa_*) \right)^2)\right]}  = \delta\kappa_*\frac{\al_*}{\sqrt{\delta}\kappa_*} \Rightarrow \delta\geq 1.
\end{align}
\end{remark}

\begin{remark}[Stein's Formula]\label{rem:Stein_no}
Assume $e_\ell$ is two times differentiable (e.g., this is the case if $\ell$ is two times differentiable). Then, applying Stein's formula \eqref{eq:Stein's}, a simple rearrangement of \eqref{eq:f=0} shows that
\begin{align}\label{eq:Montanari}
\al_*^2 = \frac{1}{\delta}  \frac{\Exp\left[\left(\mathrm{e}_\ell^\prime(\al_*G+Z,\kappa_*)\right)^2\right]}{\left(\Exp\left[ \mathrm{e}_\ell^{\prime\prime}(\al_*G+Z,\kappa_*)\right]\right)^2}.
\end{align}
The formula above coincides with the corresponding expression in \cite[Thm.~4.1]{montanari13}, but the latter requires additional smoothness and strong-convexity assumptions on $\ell$, which are not necessary for \eqref{eq:f=0} to hold. 
%We have shown the result to hold under significantly relaxed assumptions and have derived it as a simple corollary of Theorem \ref{thm:master}. 
The proof of \cite{montanari13} is based on the AMP framework \cite{AMP}.
%The case of un-regularized regression has been also considered in \cite{Montanari1,karoui131}. The analysis is different than ours and our result holds under more general assumptions. In particular, both \cite{Montanari1,karoui131} require $\ell$   under 
\end{remark}

%\subsubsection{Least-Squares}
\begin{remark}[Least-Squares]
The simplest instance of the general M-estimator is the Least-squares, i.e. $\hat\x:=\min_{\x}\|\y-\A\x\|_2^2$. Of course, in this case, $\hat\x$ has a closed form expression which can be directly used to predict the error behavior (e.g.\cite{LS}). However, for illustration purposes, we show how the same result can be also obtained from \eqref{eq:f=0}. This is also one of the few cases where $\alpha_*$ can be expressed in closed form.
Assume $\delta>1$ and $\z_j\simiid p_Z$ with bounded second moment, i.e. $0<\Exp Z^2=\sigma^2<\infty$. Then, it can be readily checked that all assumptions hold for $\frac{1}{2}(\cdot)^2,p_z$. Also, $e_{\frac{1}{2}(\cdot)^2}^\prime({\chi};{\tau}) = \frac{\chi}{1+\tau}$ and $e_{\frac{1}{2}(\cdot)^2}^{\prime\prime}({\chi};{\tau}) = \frac{1}{1+\tau}$. Solving for the second equation in \eqref{eq:f=0} gives $\kappa_*=\frac{1}{\delta-1}$. Substituting this into the first, we recover the well-known formula
\begin{align}\label{eq:LS}
\alpha_*^2 = \sigma^2\frac{1}{\delta-1}.
\end{align}
\end{remark}
% ++++++++++++++++++++++++++++++++++++++++++++++++++++++++++ %
% ++++++++++++++++++++++++++++++++++++++++++++++++++++++++++ %

\subsection{Ridge Regularzation}\label{sec:ridge}

A popular regularizer in the machine learning and statistics literature is the ridge regularizer  (also known as Tikhonov regularizer), i.e. 
\begin{align}\label{eq:ridge_1}
\hat\x:= \arg\min_{\x}\sum_{j=1}^m\ell(\y_j-\ab_j^T\x_j) + \la \frac{\|\x\|_2^2}{2}.
\end{align}
We specialize Theorem \ref{thm:master} to that case. For simplicity, we assume a separable loss function, and, $\z_j\simiid p_Z$ and $\x_{0,i}\simiid p_X$. 

We will apply Theorem \ref{thm:sep}. Suppose that $\ell$ satisfies the assumptions. Also, assume $\E X_0^2=\sigma_x^2 <\infty$. Then, for $f=\frac{1}{2}(\cdot)^2$, it is easily verified that $\E[(f^\prime(cH + X_0))^2] = \E[(cH + X_0)^2]<\infty$. Hence, the squared-error of \eqref{eq:ridge_1} is predicted by $\alpha_*$, the unique minimizer to the (SPO) in \eqref{eq:AO_det_thm} with 
$$
F(c,\tau)=\frac{ c^2 + \sigma_x^2 }{2(\tau+1)}-\sigma_x^2.
$$

The first-order optimality conditions (see \eqref{eq:4eq}) of this problem simplify after some algebra to the following two equations in two unknowns:
\begin{equation}\label{eq:ridge}
 \left\{
 \begin{aligned}
        \al^2 &= \delta\kappa^2\cdot \Exp\left[ \mathrm{e}_\ell^\prime(\GZ,\kappa)^2 \right] + \la^2\kappa^2\sigma_x^2,\\
\alpha\left(1-\la\kappa\right) &= \delta\kappa\cdot \Exp\left[ \mathrm{e}_\ell^\prime(\GZ,\kappa) \cdot G\right].
       \end{aligned}
 \right.
 \end{equation}

\begin{remark}[Stein's Formula]
Assume $\prox{\ell}{x}{\tau}$ is two times differentiable with respect to $c$ (e.g., this is the case if $\ell$ is two times differentiable), and write $\mathrm{prox}^\prime_\ell(x,\tau)$ for the derivative with respect to $x$. Applying \eqref{eq:Stein's}, a simple rearrangement of \eqref{eq:ridge} yields the following equivalent system of equations
\begin{equation}\label{eq:karoui13}
 \left\{
 \begin{aligned}
% 1-\kappa\la&=\delta\cdot\left( 1 - \Exp\left[ \mathrm{prox}^\prime_\ell(\GZ;\kappa) \right]\right),\\
\delta - 1 +\kappa\la &=\delta\cdot \Exp\left[ \mathrm{prox}^\prime_\ell(\GZ;\kappa) \right] ,\\
\alpha^2 &= \delta\Exp\left[\left(\GZ -  \mathrm{prox}_\ell(\GZ;\kappa) \right)^2\right] + \la^2\kappa^2\sigma_x^2.
       \end{aligned}
 \right.
 \end{equation}
The formula above coincides with the corresponding expression in \cite[Thm.~2.1]{karoui13}\footnote{In comparing \eqref{eq:karoui13} to \cite[Eqn.~(4)]{karoui13}, due to some differences in normalizations the following ``dictionary" needs to be used to match the results: $\alpha\leftrightarrow{r_{\rho}(\kappa)}$, $\kappa\leftrightarrow c_{\rho}(\kappa)$,  $\delta^{-1}\lambda\leftrightarrow\tau$ and $\delta^{-1}\leftrightarrow{\kappa}$.}. The result in \cite{karoui13} requires additional smoothness assumptions on $\ell$. Our result holds under relaxed assumptions and has been derived as a  corollary of Theorem \ref{thm:master}. On the other hand, \cite[Thm.~2.1]{karoui13} is shown to be true for design matrices $\A$ with iid entries beyond Gaussian, e.g. sub-gaussian. 
%The case of un-regularized regression has been also considered in \cite{Montanari1,karoui131}. The analysis is different than ours and our result holds under more general assumptions. In particular, both \cite{Montanari1,karoui131} require $\ell$   under 
\end{remark}

\begin{remark}[Least-squares loss]\label{rem:ls_ridge}
Consider a least-squares loss function where $\ell(x)=\frac{1}{2}x^2$ and a noise distribution of variance $\E Z^2=\sig_z^2<\infty$. Then $\prox{\ell}{x}{\tau}=\frac{x}{1+\tau}$ and $\mathrm{prox}^\prime_\ell(x;\tau)=\frac{1}{1+\tau}$. Substituting  in \eqref{eq:karoui13} gives 
\begin{equation}\label{eq:Ridge_LASSO}
 \left\{
 \begin{aligned}
 1-\kappa\la&=\frac{\delta\kappa}{1+\kappa},\\
\alpha^2(1-\delta\cdot\frac{\kappa^2}{(1+\kappa)^2}) &= \delta\cdot\frac{\kappa^2}{(1+\kappa)^2}\sig_z^2 + \la^2\kappa^2\sigma_x^2.
       \end{aligned}
 \right.
\end{equation}
Now, we can solve these to get the following closed form expression for $\al^*$:
\begin{equation}\label{eq:R_Lasso_final}
\al^2=\left(\delta\cdot\frac{\kappa^2}{(1+\kappa)^2} \cdot\sig_z^2+\la^2\sig_x^2\kappa^2\right)\cdot\left(1-\delta\cdot\frac{\kappa^2}{(1+\kappa)^2}\right)^{-1},
\end{equation}
where 
\begin{equation}\label{eq:ridge_kappa}
\kappa=\frac{1-\delta-\la+\sqrt{(1-\delta-\la)^2+4\la}}{2\la}.
\end{equation}
Observe that letting $\la\rightarrow0$ (which would correspond to ordinary least-squares) and assuming $\delta>1$, $\kappa$ in \eqref{eq:ridge_kappa} approaches
$
{1}/{(\delta-1)}
$
and the optimal $\al^2$ in \eqref{eq:R_Lasso_final} becomes
%\begin{equation}\label{eq:ridge_lasso_nolambda}
${\sig_z^2}/({\delta-1}),$
%\end{equation}
which agrees with \eqref{eq:LS}, as expected.
\end{remark}

\begin{remark}[Achieving the MMSE] 
Let a Gaussian input distribution $\x_{0,i}\simiid\Nn(0,1)$ and any noise distribution of power $\E Z^2=\sigma_z^2<\infty$. We show that a ridge-regularized M-estimator with a least-squares loss function and optimally tuned $\la$ achieves asymptotically the Minimum Mean-Squared Error (MMSE) of estimating $\x_0$ from $\y=\A\x_0+\z$
%, i.e.
%\begin{align}
%\inf_{\la>0}\lim_{n\rightarrow\infty}\frac{1}{n}\|\hat\x-\x_0\|_2^2 = \mathrm{mmse}\left(\x_0|(\A\x_0+\z,\A)\right):=.
%\end{align}
%The quantity on

 First, we use the results of Remark \ref{rem:ls_ridge} to calculate the achieved error of the M-estimator optimized over the values of the regularizer parameter:
% $$
%
% $$
%Given \eqref{eq:R_Lasso_final} and \eqref{eq:ridge_kappa}), this quantity
\begin{align}\label{eq:prove_re}
o_*:= \inf_{\la>0}\lim_{n\rightarrow\infty}\frac{1}{n}\|\hat\x-\x_0\|_2^2 =  \inf_{\la>0}\left\{ \alpha^2(\kappa(\la),\la) \text{ as in  \eqref{eq:R_Lasso_final} } ~|~ \kappa(\la) \text{ satisfies } \eqref{eq:ridge_kappa} \right\}. 
\end{align}
The optimization over $\la$ is possible as follows. From \eqref{eq:Ridge_LASSO}, we find 
\begin{align}\label{eq:pou_na}
\delta\left(\frac{\kappa}{\kappa+1}\right)^2 = \frac{(1-\kappa\la)^2}{\delta}.
\end{align} Substituting this in \eqref{eq:R_Lasso_final}, and denoting $x=\kappa\la$, gives 
\begin{align}
\alpha^2 =  \frac{\delta x^2 + \sig^2(1-x)^2}{\delta - (1-x)^2}.\label{eq:perform}
\end{align}
Minimizing $\alpha^2$ over $\la>0$ in \eqref{eq:prove_re}  is equivalent to minimizing the fraction above over $0<x<1$, since there always exist $\kappa,\la$ satisfying $x=\kappa\la$ and \eqref{eq:pou_na}. Thus, performing the optimization over $0<x<1$ in \eqref{eq:perform} we find %the desired expression for $\alpha^2$ as in \eqref{eq:deal}.
%
%
%%% used Mathematica to do that:           ************************************
%%%%              minimize {d*x^2+r^2*(1-x)^2}/{d-(1-x)^2} , 1>x>0, d>1, r>0 with respect %%%%%%                                                       to x
%%%%
%

%In Appendix \ref{app:ridge_mmse} we have performed the optimization over $\la$ above which shows that
\begin{align}\label{eq:deal}
o_* = \frac{1}{2}\left( 1-\sig^2-\delta + \sqrt{(1-\delta)^2+2\sig^2(\delta+1)+\sig^4} \right).
\end{align}

Next, Wu and Verdu have shown in \cite[Thm.~8,~Eqn.~(56)]{wu2012optimal} that the MMSE is given by the expression in the right-hand side above as well. This, completes the proof of the claim.  
\end{remark}

% ++++++++++++++++++++++++++++++++++++++++++++++++++++++++++ %
% ++++++++++++++++++++++++++++++++++++++++++++++++++++++++++ %
\subsection{Cone-constrained M-estimators}\label{sec:cone}

\subsubsection{Motivation}
Constrained M-estimators solve
\begin{align}\label{eq:cone_LASSO_1}
\min_{\x\in\Cc} 
%\loss(\y-\A\x)
\sum_{j=1}^{m}\ell(\y_j-\ab_j^T\x),
\end{align} 
for some set $\x\in\Cc$. The role of the regularizer in \eqref{eq:genLASSO} is played here by the constraint $\x\in\Cc$. It is common that $\Cc$ takes the form $\Cc=\{ \x ~|~ g(\x)\leq g(\x_0)\}$, i.e. the set of descent directions of some convex function $g$, which is structure inducing for $\x_0$ \cite{Cha,Foygel,OTH13,plan2015generalized}. Of course, such a formulation assumes prior knowledge of the value of $g$ at $\x_0$. Also, in this case, there exists by Lagrangian duality a value of $\la$ for which the regularized M-estimator with $f(x)=g(x)$ is equivalent to \eqref{eq:cone_LASSO_1}.  

A relaxation that is often undertaken to facilitate the analysis of \eqref{eq:cone_LASSO_1}, involves substituting $\Cc$ by its conic hull, which is also known as the tangent cone of $g$ at $\x_0$ (e.g. \cite{Cha}). We call the resulting program, a cone-constrained M-estimator. For the special case of an $\ell_2$-loss function, the squared error performance of constrained M-estimators has been previously considered in \cite{StoLASSO,OTH13} (also, see Remark \ref{rem:LASSO_cone} below). The analysis was performed in the high-SNR regime, where noise variance approaches zero. In this regime it was shown that the conic-relaxation above is \emph{exact}. In this section, we analyze the error performance of cone-constrained M-estimators with general loss functions and derive some interesting conclusions. Also, as we will see, the example here corresponds to an instance of \eqref{eq:genLASSO} with a non-separable regularizer.

\subsubsection{Error Performance} 

We consider
\begin{align}\label{eq:cone_LASSO}
\hat\x := \arg\min_{\x\in\Cc} 
%\loss(\y-\A\x)
\sum_{j=1}^{m}\ell(\y_j-\ab_j^T\x),
\end{align} 
where $$\Cc=\Kc+\x_0:=\{ \la\h~|~ \la\geq 0,  g(\x_0+\h)\leq g(\x_0) \}+\x_0$$ and $g$ a proper, closed, convex function. Here, $\Kc$ is the tangent cone of $g$ at $\x_0$, which is assumed fixed. 
%\chris{Motivate this cone thing. Give few examples. This is a relaxation of Constrained M-estimatots $f(x)\leq f(x_0)$. Mention that in the LASSO, the relaxation is exact in high-SNR, thus predicts high-SNR error. Gives useful insights...}
The constrained minimization above can be written in the general form of regularized M-estimators in \eqref{eq:genLASSO} by choosing the regularizer to be the indicator function for the cone, i.e. $f(\x) = \deltab_{\{\x\in\Cc\}}$ \footnote{Note that this is a non-separable regularizer function.}. Let $\dist{\Cc}{\vb}$, denote the distance of a vector $\vb$ to a set $\Cc$. 
%Also, denote the translation of $\Cc$ at the origin as $\Kc:=\Cc-\x_0$. Note that $\Kc$ is a \emph{cone}. 
We have,
$$
\env{\deltab_{\{\x\in\Cc\}}}{c \h + \x_0 }{\tau} = \frac{1}{2\tau}\min_{\vb\in\Cc-\x_0}\|c\h-\vb\|_2^2 =  \frac{1}{2\tau}\distS{\Kc}{c\h} = \frac{c^2}{2\tau}\distS{\Kc}{\h}.
$$
In the last equality above we have used the homogeneity of the cone $\Kc$.
Let $\Kc^\circ$ denote the polar cone of $\Kc$, and,
$$\DKP:= \Exp\left[ \distS{\Kc^\circ}{\h} \right]=\Exp\left[\|\h\|_2^2 -  \distS{\Kc}{\h} \right].$$
This quantity is known as the \emph{statistical dimension} \cite{TroppEdge} of the cone $\Kc$, or, as the \emph{Gaussian distance squared} \cite{OTH13}. It  can be though of as a measure of the size of the cone, and also, it is very closely  related to the \emph{gaussian width} of $\Kc$\cite{TroppEdge}. We assume that 
\begin{align}\label{eq:DKP}
\frac{\DKP}{n}\rightarrow {\DKPo}\in(0,1).
\end{align}
This translates to an assumption on the degrees of freedom of the structured signal $\x_0$ being proportional to its dimension. For example, for a $k$-sparse $\x_0$ and $g(\x)=\|\x\|_1$, \eqref{eq:DKP} is satisfied for $k=\rho n$, $\rho\in(0,1)$.

With \eqref{eq:DKP}, Assumption \ref{ass:tech}(a) holds with $F(c,\tau) = \frac{c^2}{2\tau}(1-\DKPo)$. For this, it is straightforward to check that Assumption \ref{ass:prop}(a) is also satisfied. Overall, if $\ell,p_Z$ satisfy the conditions of Theorem \ref{thm:sep} and $g,\x_0$ are such that \eqref{eq:DKP} holds, then Theorem \ref{thm:master} applies. Then, the squared error of the cone-constrained M-estimator in \eqref{eq:cone_LASSO} is predicted by the unique minimizer $\alpha_*$ of the (SPO) problem below:
%
%The corresponding (DO) becomes
%\begin{align}\nn%\label{eq:DO_cone}
%%\min_{\substack{\alpha\geq0 \\ \taug\geq 0}}~\max_{\substack{\beta\geq0 \\ \tauh\geq 0}}~~\frac{\beta\taug}{2} + 
%%\delta\cdot\Lm{\alpha}{\frac{\taug}{\beta}}
%%- \frac{\alpha\tauh}{2} + \frac{\tauh}{2\alpha}\sigma_x^2  - \la\cdot \Fm{\frac{\beta}{\la}}{\frac{\tauh}{\alpha\la}}{\frac{\tauh}{\alpha\la}}.
%\min_{\substack{\alpha\geq0 \\ \taug\geq 0}}~\max_{\substack{\beta\geq0 \\ \tauh\geq 0}}~~\frac{\beta\taug}{2} + 
%\delta\cdot\E\left[\env{\ell}{\alpha G+Z}{\taug/\beta}-\ell(Z)\right]
%- \frac{\alpha\tauh}{2} - \frac{\alpha\beta^2}{2\tauh} \DKPo.
%\end{align}
%The optimization over $\tauh$ is easy to perform and yields,
\begin{align}\label{eq:DO_cone}
%\min_{\substack{\alpha\geq0 \\ \taug\geq 0}}~\max_{\substack{\beta\geq0 \\ \tauh\geq 0}}~~\frac{\beta\taug}{2} + 
%\delta\cdot\Lm{\alpha}{\frac{\taug}{\beta}}
%- \frac{\alpha\tauh}{2} + \frac{\tauh}{2\alpha}\sigma_x^2  - \la\cdot \Fm{\frac{\beta}{\la}}{\frac{\tauh}{\alpha\la}}{\frac{\tauh}{\alpha\la}}.
\inf_{\substack{\alpha\geq0 \\ \taug> 0}}~\sup_{\substack{\beta\geq0 }}~~\frac{\beta\taug}{2} + 
\delta\cdot\E\left[\env{\ell}{\alpha G+Z}{\taug/\beta}-\ell(Z)\right] - \alpha\beta\sqrt{\DKPo}.
\end{align}
Compared to \eqref{eq:AO_det_thm}, we have performed the (straightforward) optimization over $\tauh$: $\inf_{\tauh>0}\frac{\tauh}{2}+\frac{\beta^2\DKPo^2}{2\tauh} = \beta\DKPo.$ 
%\chris{Show it has unique $(\beta,\tauh)$ solution. To see that the above is concave in $\beta$ open the envelope expression. Then, it is concave as the point wise minimum of linear (thus, concave) functions. But, how to show strict concavity w.r.t. $\beta$? Should study $\E[\min_v \frac{\beta}{2}(\alpha G + Z-v)^2 + \ell(v)]$. Maybe use the line trick? }.
%Lemma \ref{lem:strict_beta} shows that $\E\left[\env{\ell}{\alpha G+Z}{\taug/\beta}-\ell(Z)\right]$ is strictly concave in $\beta$. Furthermore, from Lemma \ref{lem:strict_ell} it is strictly convex in $(\alpha,\taug)$. Hence, Assumption \ref{ass:uni} holds, i.e. if the optimization in \eqref{eq:DO_cone} attains an optimal point, then this is unique. \chris{I can see $\beta$ is attained because probably $\al,\taug$ will be chosen such that $\taug/2<\al\sqrt{D}$}
%We may now apply Theorem \ref{thm:master} to conclude that if the optimization in \eqref{eq:DO_cone} attains an optimal point $(\al_*,\taug_*,\beta_*)$, then, any optimal solution $\hat\x$ of \eqref{eq:cone_LASSO} satisfies $\lim_{n\rightarrow\infty}\|\hat\x-\x_0\|_2^2=\alpha_*^2$ with probability 1.

\subsubsection{Remarks}
\begin{remark}[Stable recovery]\label{rem:stable_cone}
Starting from \eqref{eq:DO_cone} we can conclude on the minimum number of measurements required for stable recovery. We show that the normalized number of measurements $\delta$ need to be  at least as large as $\DKPo$, in order for the  error to be finite. This is to be compared with the case where no regularization is used that required $\delta\geq 1>\DKPo$ (see Remark \ref{rem:stable_f=0}). To prove the claim, assume finite error, then the value where it converges is predicted by \eqref{eq:DO_cone}. Standard first-order optimality conditions give\footnote{ The three equations in \eqref{eq:cone_eq} correspond to differentiation of the objective of \eqref{eq:DO_cone} with respect to $\tau,\alpha$ and $\beta$, respectively. If any of the variables is zero at the optimal, then, the corresponding equation holding with an inequality is necessary and sufficient. On the other hand, if the optimal is strictly positive, then the equation should hold with equality. }
%\begin{equation}%\label{eq:CLASSO}
\begin{subequations}\label{eq:cone_eq}
\begin{align}
% \left\{
\beta - \frac{\delta}{\beta} \Exp\left[\left(\mathrm{e}^\prime_{\ell}(\GZ;\taug/\beta)\right)^2\right]&\geq 0,\label{eq:cone_1}\\
 \delta \Exp[\mathrm{e}^\prime_{\ell}(\GZ;\taug/\beta)\cdot G] - \beta\sqrt{\DKPo} &\geq 0, \label{eq:cone_2}
 \\
\frac{\tau}{2} + \frac{\delta\tau}{2\beta^2}\Exp\left[\left(\mathrm{e}^\prime_{\ell}(\GZ;\taug/\beta)\right)^2\right] - \alpha\sqrt{\DKPo}&\leq 0.
 \label{eq:cone_3}
%\\
%\taug &= \alpha\sqrt{\DKPo}.
% \right.
 \end{align}
 \end{subequations}
Starting from the second equation, applying the \Cauchy inequality and substituting back the first equation we conclude as follows:
\begin{align}\nn
\beta\sqrt{\DKPo} &\leq \delta \Exp[\mathrm{e}_{\ell}(\GZ;\taug/\beta)\cdot G] \leq \delta \sqrt{\Exp[\left(\mathrm{e}^\prime_{\ell}(\GZ;\taug/\beta)\right)^2]} \leq \delta\frac{\beta}{\sqrt{\delta}} \Rightarrow \delta\geq \DKPo.
\end{align}
\end{remark}

\begin{remark}(Least-squares loss)\label{rem:LASSO_cone}
Consider a least-squares loss function and a noise distribution of variance $\E Z^2=\sigma^2<\infty$. Then, the solution to \eqref{eq:DO_cone} admits an insightful closed form expression. First, in \eqref{eq:DO_cone} perform the optimization over $\taug$. Equating \eqref{eq:cone_1} to 0, gives $\taug=\sqrt{\delta}\sqrt{\alpha^2+\sigma^2}-\beta$. Substituting this in \eqref{eq:DO_cone}, we are left to solve for
$$
\inf_{\alpha\geq 0}\sup_{\beta\geq 0}\beta\left( \sqrt{\delta}\sqrt{\alpha^2+\sigma^2} - \alpha\sqrt{\DKPo} \right) - \frac{\beta^2}{2}.
$$
It can be easily checked that if $\delta>\DKPo$, then the optimal $\alpha_*$ is
\begin{align}\label{eq:LS_cone}
\alpha_*^2 = \sigma^2\frac{\DKPo}{\delta-\DKPo}.
\end{align}
It is insightful to compare this with \eqref{eq:LS}, the corresponding error formula for least-suares: the only difference is that $1$ is substituted with the statistical dimension $\DKPo$. Also, verifying the conclusion of the previous remark, we now require $\delta>\DKPo$ instead of $\delta>1$, implying that recovery is in general possible with less measurements than the dimension of the signal. 

The result in \eqref{eq:LS_cone} was first proved for $\ell_1$-regularization in \cite{StoLASSO}, and, was later generalized in \cite{OTH13,TOH14} (also, \cite{LS}). Compared to the lengthy treatments in those references, the result followed here as a simple corollary of Theorem \ref{thm:master}.
\end{remark}

\begin{remark}(Lower Bound)
In \eqref{eq:cone_2} apply Stein's inequality and combine it with \eqref{eq:cone_1} to yield
\begin{align}\label{eq:above99}
\alpha^2\geq \frac{\DKPo}{\delta}\frac{{\beta^2}/{\delta}}{\Exp\left[\mathrm{e}^{\prime\prime}_{\ell}(\GZ;\taug/\beta)\right]}
\geq  \frac{\DKPo}{\delta}\frac{\Exp\left[\left(\mathrm{e}^\prime_{\ell}(\GZ;\taug/\beta)\right)^2\right]}{\Exp\left[\mathrm{e}^{\prime\prime}_{\ell}(\GZ;\taug/\beta)\right]}
\end{align}
For the first inequality above, we have assumed that at the optimal, ${\Exp\left[\mathrm{e}^{\prime\prime}_{\ell}(\GZ;\taug/\beta)\right]}<\infty$. When this holds, (see Remark \ref{rem:LAD_zero} for an instance where this is not the case) we can use the above to lower bound the error performance in terms of the Fisher information of the noise. Based on a result of \cite{madiman2007generalized}, Donoho and Montanari prove in \cite[Lem.~3.4,3.5]{montanari13} that  the right-hand side in \eqref{eq:above99} is further lower bounded by ${I(Z)}/(1+\alpha^2 I(Z))$, where $I(Z)=\E\left(\frac{\partial}{\partial z}\log p_Z(z)\right)^2$ denotes the Fisher information of the random variable $Z$, which is assumed to have a differentiable density. Using this and solving for $\alpha^2$, we conclude with 
\begin{align}\label{eq:lbbb}
\alpha^2 \geq \frac{\DKPo}{\delta-\DKPo} \frac{1}{I(Z)}.
\end{align}
For Gaussian noise of variance $\sigma^2$, we have $1/I(Z)=\sigma^2$. In this case the lower bound in \eqref{eq:lbbb} coincides with the error formula of the least-squares loss function, which then proves optimality of the latter.

\end{remark}

\begin{remark}(Consistent Estimators)\label{rem:LAD_zero}
The lower bound in \eqref{eq:lbbb} only holds if the optimal $\alpha_*$ in \eqref{eq:DO_cone} is strictly positive. This is not always the case: under circumstances, it is possible to choose the loss function such that the resulting cone-constrained M-estimator is consistent. Theorem \ref{thm:master} is the starting point to identifying such interesting scenarios.

Here, we illustrate this through an example: we assume a sparse gaussian-noise model and use a Least Absolute Deviations (LAD) loss function. More precisely, $p_Z(Z)= \bar{s}\delta_0(Z) + (1-\bar{s}) \frac{1}{\sqrt{2\pi}}\exp(-Z^2/2), \bar{s}\in(0,1)$ and $\ell(v)=|v|$. In Section \ref{app:LAD_zero} we prove that when $\bar{s},\delta$ and $\DKPo$ are such that  
\begin{equation}\label{eq:cond_per}
\delta\geq \DKPo+ \min_{\kappa>0} \left\{\bar{s}(1+\kappa^2)+(\delta-\bar{s})\sqrt\frac{2}{\pi}\int_\kappa^\infty (G-\kappa)^2\exp(-G^2/2)\mathrm{d}G \right\},
\end{equation}
then the first-order optimality conditions in \eqref{eq:cone_eq} are satisfied for $\alpha\rightarrow0,\taug\rightarrow0$ and some $\beta>0$. Thus, when the number of measurements is large enough such that \eqref{eq:cond_per} holds, then $\alpha_*=0$, and, $\x_0$ is perfectly recovered\footnote{
In the context that it appears here, the perfect recovery condition in \eqref{eq:cond_per} has been shown previously in \cite{Allerton14}. The problem is very closely related to the demixing problem in which one aims to extract  two (or more) constituents from a mixture of structured vectors \cite{mccoy2014convexity}. In that context, recovery conditions like the one in \eqref{eq:cond_per} have been generalized to other king of structures beyond sparsity \cite{mccoy2014sharp,mccoy2014convexity,Foygel}. Our purpose here has been to illustrate how Theorem \ref{thm:master} can be used to derive such results. Besides,  the generality of the paper's setup offers the potential of extending such consistency-type results beyond cone-constrained M-estimators and beyond fixed signals $\x_0$. This is an interesting direction of future research. }.

% this through an example that is summarized in the lemma below. The proof is deferred to Section \ref{sec:proofs_4}.
%
% Using \eqref{eq:cone_eq}, one can derive the following sufficient condition for the perfect recovery which is also investigated in \ehsan{refer to trop paper}
%\begin{lem}\label{lem:pref_rec}
%Consider the special case where $\mathcal{L}(x)=\|.\|_1$. For convenience, we assume $Z$ is a s-sparse signal whose first $s$ entries are Normal Gaussian and the rest zero. Then the following condition is sufficient for perfect recovery.
%\begin{equation}
%\DKPo+\min_{\kappa>0} \bar{s}(1+\kappa^2)+2(\delta-\bar{s})\int_\kappa^\infty (G-\kappa)^2\phi(G)\mathrm{d}G\leq \delta
%\end{equation}
%Where $\min_{\kappa>0} \bar{s}(1+\kappa^2)+2(\delta-\bar{s})\int_\kappa^\infty (G-\kappa)^2\phi(G)\mathrm{d}G$ is known as the statistical dimension of the noise.
%\end{lem}
%We will prove this lemma in section \ref{sec:proofs_4}.
\end{remark}

% ++++++++++++++++++++++++++++++++++++++++++++++++++++++++++ %
% ++++++++++++++++++++++++++++++++++++++++++++++++++++++++++ %
\subsection{Generalized  LASSO}

The generalized LASSO solves
\begin{equation}\label{eq:LASSO}
\hat\x := \arg\min_{\x} \frac{1}{2}\|\y-\A\x\|_2^2 + \la f(\x).
\end{equation}
For simplicity, suppose that $f$ is separable and satisfies the assumptions of Theorem \ref{thm:sep}. Also, assume $\z_j\simiid p_Z$ such that $0< \E Z^2 =:\sigma^2 <\infty$. Then, for $\ell=\frac{1}{2}(\cdot)^2$, it is easily verified that $\E[(\ell^\prime(cG + Z ))^2] = \E[(cG + Z)^2]<\infty$. Hence, the squared-error of \eqref{eq:LASSO} is predicted by $\alpha_*$, the unique minimizer to the (SPO) in \eqref{eq:AO_det_thm} with 
$
L(c,\tau)=\frac{ c^2 + \sigma^2 }{2(\tau+1)}-\sigma^2.
$

Equivalently, the error is predicted by the solution to the stationary equations in \eqref{eq:4eq} with $e^\prime_{\frac{1}{2}(\cdot)^2}(\chi;\tau) = \frac{\chi}{1+\tau}.$
The second and third equations in \eqref{eq:4eq} give
\begin{align*}
\beta^2(1+\kappa)^2 &= \delta(\alpha^2+\sigma^2),\nonumber\\
\nu(1+\kappa) &= \delta.
\end{align*}
Solving these for $\kappa$ and $\nu$, and substituting them in the remaining two equations results in the following system of two nonlinear equations in two unknowns
\begin{align}
\begin{cases}\label{eq:gen_lasso_1}
\delta\frac{\al^2}{\al^2+\sig^2}=\E\left[\left(\frac{\la}{\beta} e^\prime_{f}\left(\frac{\sqrt{\al^2+\sig^2}}{\sqrt\delta}H+X_0,\la\frac{\sqrt{\al^2+\sig^2}}{\beta\sqrt\delta}\right)-H\right)^2\right] \\
\beta(1-\delta)+\beta^2\frac{\sqrt\delta}{\sqrt{\al^2+\sig^2}}=\la\E\left[e^\prime_{f}\left(\frac{\sqrt{\al^2+\sig^2}}{\sqrt\delta}H+X_0,\la\frac{\sqrt{\al^2+\sig^2}}{\beta\sqrt\delta}\right)\cdot H\right].
\end{cases}
\end{align}
For the special case of $\ell_1$-regularization, the result above was proved by Bayati and Montanari \cite{montanariLASSO} using the AMP framework. In the generality presented here, the result appears to be novel.

\begin{remark} (Not consistent) An interesting observation from \eqref{eq:gen_lasso_1} is that the generalized LASSO cannot achieve perfect recovery, irrespective of the choice of the regularizer function. 
To see this, the first equation in \eqref{eq:gen_lasso_1} for $\al=0$ gives 
%\begin{align}
$\E\left[\left( \frac{\la}{\beta}e^\prime_{f}(\frac{\sig}{\sqrt\delta}H+X_0,\frac{\la\sig}{\beta\sqrt\delta})-H\right)^2\right]=0.$
%\end{align}
Then,  it must hold, almost surely, that the argument under the expectation sign be equal to zero.
%\begin{equation}
%e^\prime_{f}\left(\frac{\sig}{\sqrt\delta}H+X_0,\frac{\la\sig}{\beta\sqrt\delta}\right)=\frac{\beta}{\la}H.\nn
%\end{equation}
Evaluating the derivative of the envelope function as in Lemma \ref{lem:Mconvex}(iii), this becomes equivalent to 
$
X_0=\prox{f}{\frac{\sig}{\sqrt\delta}H+X_0}{\frac{\la\sig}{\beta\sqrt\delta}}.
$
This, when combined with the optimality conditions for the Moreau envelope (see \eqref{eq:opt}) gives that almost surely
$
\frac{\sig}{\sqrt\delta}H\in\partial f(X_0).
$
Thus, we have reached a contradiction because $H$ can take any real value as a Gaussian random variable.
\end{remark}

%
%\begin{remark} (Not consistent) An interesting observation from \eqref{eq:gen_lasso_1} is that the generalized LASSO cannot achieve perfect recovery, irrespective of the choice of the regularizer function. 
%To see this, assume that there exists $\beta$ such that $\alpha=0$ is a solution to \eqref{eq:gen_lasso_1}. The, the first equation would give
%\begin{align}
%\E\left[\left( e^\prime_{f}(\frac{\sig}{\sqrt\delta}H+X_0,\frac{\la\sig}{\beta\sqrt\delta})-\frac{\sig}{\sqrt\delta}H\right)^2\right]=0\nonumber,
%\end{align}
%and t
%Substituting this , it must be $\beta\rightarrow 0$. Thus, almost surely it must hold
%\begin{equation}\label{eq:outekan}
%\lim_{\tau\rightarrow\infty}e^\prime_{f}\left(\frac{\sig}{\sqrt\delta}H+X_0,\tau\right)=\frac{\sig}{\sqrt\delta}H \quad\Leftrightarrow \quad \min_{x} f(x) = \frac{\sig}{\sqrt\delta}H,
%\end{equation}
%where the equivalence follows from \ref{lem:Mconvex}(vi). Since $H$ has a Gaussian distribution, \eqref{eq:outekan} is a clear contradiction.
% ++++++++++++++++++++++++++++++++++++++++++++++++++++++++++ %
% ++++++++++++++++++++++++++++++++++++++++++++++++++++++++++ %

\subsection{Square-root LASSO}\label{sec:sq_LASSO}
The (Generalized) Square-root LASSO (also known as $\ell_2$-LASSO \cite{OTH13}) solves\footnote{We refer the interested reader to \cite{Belloni,OTH13,TOH14} for a discussion on the similarities and differences between \eqref{eq:square-root_1} and the Generalized LASSO of \eqref{eq:LASSO}.} 
\begin{align}\label{eq:square-root_1}
\hat\x := \arg\min_{\x} \sqrt{n}\|\y-\A\x\|_2 + \la f(\x).
\end{align} 
 In contrast, to the other examples in this section,  the square-root LASSO is an instance of \eqref{eq:genLASSO} with a \emph{non}-separable loss function. Observe the normalization of the loss function with a $\sqrt{n}$-factor. This is to satisfy our condition of Section \ref{sec:model} that $(\forall c>0) (\exists C>0)~\left[\|\vb\|_2\leq c\sqrt{n}\implies \frac{1}{\sqrt{n}}\sup_{\s\in\partial\loss(\vb)}\|\s\|_2\leq C\right]$.
 
  In Section \ref{app:square} we show that when $\loss(\vb)=\sqrt{n}\|\vb\|_2$ and $\z\sim p_\z$ with $\E\left[{\|\z\|_2^2}/{m}\right]=\sigma^2\in(0,\infty)$, then Assumption \ref{ass:tech}(a) holds with 
\begin{align}\label{eq:sqr_L}
L(\al ,\tau)=\begin{cases}
\frac{1}{\sqrt{\delta}}(\sqrt{\al^2+\sig^2}-\sig) - \frac{\tau}{2\delta} &,\text{if}~ \sqrt\delta\sqrt{\al^2+\sig^2}\geq\tau,\\ 
\frac{1}{2\tau}(\al^2+\sig^2)-\frac{\sig}{\sqrt\delta} &,\text{otherwise}.
\end{cases}
\end{align}
Also, Assumption \ref{ass:tech}(b) is trivially satisfied, and, Section \ref{app:square} shows the same for Assumptions \ref{ass:prop}(b)-(d). Thus, considering any regularizer that satisfies Assumptions \ref{ass:tech}(a) and \ref{ass:prop}(a), Theorem \ref{thm:master} applies, and predicts the squared error of \eqref{eq:square-root_1} as the unique minimizer $\alpha_*$ to the following optimization:
\begin{align}\label{eq:square-root_4}
\inf_{\substack{\alpha\geq0 }}~\sup_{\substack{\beta\geq0 \\ \tauh> 0}}~~ 
- \frac{\alpha\tauh}{2} - \frac{\alpha\beta^2}{2\tauh} + \la\cdot \FFm{\frac{\alpha\beta}{\tauh}}{\frac{\alpha\la}{\tauh}}+
\begin{cases}
\beta\sqrt\delta\sqrt{\al^2+\sig^2}&,\text{if}~ \beta\leq1\\
\sqrt\delta\sqrt{\al^2+\sig^2}&,\text{otherwise}
\end{cases}.
\end{align}
To arrive to \eqref{eq:square-root_4} starting from \eqref{eq:AO_det_thm}, we have replaced $L$ with \eqref{eq:sqr_L} and have performed the minimization over $\taug$ as shown below:
%First, we make use of the following lemma which we will prove in section \ehsan{refer}.
%\begin{lem}\label{lem:square_L_func}
%In the case that $\mathcal{L}(x)=\sqrt n\|x\|_2$ and $\z\sim p_z$ with $\E\frac{\|\z\|_2^2}{m}=\sig^2\in(0,\infty)$, we have
%\begin{equation*}
%\frac{1}{m}(\env{\sqrt{n}\|\cdot\|_2}{\al \g+\z}{\tau}-\|\z\|_2\sqrt n) \xrightarrow{p}L(\al ,\tau)
%\end{equation*}
%where
%\begin{align}\label{eq:sqr_L}
%L(\al ,\tau)=\begin{cases}
%\frac{1}{\sqrt{\delta}}(\sqrt{\al^2+\sig^2}-\sig) - \frac{\tau}{2\delta} &,\text{if}~ \sqrt\delta\sqrt{\al^2+\sig^2}\geq\tau,\\ 
%\frac{1}{2\tau}(\al^2+\sig^2)-\frac{\sig}{\sqrt\delta} &,\text{otherwise}.
%\end{cases}
%\end{align}
%Besides, this function also satisfies assumption \ref{ass:prop}.
%\end{lem}
\begin{align}\label{eq:square-root_3}
\inf_{\tau_g\geq0}\begin{cases}
\frac{\beta\tau_g}{2}-\frac{\tau_g}{2\beta} +\sqrt\delta\sqrt{\al^2+\sigma^2}&,\text{if}~ \delta(\al^2+\sigma^2)\geq\frac{\tau_g^2}{\beta^2}\\
\frac{\beta\delta}{2\tau_g}(\al^2+\sigma^2)+\frac{\beta\tau_g}{2}&,\text{otherwise}.
\end{cases}=\begin{cases}
\beta\sqrt\delta\sqrt{\al^2+\sig^2}&,\text{if}~ \beta\leq1\\
\sqrt\delta\sqrt{\al^2+\sig^2}&,\text{otherwise}
\end{cases}.
\end{align}

The optimization in \eqref{eq:square-root_3} can be simplified one step further. It is shown in Section \ref{app:square} that $-\frac{\alpha\beta^2}{2\tauh} + \la F\left(\frac{\alpha\beta}{\tauh},\frac{\alpha\la}{\tauh}\right)$ is a non-increasing function of $\beta$ for $\beta>0$. Therefore, the (SPO) becomes equivalent to the following
\begin{align}\label{eq:square-root_7}
\inf_{\substack{\alpha\geq0 }}~\sup_{\substack{0\leq\beta\leq1 \\ \tauh\geq 0}}~~ \beta\sqrt\delta\sqrt{\al^2+\sig^2}
- \frac{\alpha\tauh}{2} - \frac{\alpha\beta^2}{2\tauh} + \la\cdot \FFm{\frac{\alpha\beta}{\tauh}}{\frac{\alpha\la}{\tauh}}.
\end{align}

%\begin{remark}
The fact that the optimization in \eqref{eq:square-root_7} predicts the squared error of \eqref{eq:square-root}, has been recently shown by the authors in \cite{NIPS15}. That work only considers the square-root LASSO\footnote{Note however, that \cite{NIPS15} considers a more general measurement model than the one of the current paper, one that allows for nonlinearities.}, while here, we have (re)-derived the result as a corollary of the general Theorem \ref{thm:master}. 
%\end{remark}

%% ++++++++++++++++++++++++++++++++++++++++++++++++++++++++++ %
%% ++++++++++++++++++++++++++++++++++++++++++++++++++++++++++ %
%

%\begin{figure}[h!]
%    \centering
%	\begin{subfigure}{.4\textwidth}
%	\centering
%    \includegraphics[width=1\textwidth]{l1l1_sparse07}
%    \caption{$\delta$=0.7}
%    \label{fig:l1l1_sparse07}
%	\end{subfigure}
%\begin{subfigure}{.4\textwidth}
%	\centering
%    \includegraphics[width=1\textwidth]{l1l1_sparse12}
%    \caption{$\delta$=1.2}
%    \label{fig:l1l1_sparse07}
%	\end{subfigure}
%\caption{Squared error of the $\ell_1$-Regularized LAD under sparse noise  with a Gaussian measurement matrix
%$\A$ ($\circ$) and with a Bernoulli one ($+$) as a function of the regularizer parameter $\la$; 
%   both compared to the asymptotic prediction. Here, the noise $\z$ is assumed iid sparse-Gaussian with $p_z=.7\delta_0 + .3\Nn(0,1)$; same model is for $\x_0$, with $p_x=(.9\delta_0 + .1\Nn(0,1))/\sqrt{.1}$. The ambient dimension of the unknown signal was chosen $n=768$. The different curves correspond to $0.7n$
%and $1.2n$ number of measurements, respectively. 
%Simulation points are averages over 25 problem realizations.}
%\end{figure}
%%\subsection{Generalized-LASSO}
%%As the next step, we investigate Generalized Lasso under Gau

\subsection{Heavy-tails}
In this section, we investigate instances where the noise distribution has unbounded moments. In the presence of (say) heavy-tailed noise, it is a common practice to use a loss function that grows to infinity no faster than linearly. This is also suggested by Assumption \ref{ass:tech}(b) (cf. \eqref{eq:ass_main1} for the separable case), as has already been discussed.

For illustration, we assume $\z\simiid \mathrm{Cauchy}(0,1)$ and consider two examples of loss functions for which we show that Theorem \ref{thm:master} is applicable.

\subsubsection{LAD}\label{sec:LAD}
As a first example, consider the regularized-LAD estimator:
\begin{equation}\label{eq:eg1}
\hat{\x}=\arg\min_{\x}\|\y-\A\x\|_1+\la f(\x),
\end{equation}
The loss function is separable, with $\ell(v)=|v|$. Easily, for all $c\in\R$
$$
\Exp\left[ |\ell_+^\prime(cG+Z)|^2 \right] = \Exp\left[ |\mathrm{sign}(cG+Z)|^2 \right] = 1 <\infty,
$$
satisfying Assumption \eqref{eq:ass_main}. Also, $\E Z^2$ is undefined, but, $\sup_{v}\frac{\ell(v)}{|v|}=1<\infty$, thus, \eqref{eq:ass_main1} holds. Finally, $|\cdot|$ is not differentiable at zero satisfying the conditions of Lemma \ref{lem:strict_ell}. With these, Theorem \ref{thm:sep} is applicable.

\subsubsection{Huber-loss}\label{sec:Hl}
The Huber-loss function with parameter $\rho>0$ is defined as
\begin{align}\label{eq:Huber}
h_\rho(v) = \begin{cases} \frac{v^2}{2} &,|v|\leq\rho, \\ \rho|v| - \frac{\rho^2}{2} &,\text{otherwise}. \end{cases}
\end{align}
Consider a regularized M-estimator with $\ell(v) = h_\rho(v)$. We show here that this choice satisfies the Assumptions of Theorem \ref{thm:sep}. Indeed, for all $c\in\R$
\begin{align}\nn
\Exp\left[ |\ell_+^\prime(cG+Z)|^2 \right] \leq \Exp\left[ |cG+Z| ~\big|~|cG+Z|\leq \rho  \right] +    \Exp\left[ \rho ~\big|~ |cG+Z|> \rho \right]  <\infty,
\end{align}
satisfying Assumption \eqref{eq:ass_main}. Also, $\sup_{v}\frac{\ell(v)}{|v|}=\rho<\infty$, thus, \eqref{eq:ass_main1} holds. Finally, $h_\rho$ is differentiable with a strictly increasing derivative in the interval $[-\rho,\rho]$. With these, Theorem \ref{thm:sep} is applicable. Figure \ref{fig:Huber} illustrates the validity of the prediction via numerical simulations.

%As the first example we investigate the case when both loss function and regularizer are $\|.\|_1$. Such scenario could be useful for situations where noise is spars Gaussian while the signal is also spars. The optimization problem would be as follows
%\begin{equation}\label{eq:eg1}
%\hat{x}=arg\min_{x}\|y-Ax\|_1+\la\|x\|_1
%\end{equation}
%The following lemma shows that we have the conditions of the assumption \ehsan{refer} in this case. (See section \ehsan{refer} for proof).
%\begin{lem}\label{lem:ell1_L} For a given $\al\geq0$, where $G\sim\mathcal{N}(0,1)$ and $Z\sim p_z$, we have
%$$\E[|\env{\ell_1}{\GZ}{\tau}-|Z||]<\infty$$ 
%\end{lem}
%Now to find the solutions to the system of equations in \eqref{eq:4eq}, we are interested in $\prox{\ell_1}{\chi}{\tau}$ which can be easily calculated.
%\begin{equation}
%\prox{\ell_1}{\chi}{\tau}=
% \left\{
% \begin{aligned}
%       & -\tau+\chi\quad \chi>\tau\\
%        &\tau+\chi\quad \chi<-\tau\\
%        &0\qquad\quad o.w.
%       \end{aligned}
% \right.
%\end{equation}
%In the figure (?) bellow, behavior of the solution to \eqref{eq:4eq} is plotted for a range of $\la\in[0,3]$ for two value of $\delta:=\frac{m}{n}=\{0.7,1.2\}$ and being compared to the empirical results which are derived by solving \eqref{eq:eg1}.

\begin{figure}[h!]
    \centering
    \includegraphics[width=0.65\textwidth]{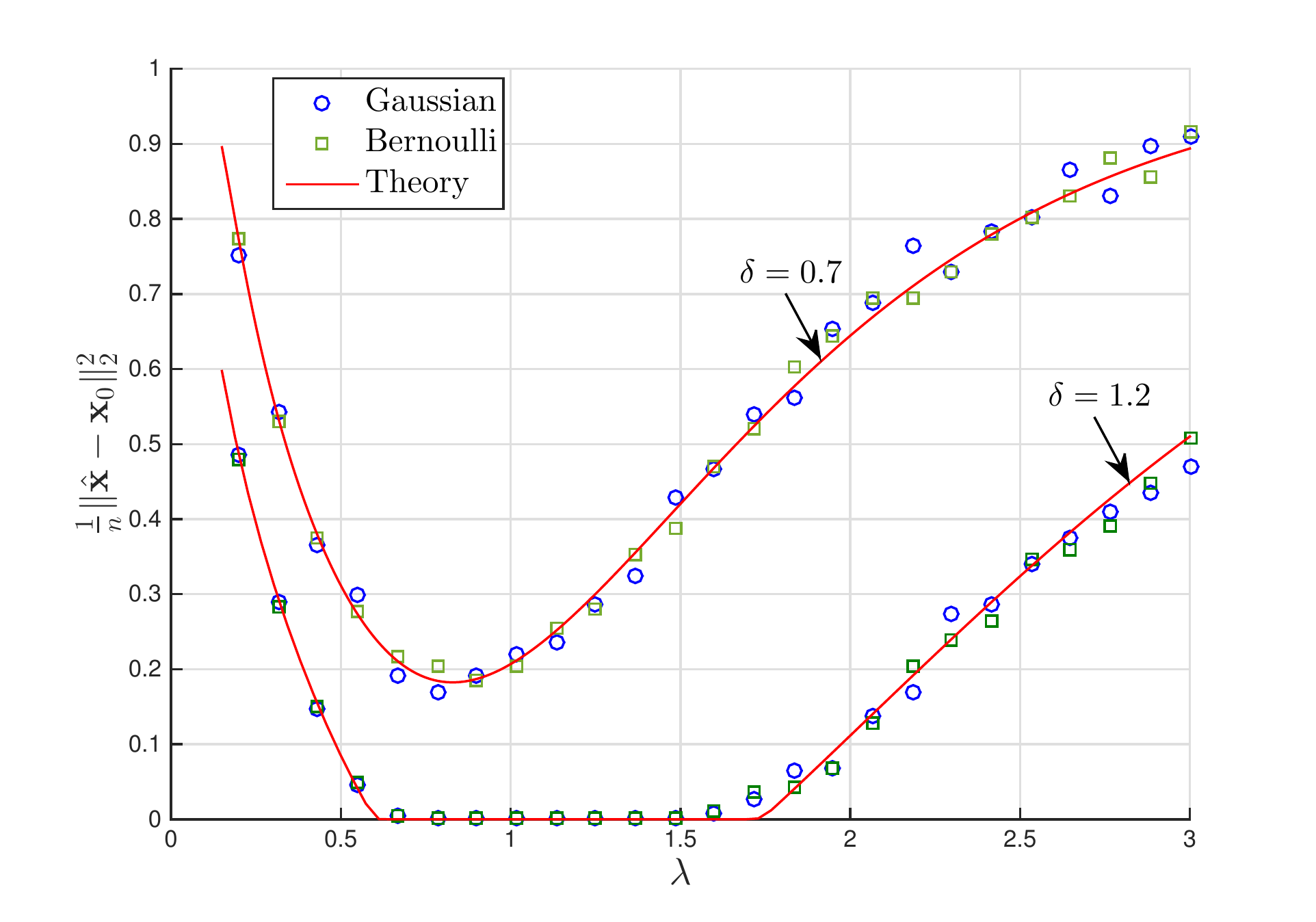}
    \caption{\footnotesize{
        Squared error of the $l_1$-Regularized LAD with Gaussian ($\circ$) and Bernoulli ($\square$) measurements as a function of the regularizer parameter $\la$ for two different values of the normalized number of measurements, namely $\delta=0.7$ and $\delta=1.2$. Also, $\x_{0,i}\simiid p_x(x)=0.9\delta_0(x)+0.1\phi(x)/\sqrt{0.1}$ and $\z_j\simiid p_{z}(z)=0.7\delta_0(z)+0.3\phi(z)$ for $\phi(x)=\frac{1}{\sqrt{2\pi}}e^{-x^2/2}$. For the simulations, we used $n=768$ and the data were averaged over 5 independent realizations.    }}
    \label{fig:bern}
\end{figure}

\begin{figure}[h!]
    \centering
    \includegraphics[width=0.65\textwidth]{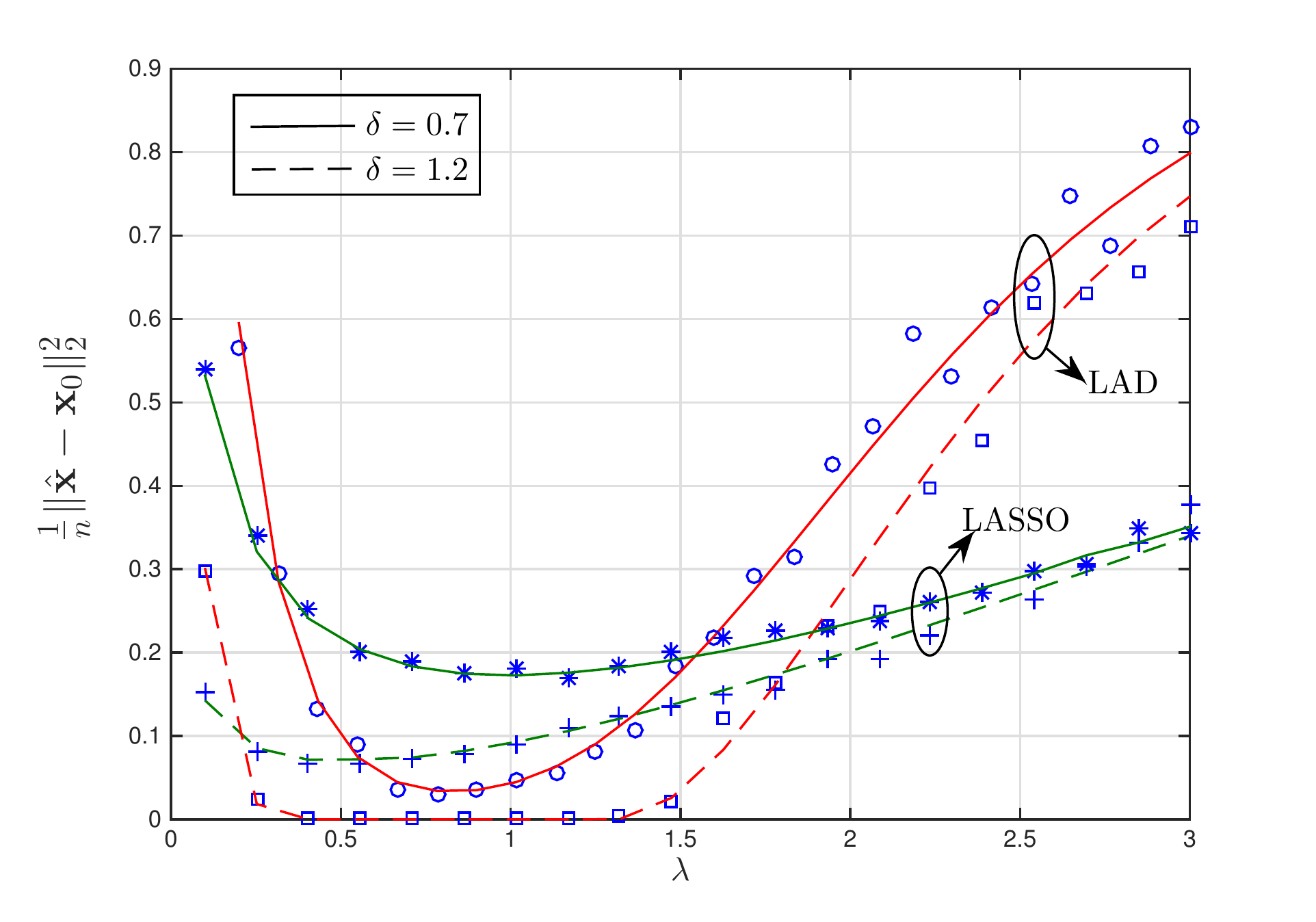}
    \caption{\footnotesize{Comparing the squared error of the $\ell_1$-Regularized LAD with the corresponding error of the LASSO. Both are plotted as functions of the regularizer parameter $\la$, for two different values of the normalized measurements, namely $\delta=0.7$ and $\delta=1.2$. The noise and signal are iid sparse-Gaussian as follows: $\x_{0,i}\simiid p_{x}(x)=0.9\delta_0(x)+0.1\phi(x)/\sqrt{0.1}$ and $\z_j\sim p_{z}(z)=0.9\delta_0(z)+0.1\phi(z)$ with $\phi(x)=\frac{1}{\sqrt{2\pi}}e^{-x^2/2}$. For the simulations, we used $n=768$ and the data were averaged over 5 independent realizations.   }}
    \label{fig:LAD_sparse}
\end{figure}

\begin{figure}[h!]
    \centering
    \includegraphics[width=0.65\textwidth]{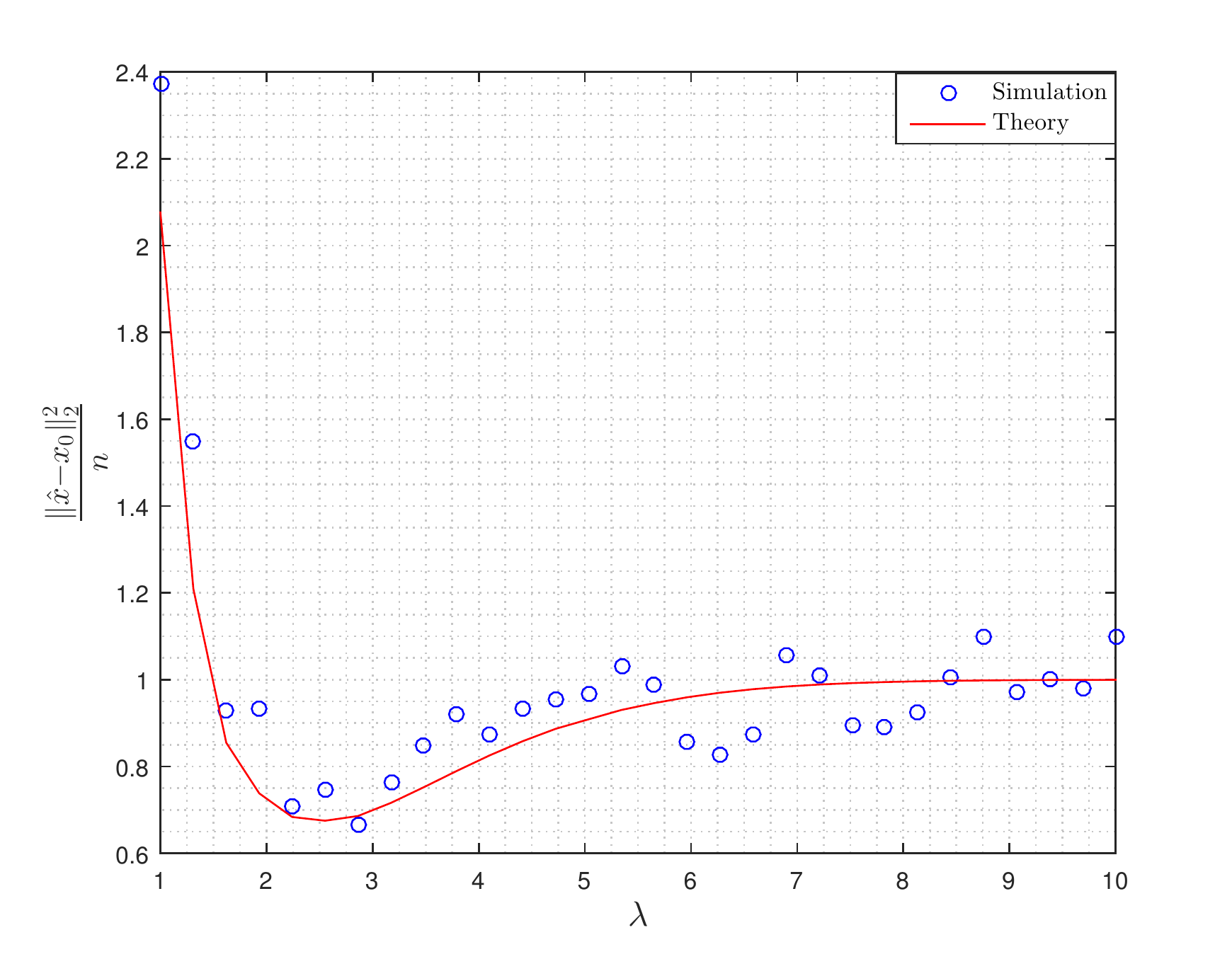}
    \caption{\footnotesize{Squared error of the $\ell_1$-Regularized M-Estimator with Huber-loss as a function of the regularizer parameter $\la$. Here, $\delta=0.7$, $\x_0\simiid p_{x}(x)=0.9\delta_0(x)+0.1\phi(x)/\sqrt{0.1}$ and $p_{z}(z)=0.9\delta(z)+0.1\eta(z)$ with $\phi(x)=\frac{1}{\sqrt{2\pi}}e^{-x^2/2}$ and $\eta(z)=\frac{1}{\pi (1+z^2)}$. For the simulations, we used $n=1024$ and the data are averaged over 5 independent realizations.  }}
    \label{fig:Huber}
\end{figure}

\begin{figure}[h!]
    \centering
    \includegraphics[width=0.65\textwidth]{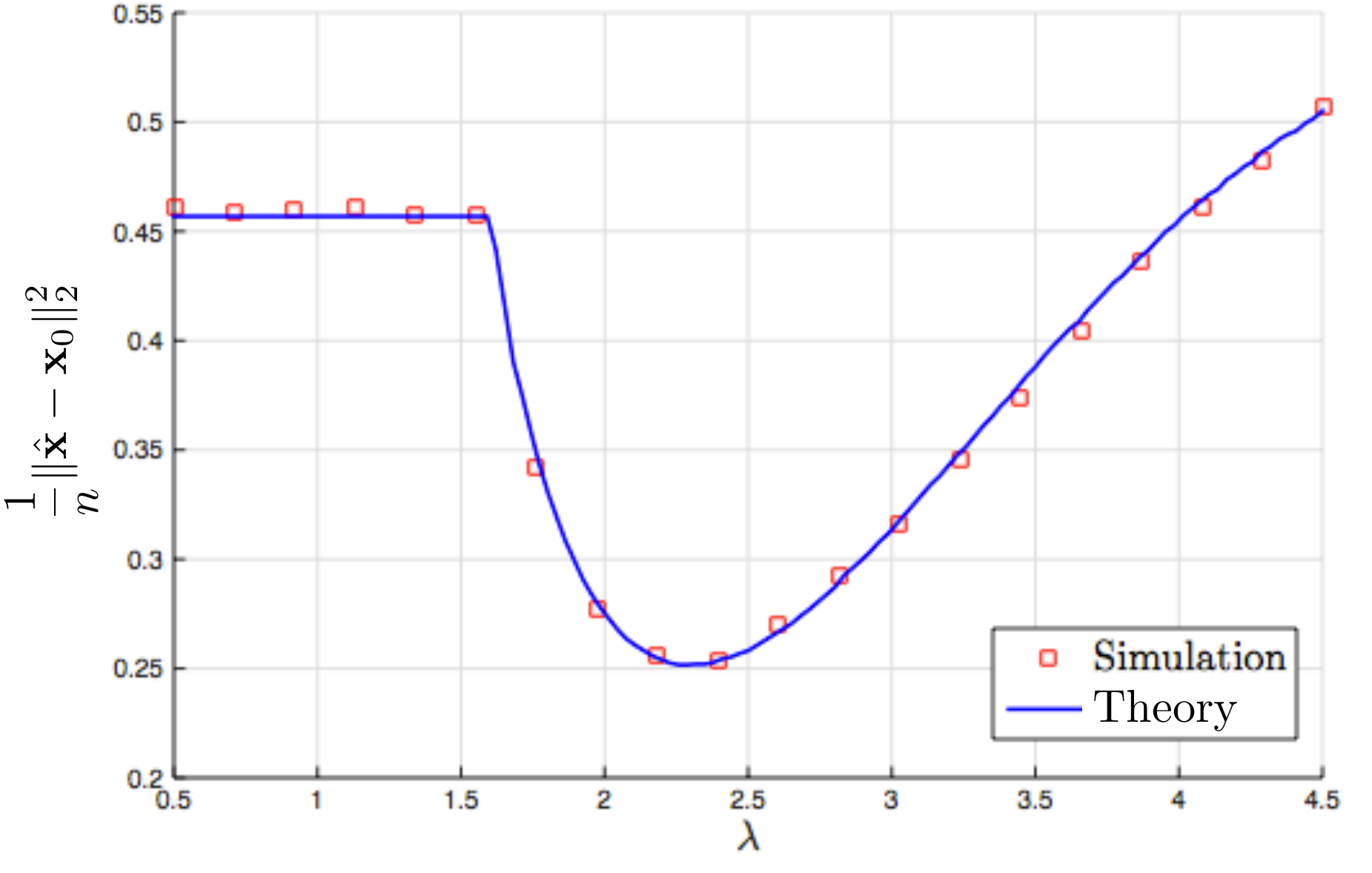}
    \caption{\footnotesize{Squared error of the $\ell_{1,2}$-Regularized Lasso for group sparse signal composed of $512$ blocks of  size $3$ each,  as a function of the regularizer parameter $\la$. Here, $\delta=0.75$, each block is zero with probability $0.95$, otherwise its entries are i.i.d. $\mathcal{N}(0,1)$ and $\z_j\simiid p_{z}(z)=0.3\phi(z)$ with $\phi(x)=\frac{1}{\sqrt{2\pi}}e^{-x^2/2}$. The simulations are averaged over 10 independent realizations.   }}
    \label{fig:group}
\end{figure}

%\vspace{-10pt}
%\subsection{Specific instances}
\subsection{Numerical Simulations}\label{sec:num}

%We illustrate the validity of the prediction of Theorem \ref{thm:master} via numerical simulations on the following specific instances of \eqref{eq:genM}:

We have performed a few numerical simulations on specific instances of M-estimators that were previously discussed in Section \ref{sec:sim}. The purpose is to illustrate both the validity of the prediction of Theorem \ref{thm:master}, as well as, that of the remarks that followed as a consequence of it. 

\vp
\noindent{\underline{ Figure \ref{fig:bern} }.}~We consider the regularized LAD estimator of \eqref{eq:eg1} under an iid sparse-Gaussian noise model. The unknown signal is also considered sparse, which leads to the natural choice of $\ell_1$ regularization, i.e. $f(x) = \|\x\|_1$. Apart from the very close agreement of the theoretical prediction of Theorem \ref{thm:master} to the simulated data, the following facts are worth observing.
\begin{itemize}
\item[-] When the number of measurements $m$ gets large enough, then, for an appropriate range of values of the regularizer parameter, the estimator is consistent, i.e. the unknown signal $\x_0$ is perfectly recovered. This is relevant to Remark \ref{rem:LAD_zero} where we proved this to be the case for the closely related cone-constrained LAD estimator. For that, we were able to quantify how large $m$ should be as a function of the sparsities of the noise and of the signal, see \eqref{eq:cond_per}.

\item[-] The prediction of Theorem \ref{thm:master} remains accurate when the measurement matrix has entries iid Bernoulli ($\{\pm1\}$). This suggests that the error behavior (at least of this specific instant of M-estimator) undergoes some universality properties. See also the relevant discussion in Section \ref{sec:conc}.
\end{itemize}

\vp
\noindent{\underline{ Figure \ref{fig:LAD_sparse} }.}~ The model for both the noise and for  the unknown signal is here the same as in Figure \ref{fig:bern}, i.e. both are iid sparse. We use $\ell_1$-regularization, and, two different loss functions, namely, a least-absolute-deviations one and a least-squares one, corresponding to a LAD and a LASSO estimator, respectively. The figure aims to compare the performance of the two. Intuition suggests that the LAD is more appropriate for a sparse noise model, since $\ell_1$ promotes sparsity. This is indeed the case, in the sense that for good choices of the regularizer parameter $\la$, the LAD outperforms by far the LASSO. (In the extreme of a large enough number of measurements, the LAD is consistent and this is not the case for the LASSO.) However, it is worth observing that for a different and relatively big range of values of $\la$, the LASSO performs better. This indicates the importance of the tuning of the regularizer parameter, to which the predictions of Theorem \ref{thm:master} can offer valuable guidelines and insights. 

\vp
\noindent{\underline{ Figure \ref{fig:Huber} }.}~ For this figure, we have assumed an $\ell_1$-regularized estimator with Huber-loss $\ell(v)=H_1(v)$. The noise is iid $\mathrm{Cauchy}(0,1)$. In Section \ref{sec:Hl} it was shown that all the Assumptions of Theorem \ref{thm:sep} are satisfied in this setting. The figure, validates the prediction. To obtain the prediction we numerically solved the corresponding system of nonlinear equations (see \eqref{eq:4eq}) using the efficient iterative scheme described in Remark \ref{rem:num}.

\vp
\noindent{\underline{ Figure \ref{fig:group} }.}~ We include this as an example of an M-estimator with non-separable loss function. For the plot, we use the square-root LASSO with $\ell_{1,2}$-regularization. The analytical prediction was derived solving \eqref{eq:square-root_7}.

% Hence, the theoretical predictions are derived by numerically solving the corresponding system of nonlinear equations using the efficient iterative scheme described in Section \ref{sec:sep}. 

\section{Proof Highlights }\label{sec:outline}

%In this section, we outline some key ideas behind the proof of Theorem \ref{thm:master}. The main ingredient of the proof is what we call the \emph{Convex Gaussian Min-max Theorem} (CGMT). The CGMT might be of individual interest and may have applications that go beyond the topic of this paper. With this in mind, we first  present the CGMT in Section \ref{sec:CGMT} its most general version. In the upcoming sections, we specialize to the problem of error prediction of M-estimators, we show how the CGMT is applicable there, and, we highlight the most basic ideas. The detailed proof is included in Appendix \ref{sec:proof}.

Here, we outline the fundamental ideas behind the proof of Theorem \ref{thm:master}. The detailed proof is deferred to Appendix \ref{sec:proof}. Leaving some technical challenges aside, the mechanics are easy to explain and provide valuable intuition regarding both the assumptions required and the flavor of the final result. For instance, we will be able to show without much effort, how the Moreau envelope functions $\env{\loss}{c\g+\z}{\tau}$ and $\env{f}{c\h+\x_0}{\tau}$ appear in the final result. A key ingredient of the proof is the Convex Gaussian Min-max Theorem (CGMT). We include the statement  of the theorem in this section, as well.

% ----------------------------------------------------------- %
\subsection{Starting Idea}\label{sec:CGMT}
% ----------------------------------------------------------- %
Our goal is to characterize the nontrivial limiting behavior of $\|\hat\x-\x_0\|_2$, where $\hat\x$ is any solution to the following minimization,
\begin{align}\nn
 \min_\x \loss(\y-\A\x) + \la f(\x).
\end{align}
To get a direct handle on the error term, it is convenient to change the optimization variable to $\w:=\x-\x_0$, so then $\hat\w:=\hat\x-\x_0$ is a solution to (recall $\y=\A\x_0+\z$)
\begin{align}\label{eq:mini}
\hat\w :=\arg\min_\w \loss(\z-\A\w) + \la f(\x_0+\w)=: M(\w).
\end{align}

There is a simple but standard argument that is in the heart of  most analyses of such minimization estimators, and comes as follows. Suppose we knew that the error $\|\hat\w\|_2$ converges eventually to some deterministic value, call it $\alpha_*$. This is equivalent to $\hat\w$ belonging in the following set 
\begin{align}\label{eq:the_set}
\Sc_\eps=\{ \w~|~ | \|\w\|_2 - \alpha_* |<\eps \},
\end{align}
with probability one (w.p.1) for all $\eps>0$. Letting $\Sc_\eps^c$ denote the complement of that set, observe, that if w.p. 1,
\begin{align}\label{eq:simple}
M(\hat\w) < \inf_{\w\in\Sc_\eps^c} M(\w),
\end{align}
then $\hat\w$ must lie in $\Sc_\eps$. Note that with this standard trick we have translated a question on the optimal solution of the minimization problem in \eqref{eq:mini} to one regarding its optimal cost. One possible approach in comparing the two random processes in \eqref{eq:simple} would be to first identify the converging limits of both. If say \begin{align}\label{eq:comp54}
M(\hat\w)\rP \overline{M}\qquad\text{ and }\qquad\inf_{\w\in\Sc_\eps^c} M(\w) \rP \overline{M}_{\Sc_\eps^c},
\end{align}
then, \eqref{eq:simple} holds as long as 
\begin{align}\label{eq:comp55}
\overline{M}<\overline{M}_{\Sc_\eps^c},
\end{align} 
which is just a comparison between two deterministic quantities. 

This is exactly the approach we want to take here: show \eqref{eq:comp54} and \eqref{eq:comp55}. Unfortunately, directly working with the objective function $M$ and proving \eqref{eq:comp54} turns out to be rather challenging. Instead, we prove the desired indirectly, via working with an \emph{auxiliary} objective function which is simpler to analyze. What justifies this idea is the \emph{Convex Gaussian min-max Theorem} (CGMT), which we present next.

% ----------------------------------------------------------- %
\subsection{The CGMT}\label{sec:CGMT}
% ----------------------------------------------------------- %
The Convex Gaussian Min-max Theorem associates with a primary optimization (PO) problem a simplified auxiliary optimization
(AO) problem from which we can tightly infer properties of the original (PO), such as the
optimal cost, the optimal solution, etc.. 
% =============================================== %

Specifically, the (PO) and (AO) optimizations are given as follows:
\begin{subequations}\label{eq:POAO}
\begin{align}
\label{eq:PO_gen}
\Phi(\G)&:= \min_{\w\in\Sc_\w}~\max_{\ub\in\Sc_\ub}~ \ub^T\G\w + \psiubw,\\
\label{eq:AO_gen}
\phi(\g,\h)&:= \min_{\w\in\Sc_\w}~\max_{\ub\in\Sc_\ub}~ \|\w\|_2\g^T\ub + \|\ub\|_2\h^T\w + \psiubw,
\end{align}
\end{subequations}
where $\G\in\R^{m\times n}, \g\in\R^m, \h\in\R^n$, $\Sc_\w\subset\R^n,\Sc_\ub\subset\R^m$ and $\psi:\R^n\times\R^m\rightarrow\R$. We denote $\w_\Phi:=\w_\Phi(\G)$ and $\w_\phi:=\w_\phi(\g,\h)$ any optimal minimizers in \eqref{eq:PO_gen} and \eqref{eq:AO_gen}, respectively. 
%We let the entries of $\G,\g$ and $\h$ be distributed iid standard normal (and independent of one another). 

Then, we have the following result.
%
%The CGMT identifies conditions  which establish tight probabilistic connections between the two optimization problems, when the entries of $\G,\g$ and $\h$ be distributed iid standard normal (and independent of one another). 
%
\begin{thm}[CGMT]\label{thm:CGMT}
In \eqref{eq:POAO}, let $\Sc_\w,\Sc_\ub$ be compact sets, $\psi$ be continuous on $\Sc_\w\times\Sc_\ub$, and, $\G,\g$ and $\h$ all have entries iid standard normal. The following statements are true:
\begin{enumerate}[(i)]
\item For all $c\in\R$: 
$$\Pro(~\Phi(\G) < c~) \leq 2\Pro(~\phi(\g,\h) \leq c~).$$
%\vp
\item Further assume that $\Sc_\w,\Sc_\ub$ are convex sets and $\psi$ is convex-concave on $\Sc_\w\times\Sc_\ub$. Then, for all $c\in\R$, $$\Pro(~\Phi(\G) > c~) \leq 2\Pro(~\phi(\g,\h) \geq c~).$$

In particular, for all $\mu\in\R, t>0$, $\Pro(~|\Phi(\G)-\mu|>t~) \leq 2\Pro(~|\phi(\g,\h)-\mu|\geq t~)$.
%\vp
\item
Let $\Sc$ be an arbitrary open subset of $\Sc_\w$ and $\Scc=\Sc_\w/\Sc$ . Denote $\Phi_\Scc(\G)$ and $\phi_\Scc(\g,\h)$ the optimal costs of the optimizations in \eqref{eq:PO_gen} and \eqref{eq:AO_gen}, respectively, when the minimization over $\w$ is now constrained over $\w\in\Scc$. 
 If there exist constants $\phio, ~\phio_\Scc$ and $\eta>0$ such that
%, such that each one of the following events occurs with probability at least $1-p$:
\begin{enumerate}[(a)]
\item  $\phio_\Scc \geq \phio + 3\eta$,
\item $\phi(\g,\h) < \phio + \eta$ with probability at least  $1-p$,
\item $\phi_\Scc(\g,\h) > \phio_{\Scc} - \eta $ with probability at least  $1-p$,
\end{enumerate}
then,
$$
%\Pro( \w_\phi(\g,\h) \in \Sc ) \geq 1-2p\quad\text{ and }\quad
\Pro( \w_\Phi(\G) \in \Sc ) \geq 1-4p.
$$
\end{enumerate}
\end{thm}
%
 %Also, the first statement is true for any (open) set $\Sc$. In particular, choosing $\Sc=\emptyset$
%

The CGMT is an extension of a Gaussian comparison inequality proved by Gordon in 1988 \cite{GorLem,GorThm}. Starting with the works of Rudelson and Vershynin\cite{rudelson2006sparse} and of Stojnic \cite{Sto}, Gordon's original theorem has played a key role in the analysis of (underdetermined) noiseless linear inverse problems (also, \cite{oymRank,Cha}).  We refer the interested reader to \cite{COLT15} (also, Remark \ref{rem:CGMT})  for more details and a discussion on the relation of the CGMT to the result by Gordon.

%elaborate discussion on the role of Gordon's original theorem in the analysis of (underdetermined) linear inverse problems
 The first two statements of Theorem \ref{thm:CGMT} are identical to \cite[Thm.~3]{COLT15}, and, a proof is included therein. Statement (iii) as it appears here is novel. In particular, when compared to its counterpart in \cite[Thm.~3]{COLT15}, it holds for all problem dimensions $m,n$, and also, it holds for more general sets $\Sc$. We present a proof of the last statement of the theorem in Appendix \ref{sec:thm_app}.   %  In the The 
\begin{cor}[Asymptotic CGMT]\label{cor:asy_CGMT}
Using the same notation as in Theorem \ref{thm:CGMT}, suppose there exists constants $\phio<\phio_\Scc$ such that $\phi(\g,\h) \rP \phio$ and $\phi_\Scc(\g,\h) \rP \phio_{\Scc}$. Then, 
$$
\lim_{n\rightarrow\infty} \Pro( \w_\Phi(\G) \in \Sc ) = 1.
$$
\end{cor}

\begin{remark}
Observe that the conditions of the corollary are the same as those in \eqref{eq:comp54}-\eqref{eq:comp55} only this time they hold for the objective function of the (AO). From that, we already know that $\w_\phi(\g,\h)\in\Sc$ with probability approaching 1. The statement of the corollary is stronger in that it concludes the same for $\w_\Phi(\G)$, which is the solution to a seemingly different optimization problem.
\end{remark}

The CGMT might be of individual interest and may have applications that go beyond the topic of this paper. With this in mind, we have chosen to present it above in its most general version. In the upcoming sections we specialize the result to the study of the error performance of M-estimators.

\subsection{Applying the CGMT}
Back to the problem of analyzing \eqref{eq:mini} and our goal of proving \eqref{eq:comp54}. As already hinted, the CGMT will be handy towards this direction. The M-estimator optimization in \eqref{eq:mini} will play the role of the (PO), and, we need to identify the corresponding (AO). To do so, we first need to birng \eqref{eq:mini} in the form of \eqref{eq:PO_gen} as required by the CGMT. 

The idea here is to use duality\footnote{A preliminary  version of this idea first appeared in \cite{ell22}, in which the authors analyzed the error performance of the Generalized-LASSO. We have extended the idea here to apply to any convex loss function $\loss$.}. Specifically, we can equivalently view the minimization in \eqref{eq:mini} as follows:
$$
\min_{\w,\vb} \loss(\vb) + \la f(\x_0+\w)\quad\text{sub.to}\quad \vb=\z-\A\w.
$$
Then, associating a dual variable $\ub$ with the equality constraint above, we have
\begin{align}
\min_{\w,\vb} \max_{\ub} \ub^T\A\w \underbrace{-\ub^T\z + \ub^T\vb +  \loss(\vb) + \la f(\x_0+\w)}_{\psi(\w,\vb,\ub)}.
\end{align}
Clearly, this is now in the desired format: we can identify the bilinear form $\ub^T\A\w$ and a function $\psi(\w,\vb,\ub)$ which is convex in $(\w,\vb)$ and concave in $\ub$. Thus, immediately, the corresponding (AO) problem becomes\footnote{When compared to \eqref{eq:AO_gen} it is more convenient in \eqref{eq:AO_pp} to write the two terms $\|\w\|\g^T\ub$ and $\|\ub\|\h^T\w$ with a minus sign instead. We can do this, since $\g$ and $\h$ are Gaussian vectors; thus, their distribution is sign independent.}:
\begin{align}\label{eq:AO_pp}
\min_{\w,\vb} \max_{\ub} -\|\w\|_2\g^T\ub - \|\ub\|_2\h^T\w -\ub^T\z + \ub^T\vb +  \loss(\vb) + \la f(\x_0+\w).
\end{align}
Now that we have identified the (AO) problem, we wish to apply Corollary \ref{cor:asy_CGMT} for the set $\Sc_\eps$ of \eqref{eq:the_set}. Applying the corollary amounts to analyzing the convergence of the (AO) problem (and that of its ``restricted" counterpart). This will be performed in two stages. The first involves a deterministic analysis, in which the optimization in \eqref{eq:AO_pp} is simplified and reduced to one which only involves scalar random variables. In the second stage, we analyze the convergence properties of this scalar optimization. 

Before proceeding with those, in all the above, we have been silent regarding any compactness requirements of Theorem \ref{thm:CGMT}. These technicalities are carefully handled in Appendix \ref{sec:proof}. (In particular, this is where Assumption \ref{ass:tech}(b) becomes useful.)

\subsection{Analysis of the Auxiliary Optimization}

\subsubsection{Scalarization}\label{sec:scal_illu}
A key idea that facilitates the analysis of the (AO) in \eqref{eq:AO_pp} is to reduce the optimization into one that only involves scalar optimization variables. 
The objective function of the (AO) is tailored towards this direction, and the only modification required is to express $f(\x_0+\w)$ via its variational form as $\sup_{\s}\s^T(\x_0+\w)-f^*(\s)$, where $f^*$ is the Fenchel conjugate function. 

This way, the variables $\ub$ and $\w$ appear in the objective only through either linear terms or through their magnitudes. This observation suggests that one can easily optimize over their directions while fixing the magnitudes. To illustrate this, fixing the magnitude of $\ub$ as $\|\ub\|_2=\beta\geq 0$, we can optimize over its direction by aligning it with $-\|\w\|_2\g - \z + \vb$. Then \eqref{eq:AO_pp} simplifies to the following,
\begin{align}\label{eq:not}
\min_{\w,\vb} \max_{\beta\geq0,\s}~ \beta \|\|\w\|_2\g+\z+\vb\|_2 - \beta\h^T\w+ \loss(\vb) + \la \s^T(\x_0 + \w) - f^*(s).
\end{align}
Suppose we could switch the order of min-max above. Then, it would be possible to do the same trick with $\w$, i.e. fix $\|\w\|_2=\alpha\geq 0$ and minimize over its direction to  get 
\begin{align}\label{eq:not2}
 \max_{\beta\geq0,\s}~\min_{\alpha\geq0,\vb}~ \beta \|\alpha\g+\z+\vb\|_2 + \loss(\vb) - \alpha\|\beta\h-\la\s\|_2 + \la \s^T\x_0- f^*(s).
\end{align}
Justifying that flipping in the order of min-max is not straightforward, since the objective function in \eqref{eq:not} is \emph{not} convex-concave; thus, what would otherwise be the arguments to be called upon, namely the Minimax Theorems (e.g. \cite{Sion}), are not directly applicable here. Yet, in Appendix \ref{sec:proof}, we show that such a minimax property holds asymptotically in the problem dimensions; thus, \eqref{eq:not2} is (for our purposes) equivalent to \eqref{eq:not}. We leave the details aside for the moment, and, proceed with the simplification of \eqref{eq:not2}. 

In \eqref{eq:not2}, we have reduced the optimization over $\w$ and $\ub$ to scalars $\alpha$ and $\beta$. Next, we wish to simplify the optimization over $\s$ and $\vb$. However, the same trick as the one we applied for the former two variables won't work. The new idea that we need here is to write the terms $\|\|\w\|_2\g+\z+\vb\|_2$ and $\|\beta\h-\la\s\|_2$ using
$$
\|\tb\|_2 = \inf_{\tau>0}~\frac{\tau}{2}+\frac{\|\tb\|_2^2}{2\tau}.
$$
What we achieve with this is that the corresponding terms become now \emph{separable} over the entries of the vectors $\vb$ and $\s$, which makes the optimization over them easier. The only price we have to pay is introducing just two more scalar optimization variables. That is \eqref{eq:not2} becomes
\begin{align}\nn%\label{eq:not2}
 \sup_{\substack{\beta\geq0 \\ \tauh>0}}~\inf_{\substack{\alpha\geq0 \\ \taug>0}}~ \frac{\beta\taug}{2} + \min_\vb \left\{ \frac{\beta}{2\taug} \|\alpha\g+\z+\vb\|_2^2 + \loss(\vb) \right\} - \frac{\alpha\tauh}{2} - \min_{\s}\left\{ \frac{\alpha}{2\tauh}\|\beta\h-\la\s\|_2^2 - \la \s^T\x_0 + f^*(s)\right\}.
\end{align}
It can be readily seen that the minimization over $\vb$ gives rise to the Moreau envelope function of $\loss$ evaluated at $\z+\alpha\g$ with index $\taug/\beta$. A rather straightforward completion of squares arguments and a call upon the relation between the Moreau envelopes of conjugate pairs, leads to a similar conclusion regarding the minimization over $\s$, as well. Deferring the details to the appendix, we have reached the following scalar optimization
\begin{align}\label{eq:not4}
 \sup_{\substack{\beta\geq0 \\ \tauh>0}}~\inf_{\substack{\alpha\geq0 \\ \taug>0}}~ \frac{\beta\taug}{2} + \env{\loss}{\alpha\g+\z}{\frac{\taug}{\beta}} - \frac{\alpha\tauh}{2} - \frac{\alpha\beta^2}{2\tauh}\|\h\|_2^2 + \la\cdot\env{f}{\frac{\beta\alpha}{\tauh}\h+\x_0}{\frac{\alpha\la}{\tauh}}.
\end{align}

\subsubsection{Convergence}
Once we have simplified the (AO), it is now possible to analyze the convergence of its optimal cost. We start with the objective function of \eqref{eq:not4}, which we shall denote $\Rc_n(\alpha,\taug,\beta,\tauh)$ for convenience. Fix\footnote{To be precise, an appropriate rescaling is required here. See Section \ref{sec:proof}.} $\alpha,\taug,\beta,\tauh$, then, 
\begin{align}\label{eq:ptw_cc}
\frac{1}{n}\Rc_n(\alpha,\taug,\beta,\tauh) \rP \frac{\beta\taug}{2} + L\left(\alpha,\frac{\beta}{\taug}\right) - \frac{\alpha\tauh}{2} - \frac{\alpha\beta^2}{2\tauh} + \la\cdot F\left(\frac{\alpha\beta}{\tauh},\frac{\alpha\la}{\tauh}\right)=:\Dc(\alpha,\taug,\beta,\tauh),
\end{align}
where we have assumed that $L$ and $F$ above are such that
$$
\frac{1}{n}\env{\loss}{c\g+\z}{\tau}\rP L(c,\tau)\qquad\text{and}\qquad \frac{1}{n}\env{f}{c\h+\x_0}{\tau}\rP F(c,\tau).
$$
This corresponds to Assumption \ref{ass:tech}(a), except that in the latter we have $\env{\loss}{c\g+\z}{\tau}-\loss(\z)$ instead of just $\env{\loss}{c\g+\z}{\tau}$, and, similar for $f$. The reason for this slight tweak, is to account for noise vectors $\z$ with unbounded moments. For those, $\env{\loss}{c\g+\z}{\tau}$ might not converge, but $\env{\loss}{c\g+\z}{\tau}-\loss(\z)$ will. We handle these issues in the Appendix.

Our next step is to use the \emph{point-wise} convergence of \eqref{eq:ptw_cc} in order to prove the following result:
\begin{align}\label{eq:ptw_dd}
\inf_{\substack{\alpha\geq0 \\ \taug>0}}~ \sup_{\substack{\beta\geq0 \\ \tauh>0}}~ \frac{1}{n}\Rc_n(\alpha,\taug,\beta,\tauh) \rP\inf_{\substack{\alpha\geq0 \\ \taug>0}}~ \sup_{\substack{\beta\geq0 \\ \tauh>0}}~\Dc(\alpha,\taug,\beta,\tauh)=:\phio.
\end{align}
This statement is of course much stronger than the one in \eqref{eq:ptw_cc}. The proof requires two main ingredients: (i) translating the point-wise convergence into a \emph{uniform} one over compact sets, (ii) proving that $\Dc$ is level-bounded with respect to its arguments, thus, the sets of optimizers in \eqref{eq:ptw_dd} are bounded. For the first point, convexity turns out to be critical, while the latter can be shown if Assumption \ref{ass:prop} holds.   

\subsection{Concluding}
The analysis of the (AO) problem led us to \eqref{eq:ptw_dd}. The same arguments also show that
\begin{align}\label{eq:ptw_dd_res}
\inf_{\substack{|\alpha-\alpha_*|\geq \eps \\ \taug>0}}~ \sup_{\substack{\beta\geq0 \\ \tauh>0}}~ \frac{1}{n}\Rc_n(\alpha,\taug,\beta,\tauh) \rP\inf_{\substack{|\alpha-\alpha_*|\geq \eps  \\ \taug>0}}~ \sup_{\substack{\beta\geq0 \\ \tauh>0}}~\Dc(\alpha,\taug,\beta,\tauh)=:\phi_{\Sc^c_\eps}.
\end{align}
Recall from Section \ref{sec:scal_illu} that the variable $\alpha$ plays the role of the magnitude of  $\w$, hence the random optimization in the LHS of \eqref{eq:ptw_dd_res} corresponds to the restricted (AO) problem $\phi_{\Sc^c_\eps}(\g,\h)$ of Corollary \ref{cor:asy_CGMT}. What remains for the corollary to apply is showing that $\phio_{\Sc^c_\eps}>\phio$. This follows by assumption of the theorem that the minimizer over $\alpha$ in the RHS of \eqref{eq:ptw_dd} is unique. Applying the corollary, shows the desired and concludes the proof.

\section{Prior Literature}\label{sec:prior}

In Section \ref{sec:rel} we gave a brief overview of the results most closely aligned with our work. Here, we expand on this discussion. 

\vp
\noindent{\textbf{Phase transitions}.}  The work on phase transitions of non-smooth convex optimization used to recover structured signals from noiseless linear measurements is an essential precursor for the follow-up work on the error behavior of  regularized M-estimators.  Hence, we discuss it here in some detail. This line of work attempts to characterize the minimum number of measurements, say $m_*$, as a function of the structural complexity of $\x_0$ and of the choice of $f$, such that $\x_0$ is the unique solution of the optimization $\min_{\A\x=\A\x_0} f(\x)$ with probability approacihing 1 if and only if $m>m_*$.

The early works in the field studied this question in the context of sparse signal recovery and $\ell_1$-minimization (i.e. $f(\x)=\|\x\|_1$); they showed that $\ell_1$-minimization can recover a sparse signal $\x_0$ from fewer observations than the ambient dimension $n$ \cite{candes2006near},\cite{donoho2006high,donoho2009counting}. On the one hand, 
Candes \& Tao assumed the measurement matrix $\A$ satisfies certain restricted isometry properties and provided an ``order-optimal" (with very loose constants) upper bound on $m_*$. On the other hand, when $\A$ has entries iid Gaussian, Donoho and Tanner  obtained an asymptotically precise upper bound on $m_*$, via polytope angle calculations and related ideas from combinatorial geometry. The results of Donoho and Tanner were latter extended to weighted $\ell_1$-minimization and were supplemented with robustness guarantees in \cite{xu2011precise}. However, the combinatorial geometry approach has proved hard to extend to regularizers whose set of sub-gradients is non-polyhedral (the most representative such example is nuclear-norm minimization for the low-rank recovery problem, see for example \cite{recht2011null} for some early loose performance bounds using this approach). 
 
 In early 2005, Rudelson \& Vershynin \cite{rudelson2006sparse} proposed a different approach to studying $\ell_1$-minimization
 %minimization which has since then evolved into a powerful framework. Rudelson \& Vershynin were the first to 
 that uses Gordon's Gaussian Min-max Theorem (GMT) (specifically, a corollary of it known as the ``escape through a mesh" lemma \cite{GorLem}). Stojnic  refined this approach and  obtained an empirically sharp upper bound on $m_*$ both for sparse and group-sparse vectors \cite{Sto,stoBlock}. This approach is simpler than that of Donoho \& Tanner and extends to very general settings. Oymak \& Hassibi \cite{oymRank} used it to study the low-rank recovery problem, and later, Chandrasekaran et al. \cite{Cha} developed a geometric framework and were able to analyze general structures and convex 
regularizers $f$, while clarifying the key role played in the analysis by the geometric concept of ``Gaussian width" \cite{GorLem}. See also \cite{mccoy2014sharp,Foygel} for extensions to
other signal recovery problems. 
%This line of works identifies the ``Gaussian width" as the key summary parameter that 
 
The works discussed thus far only derive upper bounds on $m_*$. Matching lower bounds that prove the asymptotic tightness of the former (known as \emph{phase-transition}) are even more recent. Bayati et. al \cite{bayati2015universality} rigorously demonstrates the phase transition phenomenon for $\ell_1$-minimization. The analysis is based  on a state evolution framework for an iterative Approximate Message Passing (AMP)  algorithm inspired by statistical physics, which was earlier introduced by Donoho et. al \cite{AMP,bayati2011dynamics}. Amelunxen et. al. \cite{TroppEdge} took a different route, using tools from conic integral geometry they established for the first time that previous results of \cite{Cha} were tight. In particular, they showed that: (a) a phase transition almost always exists for general convex regularizers $f$; (b) that it can be located exactly by computing the ``statistical dimension" (which
is very related to the ``Gaussian width", but has some extra favorable
properties); and (c) that it is possible to give accurate upper and lower bounds for the statistical dimension. Subsequently, Stojnic \cite{stojnic2013upper} combined his earlier approach that was based on Gordon's GMT with a convex duality argument and used this to prove that his earlier bounds on $\ell_1$ and $\ell_{1,2}$ were asymptotically tight. (A similar observation was also reported in \cite[Rem.~2.9]{TroppEdge}.) Stojnic's approach deserves special credit under the prism of our work, since it essentially motivated the inspired all the subsequent developments on the study of the precise reconstruction error under noisy measurements using Gaussian process methods.

\vp
\noindent{\textbf{Precise reconstruction error}.}~ As mentioned in Section \ref{sec:rel},  there is a very long list of early results on the error performance of regularized M-estimators which derive ``order-wise" bounds that involve unknown scaling constants (e.g. \cite{candes2007dantzig,Belloni,bickel,negahban2012unified, wainwright2014structured, vershynin2014estimation, banerjee2014estimation,li2015geometric} and references therein). Nevertheless, in this discussion we focus entirely on more recent results that derive precise characterizations rather than  loose bounds. Unless otherwise stated, the
literature that we describe below takes  the random measurement matrix $\A$ to have  independent
Gaussian entries (but, see Remark \ref{rem:Gaussian}). Also,  it studies the high-dimensional asymptotic regime where $m$ and $n$ grow to infinity at a proportional rate.

Chronologically, the first such results  were derived using the AMP framework by Bayati, Donoho, Maleki and Montanari \cite{DMM,montanariLASSO}. Both references consider a least-squares loss function with $\ell_1$-regularization (a.k.a. LASSO) and gaussian noise distribution: \cite{DMM} developed formal expressions for the reconstruction error at high-SNR under optimal tuning of the regularizer parameter $\la>0$; \cite{montanariLASSO} explicitly characterizeed the reconstruction error for all values of $\la$ and all values of SNR. Subsequent works \cite{malekiComplexLASSO,Armeen,donoho2013accurate} involve extensions of the results to other separable regularizers (e.g. $\ell_{1,2}$-norm). In late 2013, Donoho and Montanari \cite{montanari13} introduced an extension of the AMP framework to analyze the error performance of loss functions other than  least-squares. Their analysis applies to separable, strongly-convex and smooth loss functions, to iid signal statistics, and to iid noise statistics with bounded second moments. Donoho \& Montanari consider no regularization, hence, their analysis restricts the normalized number of measurements to $\delta = m/n>1$. Very recently, Bradic \& Chen \cite{bradic2015robustness}  built upon the framework of \cite{montanari13} and extended the analysis to sparse signal recovery and $\ell_1$-regularization, under more general (but somewhat stringent) conditions on the loss function and on the  noise and signal statistics. Our work raises the assumptions on separability, smoothness and strong convexity of the loss function, considers general convex regularizers and more general signal and noise statistics. Also, our analytic approach via the CGMT framework is somewhat more direct and potentially more powerful. The AMP framework involves two steps of analysis: (a) it analyzes the error performance of the AMP algorithm based on a state evolution framework inspired by statistical physics; (b) it shows that the AMP algorithm has the same error performance as the M-estimator. This way it concludes about the behavior of the latter. In contrast, our approach directly analyzes the error behavior of the original M-estimator. Nevertheless, we remark on the algorithmic advantage of the AMP framework which (whenever applicable) comes with a fast(er) iterative algorithm with the same error performance guarantees as the convex M-estimator. Also, the AMP framework has been used for the analysis of other problems beyond noisy signal recovery from linear measurements (see \cite{montanari2015statistical} and references therein). It remains an open and potentially interesting question to study deeper connections between the two different frameworks  of analysis, namely the CGMT and the AMP frameworks.

A different approach that uses Gaussian process methods to study the precise reconstruction error was introduced by Stojnic \cite{StoLASSO} in 2013. Stojnic considered an $\ell_1$-constrained version of the LASSO under gaussian noise distribution in the high-SNR regime. Under this setting, he was the first to note that Gordon's GMT, which had been previously used to derive only upper bounds on the error performance, could be combined with a convex duality argument to yield bounds that are tight. Shortly after, Oymak et. al. \cite{OTH13} extended Stojnic's results to the regularized case by deriving tight high-SNR bounds for the square-root LASSO with general convex regularizers. Thrampoulidis et. al. \cite{ICASSP15} performed a more careful analysis further extending the results to arbitrary values of the SNR and Thrampoulidis and Hassibi \cite{Allerton14} were the first to use the GMT approach for loss functions beyond least-squares by analyzing the Least Absolute Deviations (LAD) algorithm. This line of work (and also \cite{ell22,IRO}) eventually led to a refined, clear and extended version of Stojnic's framework in \cite{COLT15}, under the name Convex Gaussian Min-Max Theorem (CGMT) framework. Thrampoulidis et. al. note in \cite{COLT15} that the framework can in principle be applied to analyze general convex loss functions and regularizers, but at the time it was not clear how to do this in a unifying way. Our work answers this question and applies the CGMT framework to regularized M-estimators with general loss and regularizer functions, and general signal and noise statistics. Also, it offers a strengthened version of the CGMT (cf. Theorem \ref{thm:CGMT}), which allows the study of performance measures beyonds the mean square error. 

Finally, a third approach to analyze the mean-squared error performance of high-dimensional M-estimators has been undertaken by El Karoui in \cite{karoui13,karoui15}. El Karoui uses leave-one-out and martingale ideas from statistics and ideas from random matrix theory to accurately predict the squared error of ridge-regularized (a.k.a. $f(\x)=\|\x\|_2^2$) M-estimators. The analysis can handle noise distributions with unbounded moments, but it requires a smooth and separable loss function. In our work, we drop both these assumptions and extend the results to general convex regularizers. In comparing the two works, we note that El Karoui's proof technique can deal with  more general assumptions  on the design matrix $\A$. (Nevertheless, please see Remark \ref{rem:Gaussian}). Beyond matrices with iid entries, El Karoui \cite{karoui15} further considers elliptical models. Even though we do not explicitely consider such an extension in the current paper, our proof technique is readily applicable to this more general scenario. Please also refer to the short discussion at the end of Section \ref{sec:conc}.

\begin{remark}[On Universality]\label{rem:Gaussian}
Since the works \cite{Sto,Cha,TroppEdge} we now have a very clear understanding of
the phase transitions of non-smooth convex signal recovery methods with iid Gaussian measurements. Under the same measurement model, the current paper extends this clear picture to the noisy setting by precisely characterizing the reconstruction error. Here, we briefly discuss relevant results that prove the universal behavior of iid Gaussian measurements over a wider class of distributions.  

Bayati et. al \cite{bayati2015universality} has rigorously demonstrated that the phase transition of $\ell_1$-minimization is universal over a wider class of iid random measurement matrices. See also \cite{donoho2009observed,korada2011applications}. Very recently, Oymak \& Tropp \cite{TroppUniversal} have significantly extended the universality result of phase-transitions to general convex regularizers and to very general distributions on the entries of $\A$ (see \cite[Prop.~5.1]{TroppUniversal}, for an exact statement). \cite{TroppUniversal} also yields conclusions for the noisy setting: it proves the universality of the error bounds of \cite{OTH13} for the constrained LASSO. It remains an open challenge to extend these results to the general setting of arbitrary loss and regularizer functions of the current paper. We remark that the results of \cite{TroppUniversal} use some of the ideas that were developed in \cite{OTH13,COLT15} and in the current paper. Also, note that the results of El Karoui \cite{karoui15} on the ridge-regularized M-estimators hold  for matrices with iid entries beyond Gaussian. 

From this discussion we have excluded random measurement models beyond ones with iid entries. An important example includes design matrices with orthogonal rows, e.g. Isotropically Random Orthogonal (IRO) matrices, randomly subsampled Fourier and Hadamard matrices, etc.. While the universality of phase transition appears to extend to such designs, this is not the case for the reconstruction error. Thrampoulidis \& Hassibi \cite{IRO} have proved that the error behavior of the LASSO is different for IRO and for Gaussian matrices. The same is true for the elliptical model considered by El Karoui in \cite{karoui15}.
\end{remark}

\begin{remark}[Heuristic results]
In parallel to the works referenced above, there have been a number of works that studied the same questions mixing heuristic-based arguments and extended simulations.  For example, \cite{guo2009single,kabashima2010statistical,rangan2009asymptotic,vehkapera2014analysis} use  the replica method from
statistical physics, which provides a powerful tool
for tackling hard analytical problems, but
still lacks mathematical rigor in some parts. Closer to the setting of our work,  the high-dimensional error performance of regularized M-estimators has been previously considered via heuristic arguments and simulations in \cite{el2013robust,bean2013optimal}. In particular,  Bean et. al. \cite{bean2013optimal} shows that maximum likelihood estimators are in general inefficient in high-dimension
and initiate the study of optimal loss functions. It is worth revisiting and extending those results in connection to the mathematically rigorous approach of the current paper. 
\end{remark}
\section{Conclusions and Future work}\label{sec:conc}

Theorem \ref{thm:master} predicts the squared error performance of general regularized M-estimators in the presence of noisy linear Gaussian measurements. The analysis is performed in the high-dimensional regime where both the number of measurements and the dimension of the signal grow large at a proportional rate. The theorem identifies the precise dependence of the error performance on the problem parameters, namely, the loss function $\loss$, the regularizer $f$, the noise and signal distributions $p_\z$ and $p_{\x_0}$, the value of the regularizer parameter $\delta$ and the normalized number of measurement $\delta$. 

We envision several interesting directions in which the results of this paper can operate as a starting point for future work, which we shall discuss next.

\vp
\noindent{\textbf{Other instances}.}~
In Section \ref{sec:sim}, all existing results in the literature were derived as special cases of Theorem \ref{thm:master} and several novel instances of M-estimators were also analyzed. The list of examples that was presented is far from being exhaustive; depending on the application in mind other choices of loss functions, regularizers, noise distributions might be of interest. As long as those satisfy the mild assumptions of Theorems \ref{thm:master} or \ref{thm:sep}, they can be analyzed using those. To give an example that was not treated here and might be of interest in applications, is evaluating the performance of an $\ell_\infty$-loss function in the presence of bounded noise.

\vp
\noindent{\textbf{Optimal tuning}.}~ Regularized M-estimators have been widely used in practice and a remaining challenging issue is that of optimally tuning the regularizer parameter $\la$. Theorem \ref{thm:master} establishes the precise dependence of the error performance on $\la$. Hence, in principle, it can be used to provide valuable insights and guidelines regarding its optimal choice. In Section \ref{sec:num} and Figure \ref{fig:LAD_sparse} we saw an example that highlights the importance of being able to choose $\la$ in the correct range of values, otherwise the performance can be significantly deteriorated.

\vp
\noindent{\textbf{Comparing performances}.}~
Theorem \ref{thm:master} can be used to evaluate the performance of general M-estimators under different settings. Figure \ref{fig:LAD_sparse} serves as a preliminary numerical illustration: under the specific setting, LAD outperforms the LASSO for appropriate choices of $\la$. The error expressions of Theorem \ref{thm:master} will allow quantifying such comparisons and yield \emph{analytic} such conclusions.

\vp
\noindent{\textbf{Optimal loss/regularizer functions}.}~ One of the most exciting (at the same time challenging) potential applications of the results of this paper is identifying optimal choices for the loss and regularizer functions under different settings. Since the error characterization differs from the corresponding results of classical statistics (where the signal dimension is fixed), we expect new phenomena to arise and the answers to differ in general. 
When it comes to the regularizer, the optimality question has been partially considered in the literature. When the structured signal $\x_0$ is considered  \emph{fixed}, then a good choice for the regularizer $f$ is one that minimizes the statistical dimension of the tangent cone of $f$ at $\x_0$ (cf. Section \ref{sec:cone}) \cite{Cha,TroppEdge,OTH13}\footnote{Based on this, Chandrasekaran et. al. have suggested the notion of ``atomic-norms" as a principled way for constructing appropriate convex regularizer functions for different kind of structures \cite{Cha}.}. The results of \cite{Cha} and \cite{TroppEdge} combined prove that  this is indeed the optimal choice in the noiseless case. The same is true  in the high-SNR regime when a least-squares loss function is used as shown in \cite{OTH13,ell22}. The more general setting of the current paper, will allow revisiting this question and extending the results to capture instances where $\x_0$ is associated with a prior distribution $p_{\x_0}$, the loss function differs from a least-squares one, and, the noise variance is not necessarily tending to zero. Theorem \ref{thm:master} suggests that the quantity that will be involved in the optimization is the Expected Moreau envelope, which is in fact a generalization of the statistical dimension (cf. Section \ref{sec:cone}). When it comes to the optimal choice of the loss function with respect to the noise distribution $p_\z$,  less is known. Again, the expected Moreau envelope will be central in the optimization, but is yet to be understood how this will translate into practical recipes for the design of optimal loss functions.

\vp
\noindent{\textbf{Consistency}.}~ Another important question that is also related to the optimal choice of loss/regularizer functions, asks for conditions under which the squared error is zero, if at all this is possible. In Remark \ref{rem:LAD_zero}, we discussed an example of a an M-estimator that under specific noise and signal distributions, becomes consistent provided that the normalized number of measurements is large enough and that the regularizer parameter is chosen on the correct range (also, see Figure  \ref{fig:LAD_sparse}). Answering questions regarding consistency, amounts to  identifying conditions under which $\alpha_*=0$ can be the optimal solution to the (SPO) of Theorem \ref{thm:master}.

%\vp
%\noindent{\textbf{Bounded error}.}~ 

\vp
\noindent{\textbf{Beyond squared error}.}~ The emphasis in this work has been on characterizing the squared error $\|\hat\x-\x_0\|_2^2$ of regularized M-estimators. This appears commonly in practice, but, depending on the application, other performance metrics might be more appropriate. A few representative examples might include $\|\hat\x-\x_0\|_p$,  $\sum_{i=1}^n \mathbf{1}_{\left\{\hat\x_i =0,{\x}_{0,i}\neq 0\right\}}$\footnote{This metric is known as the subset or variable selection criterion and measures the success in the recovery of the subset of nonzero indices of $\x_0$\cite{wainwright2009sharp}. }, $\sum_{i=1}^n \mathbf{1}_{\left\{\hat\x_i <0\right\}}$\footnote{This metric is appropriate in (say) the following wireless communications setting (see \cite{ICASSP16} for details). Assume that $\x_0\in\{\pm1\}^{n}$ is a BPSK signal. A popular algorithm for recovering $\x_0$, called the box-relaxation, produces an estimate $\tilde\x=\mathrm{sign}\left(\arg\min_{\x_i\in[-1,1]}\|\y-\A\x\|_2\right)$. Assuming wlog that $\x_0=\{+1\}^n$, then $\frac{1}{n}\sum_{i=1}^n \mathbf{1}_{\left\{\hat\x_i <0\right\}}$ measures the empirical \emph{probability of error} of the scheme.}, etc.. The principles and mechanics of this paper can be used to derive characterizations for those metrics, as well. In particular, the same key idea, that of analyzing an Auxiliary optimization problem instead of the original (PO), is applicable. Note that in the statement of the CGMT Theorem \ref{thm:CGMT}(iii), there is nothing constraining the set $\Sc$ to be chosen. Here, we chose $\Sc=\{\w~|~|\|\w\|_2-\alpha_*|<\eps\}$. If  we were interested in (say) $\sum_{i=1}^n \mathbf{1}_{\left\{\hat\x_i <0\right\}}$, then, it would be appropriate to  apply the theorem for a different set, namely $\Sc=\{\w~|~|\frac{1}{n}\sum_{i=1}^n \mathbf{1}_{\left\{\hat\x_i <0\right\}}-\alpha_*|<\eps\}$ (see \cite{ICASSP16}).

%To see this note that
% main tool underlying the analysis, the CGMT Theorem \ref{thm:CGMT}, 

\vp
\noindent{\textbf{Beyond Gaussian Designs}.}~ Theorem \ref{thm:master} assumes that the entries of the design matrix  $\A$ are iid Gaussian. Yet, there are potentials of extending the results to other classes of distributions as discussed next.

\noindent\emph{Matrices with iid entries}.~ 
Preliminary numerical results (Figure \ref{fig:bern} is an example) suggest a \emph{universality} property of the prediction of Theorem \ref{thm:master} to design matrices with entries iid drawn from a wider
class of probability distributions, e.g. sub-gaussians. Besides simulation results, it is worth mentioning that El Karoui proves this to be the case for M-estimators with ridge-regulararization and a twice differentiable loss function \cite{karoui15}.
%
% remains valid for design matrices with entries iid following distributions other than the Gaussian (e.g. Bernouli, sub-gaussian etc.).
%  Such a phenomenon is often referred to as the \emph{universality property} of the Gaussian case in the statistics and random matrix theory literature.
%It has been the case for numerous results in the random matrix theory, and, in statistics, that they were first proven under gaussianity assumptions, and, were subsequently shown to be valid for a broader class of distributions \cite[]. For instance, the famous semicircle law that was initially shown to apply for matrices with iid Gaussian entries, is now known to be true for matrices with iid entries of any distribution of bounded fourth moment \cite[]. This property of the gaussian case is often referred to as universality property. 
%In Section \ref{sec:sim} we have numerically examined the validity of the prediction of Theorem \ref{thm:master} for design matrices with entries iid following distributions other than the Gaussian, e.g. Bernoulli, etc.. The empirical observations suggest that 
%The outcomes of the simulations are promising in that we \emph{empirically} observe the universality property to be applicable to the setup of this work. 

%A rigorous proof of such a statement is challenging, but certainly interesting.

\noindent\emph{Isotropically Random Orthogonal (IRO) Matrices}.~ 
An IRO matrix  $\A$ is sampled uniformly at random from the manifold of row-orthogonal
matrices satisfying $\A\A^T = \I_m$. Studying the error performance of M-estimators under such designs is of practical interest\footnote{
%Matrices with orthogonal rows are often preferred
%in practice because (i) they do
%not amplify the noise (ii) 
Certain classes of orthogonal matrices
such as discrete-cosine and Hadamard allow for fast
multiplication and reduced complexity. Numerical simulations in \cite{IRO} suggest that the error prediction for IRO matrices is valid for random DCT and Hadamard matrices. 
}.
In  \cite{IRO}, we were able to extend the CGMT framework to  accurately predict the error performance of the LASSO when $\A$ is IRO and $\z$ is iid gaussian. 
%Briefly, our result proved the noise performance of the (IRO) matrices being superior to that of an iid Gaussian matrix.
 Extending those ideas to general M-estimators in a flavor similar to the setting of this paper is a possible direction for future research.
 
% 
% It might be an interesting direction of future work to try combine the ideas of \cite{IRO} with the current work, in order to characterize the error performance of general M-estimators beyond the LASSO when $\A$ is (IRO) (or random Fourier/Hadamard). 

\noindent\emph{Elliptical Distributions}. Assume $\G$ with entries iid Gaussian, $\eps_i$'s be independent and independent of $\G$, and, $\A = \mathrm{diag}(\eps_1,\ldots,\eps_m)\G$. We are motivated to consider such ``elliptical-like" distributions by the relevant work \cite{karoui15}. 
%Interestingly, El Karoui shows that  ``the role of the distribution of the $\eps_i$'s in the performance of the estimator
%depends on much more than its second moment". 
It is rather straightforward how to extend the CGMT framework, and consequently the prediction of Theorem \ref{thm:master}, to account for such a class of distributions. We might consider explaining the details in future work.

\section*{Acknowledgement}
The authors would like to thank George Moustakides, Joel Tropp, P. P. Vaidyanathan and Panagiotis Vergados  for helpful conversations and suggestions. Christos Thrampoulidis would also like to thank Ashkan Panahi and Linqi (Daniel) Guo; some of the ideas that led to this work were born in collaboration with them, cf. \cite{ell22,ICASSP15}.
%

%\newpage
%\newpage
%\input{lemma}t

\bibliography{compbib}

\newpage
\appendix
\section{Proof of Theorem \ref{thm:master}}\label{sec:proof}% \ref{thm:master}}\label{sec:proof}

Here, we prove Theorem \ref{thm:master}. The proof consists of several steps and intermediate results, that are stated as Lemmas. The proofs of the latter are all deferred to Appendix \ref{sec:thm_app}.

\subsection{Preliminaries}
%To enlighten notation, we often d

\begin{align}\nn
\hat\x := \arg\min_\x \loss(\y-\A\x) + \la f(\x).
\end{align}

Recall that $\y=\A\x_0+\z$. Our goal is to characterize the nontrivial limiting behavior of $\|\hat\x-\x_0\|_2/\sqrt{n}$. We start with a simple  change of variables $\w:=(\x-\x_0)/\sqrt{n}$, to directly get a handle on the \emph{error vector} $\w$. Also, we normalize the objective by dividing with ${n}$ so that the optimal cost is of constant order. Then,

\begin{align}\label{eq:PO}
\hat\w := \arg\min_\w \frac{1}{{n}}\left\{\loss\left(\z-\sqrt{n}\A\w\right) + \la f\left(\x_0+\sqrt{n}\w\right)\right\}.
\end{align}
Instead of the optimization problem above, we will analyze a simpler Auxiliary Optimization (AO) that is tightly related to the Primary Optimization (PO) in \eqref{eq:PO} via the CGMT.

% =============================================== %
\subsection{The CGMT for M-estimators}
% =============================================== %

%In this section, we show how to write the minimization in \eqeref{eq:PO} in a form so that the CGMT Theorem \ref{thm:main} is applicable. In essence this translates to expressing \eqref{eq:PO} in the form of a (PO) as suggested

In this section, we show how the CGMT Theorem \ref{thm:CGMT} can be applied to predict the limiting behavior of the solution $\|\hat\w\|_2$ to the minimization in \eqref{eq:PO}. The main challenge here is to express \eqref{eq:PO} as a (convex-concave) minimax optimization in which the involved random matrix (here $\A$) appears in a bilinear form, exactly as in \eqref{eq:PO_gen}.
% As we will see, the idea behind this is based on the use of Lagrange duality. 
 Also, some side technical details need to be taken care of. For example, in \eqref{eq:PO_gen} the optimization constraints are required by Theorem \ref{thm:CGMT} to be bounded, which is not the case with \eqref{eq:PO}. We start with addressing this immediately next.

\subsubsection{Boundedness of the Error}\label{sec:B}

The constraint set over which $\w$ is optimized in \eqref{eq:PO_gen} is unbounded. We will introduce ``artificial" boundedness constraints that allow applying Theorem \ref{thm:CGMT}, while they do not affect the optimization itself. For this purpose, recall our goal of proving that $\|\hat\w\|_2$ converges to some (finite) $\alpha_*$ defined in Theorem \ref{thm:master}. Define the set $\Sc_\w=\{ \w ~|~ \|\w\|_2\leq K_\alpha \}$, where
\begin{align}\label{eq:K_a}
K_\alpha:=\alpha_*+ \zeta
\end{align} for a constant $\zeta>0$, and, consider the ``bounded" version of \eqref{eq:PO}:
\begin{align}\label{eq:PO_B}
\hat\w^{B} := \arg\min_{\w\in\Sc_\w} \frac{1}{{n}}\left\{ \loss\left(\z-\sqrt{n}\A\w\right) + \la f\left(\x_0+\sqrt{n}\w\right)\right\}.
\end{align}
We expect that the additional constraint $\w\in\Sc_\w$ in \eqref{eq:PO_B} will not affect the optimization with high probability when $n$ is large enough. The idea here is that the minimizer of the original unconstrained problem in \eqref{eq:PO} satisfies $\|\hat\w\|_2\approx\alpha_*<K_\alpha$ w.h.p.. Of course, this latter statement is yet to be proven! 
Once this is done, we can return and confirm that our initial expectation is met. 
%With $\w$ being bounded in \eqref{eq:PO3}, it can be shown using standard optimality conditions that the optimal value of $\ub^*$ are also bounded. 
%In a similar manner we can argue about $\ub$. 
Lemma \ref{lem:PO_B} below shows that if $\|\hat\w^B\|\rP\alpha_*<K_\alpha$, then, the same is true for the optimal of \eqref{eq:PO}.
\begin{lem}\label{lem:PO_B}
For the two optimizations in \eqref{eq:PO} and \eqref{eq:PO_B}, let $\hat\w$ and $\hat\w^B$ be optimal solutions. Also, recall the definition of $K_\alpha$ in \eqref{eq:K_a}. If  $\|\hat\w^B\|\rP\alpha_*$, then $\|\hat\w\|\rP\alpha_*$.
\end{lem}

% % --------------------- MOVED TO APPENDIX ------------------------------------------
%\begin{proof}
%For convenience, denote with $F(\w)$ the objective function in \eqref{eq:PO}. For some $\eps>0$ such that $\alpha+\eps<K_\alpha$ (e.g. $\eps=\zeta/2$ in \eqref{eq:K_a}), denote $\Dc:=\{\w~|~\alpha-\eps\leq \|\w\|_2\leq \alpha+\eps\}$. By assumption, with probability approaching 1 (w.p.a. 1).
%\begin{align}\label{eq:wBass}
%\hat\w^B\in\Dc.
%\end{align}
%For the shake of a contradiction, assume that there exists optimal solution $\hat\w$ of \eqref{eq:PO} such that $\hat\w\not\in\Dc$ w.p.a. 1. Clearly, 
%\begin{align}\label{eq:clear121}
%F(\hat\w)\leq F(\hat\w^B).
%\end{align}
%Suppose $\hat\w\in\Sc_\w$, then $\hat\w$ is optimal for \eqref{eq:PO_B} and satisfies \eqref{eq:wBass}, which contradicts our assumption. Thus, $\hat\w\not\in\Sc_\w$. 
%
%Next, let $\w_\theta:=\theta \hat\w + (1-\theta) \hat\w^B$ for $\theta\in(0,1)$ such that $\w_\theta\not\in\Dc$ and $\w_\theta\in\Sc_w$ (always possible, by definition of $\Dc$). By the convexity of $F$ and \eqref{eq:clear121}, it follows that
%$F(\hat\w_\theta) \leq F(\hat\w^B)$. Hence, $\hat\w_\theta$ is optimal for \eqref{eq:PO_B} and satisfies \eqref{eq:wBass}, which, again, is a contradiction. This completes the proof.
%\end{proof}
%
% % --------------------- MOVED TO APPENDIX ------------------------------------------

Owing to the result of the lemma, henceforth, we work with the bounded optimization in \eqref{eq:PO_B}. Using some abuse of notation, we will refer to optimal solution of \eqref{eq:PO_B} as $\hat\w$, rather than $\hat\w^B$.

\subsubsection{Identifying the (PO)}

Here, we bring the minimization in \eqref{eq:PO_B} it in the form of the (PO) in \eqref{eq:PO_gen}. For this purpose, we will use Lagrange duality.
Note that the former can be equivalently expressed as
\begin{align}\nn
\hat\w &= \arg\min_{\w\in\Sc_\w,\vb}\frac{1}{{n}}\left\{ \loss(\sqrt{n}\vb) + \la f(\x_0+\sqrt{n}\w)\right\} \quad\text{subject~to}\quad \vb={\z}-{\sqrt{n}}\A\w.%\label{eq:Slater}
\end{align}
Associating a dual variable $\ub$ to the equality constraint above, we write it as
\begin{align}
\hat\w=\arg\min_{\w\in\Sc_\w,\vb}\max_{\ub} \frac{1}{\sqrt{n}}\left\{-\ub^T(\sqrt{n}\A)\w +\ub^T\z - \ub^T\vb\right\} + \frac{1}{{n}}\left\{ \loss(\vb) + \la f(\x_0+\sqrt{n}\w)\right\}.
\label{eq:PO2a}
\end{align}
It takes no much effort to check that the objective function above is in the desired format of \eqref{eq:PO_gen}: the random matrix $\A$ appears in a bilinear term $\ub^T\A\w$, and, the rest of the terms form a convex-concave function in $\ub,\w$. Furthermore, we can use Assumption \ref{ass:tech}(b) to show that the optimal $\ub_*$ is bounded, which is a requirement of Theorem \ref{thm:CGMT}. In the same lines as in Section \ref{sec:B}, we henceforth work with the ``bounded" version of \eqref{eq:PO2a}, namely, 
\begin{align}
\hat\w=\arg\min_{\w\in\Sc_\w,\vb}\max_{\ub\in\Sc_\ub} \frac{1}{\sqrt{n}}\left\{-\ub^T(\sqrt{n}\A)\w +\ub^T\z - \ub^T\vb\right\} + \frac{1}{{n}}\left\{ \loss(\vb) + \la f(\x_0+\sqrt{n}\w)\right\}.\label{eq:PO2}
\end{align}
for $\Sc_\ub:=\{\ub~|~\|\ub\|_2\leq K_{\beta}\}$ and $K_\beta>0$ a sufficiently large constant.

\begin{lem}\label{lem:K_b}
If Assumption \ref{ass:tech}(b) holds, then there exists sufficiently large constant $K_\beta$, such that the optimization problem in \eqref{eq:PO2} is equivalent to that in \eqref{eq:PO_B}, with probability approaching 1 in the limit of $n\rightarrow\infty$.
\end{lem}

As a last step, before writing down the corresponding (AO) problem, it will be useful for the analysis of the latter, to express $f$ in a variational form through its Fenchel conjugate, which gives,
\begin{align}
\hat\w=\arg\min_{\w\in\Sc_\w,\vb}\max_{\ub\in\Sc_{\ub},\s} \frac{1}{\sqrt{n}}\left\{-\ub^T(\sqrt{n}\A)\w + \ub^T\z - \ub^T\vb\right\} + \frac{1}{{n}}\left\{\loss(\vb) + \la \s^T\x_0+\la\sqrt{n}\s^T\w - \la f^*(\s)\right\} \label{eq:PO3}.
\end{align}

\subsubsection{The (AO)}
Having identified \eqref{eq:PO3} as the (PO) in our application, it is straightforward to write the corresponding (AO) problem following \eqref{eq:AO_gen}:
\begin{align}
\min_{\w\in\Sc_\w,\vb}~\max_{\ub\in\Sc_\ub,\s} \frac{1}{\sqrt{n}}\left\{\|\w\|_2\g^T\ub - \|\ub\|_2\h^T\w + \ub^T\z - \ub^T\vb\right\} + \frac{1}{{n}}\left\{\loss(\vb) + \la \s^T\x_0+\la\sqrt{n}\s^T\w - \la f^*(\s)\right\}\label{eq:AO3}.
\end{align}
Once we have identified the (AO) problem, Corollary \ref{cor:asy_CGMT} suggests analyzing that one instead of the (PO). Our goal is showing that $\|\hat\w\|_2\rP\alpha_*$. For this, we wish to apply the corollary to the following set 
$$
\Sc= \{ \w ~ | ~ |\|\w\|_2 - \alpha_*| > \eps\},
$$
for arbitrary $\eps>0$.

\subsubsection{Asymptotic min-max property of the (AO)}
It turns out that verifying the conditions of the corollary for the (AO) as it appears in \eqref{eq:AO3} is not directly easy. In short, what makes the analysis cumbersome is the fact that the optimization in \eqref{eq:AO3} is not convex (e.g. if $\g^T\ub$ is negative, then $\|\w\|_2\g^T\ub$ is not convex). Thus, flipping the order of min-max operations that would simplify the analysis is not directly justified. 

At this point, recall that the (PO) in \eqref{eq:PO3} is itself convex. In fact, for it,  all conditions of Sion's min-max Theorem \cite{Sion} are met, thus, the order of min-max operations can be flipped. According to the CGMT, the (PO) and the (AO) are tightly related in an asymptotic setting. We use this, to translate the convexity properties of the (PO) to the (AO). In essence, we show that when dimensions grow, the order of min-max operations in the (AO) can be flipped. Thus, we will instead consider the following problem as the (AO):

\begin{align}\label{eq:AO_t}
\phi(\g,\h):=\max_{\substack{0\leq\beta\leq K_\beta \\ \s}}~ \min_{\substack{\|\w\|_2\leq K_\alpha \\ \vb}}~ \max_{\|\ub\|_2=\beta} ~& \frac{1}{\sqrt{n}}(\|\w\|_2\g+\z-\vb)^T\ub - \frac{1}{\sqrt{n}}\|\ub\|_2\h^T\w \nn\\
&\qquad\qquad+ \frac{1}{{n}}\loss(\vb) + \frac{\la}{{n}}\s^T{\x_0}+\frac{\la}{\sqrt{n}}\s^T\w -\frac{\la}{{n}} f^*(\s).
\end{align}
Observe that the objective function remains the same; it is only the order of min-max operations that is slightly modified compared to \eqref{eq:AO3}. Since the objective function is not necessarily convex-concave in its arguments, there is no immediate guarantee that the two problems in \eqref{eq:AO3} and \eqref{eq:AO_t} are equivalent for any realizations of $\g$ and $\h$. However, the lemma below essentially shows that such a strong duality holds with high probability over $\g$ and $\h$ in high dimensions. Hence, the problem in \eqref{eq:AO_t} can be as well used, instead of the one in \eqref{eq:AO3}, in order to analyze the (PO). For this reason, henceforth, we refer to \eqref{eq:AO_t} as the (AO) problem.
%
%

% =============================================== %
%
\begin{lem}\label{lem:CGMT_M}
%Recall the M-estimator optimization in \eqref{eq:PO} with
Let $\hat\w(\A)$ denote an optimal solution of \eqref{eq:PO}. Consider the (AO) problem in \eqref{eq:AO_t}.
Let $\alpha_*$ be as defined in Theorem \ref{thm:master}. For any $\eps>0$ define the set $\Sc:=\{\w~|~|\|\w\|_2-\alpha_*|< \eps\}$, and,  $\phi_{\Sc^c}(\g,\h)$ be the optimal cost of the same optimization as in \eqref{eq:AO_t}, only this time the minimization over $\w$ is further constrained such that $\w\notin \Sc$.
Assume that for any $K_\alpha>\alpha_*$ and for any sufficiently large $K_\beta$, there exist constants $\phio<\phio_\Scc$ such that for all $\eta>0$, with probability approaching one in the limit of $n\rightarrow\infty$ the following hold:
%, such that each one of the following events occurs with probability at least $1-p$:
\begin{enumerate}[(a)]
%\item  $\phio_\Scc > \phio$,
\item $\phi(\g,\h) < \phio+\eta$,
\item $\phi_\Scc(\g,\h) > \phio_{\Scc} - \eta$.
\end{enumerate}
 Then,
 $$
 \lim_{n\rightarrow\infty}\Pro\left(~ |\|\hat\w(\A)\|_2 - \alpha_*|<\eps ~\right) = 1.
 $$
\end{lem}

After Lemma \ref{lem:CGMT_M}, what remains in order to prove Theorem \ref{thm:CGMT} is satisfying the conditions of the lemma. This involves a thorough analysis of the (AO) problem in \eqref{eq:AO_t}, which is the subject of the next few sections.

% ----------------------------------------------------------- %
\subsection{Scalarization}\label{sec:scalarize}
% ----------------------------------------------------------- %
%The purpose of this section is to simplify the (AO) problem in \eqref{eq:AO_t}. 

Observe that the optimization in \eqref{eq:AO_t} is over vectors. The purpose of this section is to simplify the (AO) into an optimization involving only scalar  variables. Of course, one of this has to play the role of the norm of $\w$, which is the quantity of interest. The main idea behind the ``scalarization" step of the (AO) is to perform the optimization over only the direction of the vector variables while keeping their magnitude constant. This is already hinted by the rearrangement of the order of min-max operations going from \eqref{eq:AO3} to \eqref{eq:AO_t}. Also, this process is facilitated by the following two:
\begin{enumerate}
\item The bilinear term $\ub^T\A\w$ that appears in the (PO) conveniently ``splits" into the two terms $\|\w\|_2\g^T\ub$ and $\|\ub\|_2\h^T\w$ in the (AO),
\item The term involving the regularizer, i.e. $f(\x_0+\w)$ has been expressed in a variational form as $\sup_\s \s^T\x_0 + \s^T\w - f^*(\s)$.
\end{enumerate}
The details of the reduction step are all summarized in Lemma \ref{lem:scalar} below which shows that the (AO) reduces to the following  \emph{convex} minimax  problem on four \emph{scalar} optimization variables:
\begin{align}
\inf_{\substack{0\leq\alpha\leq K_\alpha \\ \taug> 0}} \sup_{\substack{0\leq \beta\leq K_\beta \\ \tauh> 0} }\frac{\beta\taug}{2} + \frac{1}{n} \env{\loss}{\alpha\g + \z}{\frac{\taug}{\beta}}  - 
\begin{cases}
\frac{\alpha\tauh}{2}+ \frac{\beta^2\alpha}{2\tau_h}\frac{\|\h\|^2}{n} - \la\cdot\frac{1}{n}\env{f}{\frac{\beta\alpha}{\tauh}\h+\x_0}{\frac{\alpha\la}{\tauh}} &,\alpha >0\\
 \frac{\la}{n} f(\x_0) &,\alpha=0
 \end{cases}
 \label{eq:AO_scal},
\end{align}
where recall that  $$\env{\omega}{\ub}{\tau} := \min_{\vb} \{\frac{1}{2\tau}\|\ub-\vb\|_2^2 + \omega(\vb)\}$$
denotes the (vector) $\tau$-Moreau envelope of a function $\omega:\R^d\rightarrow\R$ evaluated at $\ub\in\R^d$. 
%
% =============================================== %
%\vspace{40pt}
% =============================================== %
\begin{lem}[Scalarization of the (AO)]\label{lem:scalar}
 The following statements are true regarding the two minimax optimization problems in \eqref{eq:AO_t} and \eqref{eq:AO_scal}:
\begin{enumerate} [(i)]
\item They have the same optimal cost.
\item The  objective function in \eqref{eq:AO_scal} is continuous on its domain,  (jointly) convex in $(\alpha,\taug)$ and (jointly) concave in $(\beta,\tauh)$. 
\item The order of inf-sup in \eqref{eq:AO_scal} can be flipped without changing the optimization. 
\end{enumerate}
\end{lem}
%

% ----------------------------------------------------------- %
\subsection{Convergence Analysis}
% ----------------------------------------------------------- %
The goal of this section is to show that the (AO) satisfies the conditions of Lemma \ref{lem:CGMT_M}. This requires a convergence analysis of its optimal cost. We work with the scalarized version of the (AO) that was derived in the previous section:
\begin{align}
&\hspace{100pt}\phi(\g,\h,\z,\x_0) = \inf_{\substack{0\leq \alpha\leq K_\alpha \\ \taug> 0}} \sup_{\substack{0\leq \beta\leq K_\beta \\ \tauh> 0} } \Rcn(\alpha,\taug,\beta,\tauh;\g,\h,\z,\x_0), \label{eq:AO_randopt}\\
&\Rcn=
\frac{\beta\taug}{2} + \frac{1}{n}\left\{\env{\loss}{\alpha\g + \z}{\frac{\taug}{\beta}}-\loss(\z)\right\} - 
\begin{cases}
 \frac{\alpha\tauh}{2} + 
  \frac{\beta^2\alpha}{2\tau_h}\frac{\|\h\|^2}{n} - \frac{\la}{n}\left\{\env{f}{\frac{\beta\alpha}{\tauh}\h+\x_0}{\frac{\alpha\la}{\tauh}}-f(\x_0)\right\} &,\alpha>0
  \\ 
  0 &,\alpha=0
  \end{cases},
\nn 
\end{align}
Here, when compared to \eqref{eq:AO_scal}, we have  subtracted from the objective the terms $\loss(\z)$ and $f(\x_0)$, which of course does not affect the optimization. The optimization is of course random over the realizations of $\g,\h,\z$ and $\x_0$, and, by the WLLN, it is easy to identify the converging value of the objective function $\Rc_n$ for fixed parameter values $\alpha,\taug,\beta,\tauh$. Indeed, it converges to the objective function of the (SPO) problem in \eqref{eq:AO_det_thm}. For our goals, we need to show that minimax of the converging sequence of objectives converges to the minimax of the objective of the (SOP). 
Convexity of $\Rc_n$ plays a crucial role here since is being use to conclude local uniform convergence from the pointwise convergence. Uniform convergence is a requirement to conclude the desired.\footnote{We remark that the tools used for this part of the proof are similar to those classically used for the study of consistency of  $M$-estimators in the classical regime where $n$ is fixed and $m$ goes to infinity, cf. Arg-min theorems e.g. \cite[Thm.~7.70]{liese2008statistical}, \consist.}

\begin{lem}[Convergence properties of the (AO)]\label{lem:convergence}
Let $\Rcn(\alpha,\taug,\beta,\tauh):=\Rcn(\alpha,\taug,\beta,\tauh;\g,\h,\z,\x_0)$ be  defined as in \eqref{eq:AO_randopt}, and,
\begin{align}\label{eq:AO_rand2}
\phi^\mn_{\Ac}:=\phi^\mn_{\Ac}(\g,\h,\z,\x_0):=\inf_{\substack{\alpha\in\Ac 
\\
%0\leq\alpha\leq K_\alpha\\
 \taug>0}}~\sup_{\substack{0\leq \beta\leq K_\beta\\ \tauh>0}}~ \Rcn(\alpha,\taug,\beta,\tauh),
\end{align}
for $\Ac\subseteq[0,\infty)$.
%where $\Ac,\Bc\subseteq[0,\infty)$ and $\Tc_g,\Tc_h\subseteq(0,\infty)$. % are compact subsets.
Further consider the following deterministic convex program
\begin{align}\label{eq:AO_det}
\phio_{\Ac}:=\inf_{\substack{\alpha\in\Ac\\
% \alpha>0 \\
  \taug>0}}~\sup_{\substack{\beta\geq 0\\ \tauh>0}}~ \Dc(\alpha,\taug,\beta,\tauh):= 
\begin{cases}
\frac{\beta\taug}{2} + 
\delta\cdot\Lm{\alpha}{\frac{\taug}{\beta}}&, \beta>0\\
-\delta\cdot L_0 &,\beta=0
\end{cases}
-
\begin{cases}
 \frac{\alpha\tauh}{2} + \frac{\alpha\beta^2}{2\tauh} - \la\cdot \FFm{\frac{\alpha\beta}{\tauh}}{\frac{\alpha\la}{\tauh}} & ,\alpha>0\\
 0& ,\alpha=0
 \end{cases}.
\end{align}
where $L$ and $F$ as in Theorem \ref{thm:master}.  If Assumption \ref{ass:tech}(a) and \ref{ass:prop} hold, then,
\begin{enumerate}[(a)]
\item $\Rc_n(\alpha,\taug,\beta,\tauh)\rP\Dc(\alpha,\taug,\beta,\tauh)$, for all $(\alpha,\taug,\beta,\tauh)$, and, $\Dc(\alpha,\taug,\beta,\tauh)$ is convex in $(\alpha,\taug)$ and concave in $(\beta,\tauh)$. 
\item Assume $\alpha_*$ is the unique minimizer in \eqref{eq:AO_det} with $\Ac:=[0,\infty)$. For any $\eps>0$, define $\Sc_\eps:=\{\alpha~|~|\alpha-\alpha_*|<\eps\}$. Then, for any sufficiently large constants $K_\alpha>\alpha_*$ and $K_\beta>0$, and for all $\eta>0$, it holds with probability approaching 1 as $n\rightarrow\infty$:
%
%
%For $\Ac=[0,\infty)$, assume that $(\alpha_*,\beta_*)$ is a saddle point in the sense that $\inf_{\taug>0}\sup_{\tauh>0}\Dc(\alpha_*,\taug,\beta_*,\tauh)=\phio_{[0,\infty)}$, and $\alpha_*$ is the unique $\alpha$ such that such $\beta_*$ exists. 
%%that the optimal level set $$\left\{(\alpha_*,\taug_*,\beta_*,\tauh_*)~\big|~\Dc(\alpha_*,\taug_*,\beta_*,\tauh_*) = \inf_{\alpha,\taug}\sup_{\beta,\tauh}\Dc(\alpha,\taug,\beta,\tauh)\right\}$$ of the optimization in \eqref{eq:AO_det} is bounded. 
%%for any set $\Ac$, define $\Ac^c_\eps = \Ac \setminus \Sc_\eps$, where 
%For any $\eps>0$, define $\Sc_\eps:=\{\alpha~|~|\alpha-\alpha_*|<\eps\}$. Then, for any constants $K_\alpha>\alpha_*$ and $K_\beta>\beta_*$, the following are true:
\begin{enumerate}[(i)]
\item $\phi_{[0,K_\alpha]} < \phio_{[0,\infty)}+\eta$,
\item $\phi_{[0,K_\alpha]\setminus \Sc_\eps} \geq \phio_{[0,\infty)\setminus \Sc_\eps} - \eta$,
\item $\phio_{[0,\infty)\setminus \Sc_\eps}>\phio_{[0,\infty)}$.
\end{enumerate}
\end{enumerate}
\end{lem}
\subsection{Putting all the Pieces Together}
% ----------------------------------------------------------- %
We are now ready to conclude the proof of Theorem \ref{thm:master}. 
\begin{proof}[Proof of Theorem \ref{thm:master}]
Fix any $\eps>0$. Consider the set $\Sc_\eps = \{\w~|~|\|\w\|_2-\alpha_*\|_2<\eps$ as in Lemma \ref{lem:CGMT_M}. We use the same notation as in the lemma. Let $K_\alpha>\alpha_*$ and arbitrarily large (but finite) $K_\beta>0$. From Lemma \ref{lem:scalar}(i) $\phi(\g,\h)$ is equal to the optimal cost of the optimization in \eqref{eq:AO_scal}. But, from Lemma \ref{lem:convergence}(b)(i), the latter converges in probability to some constant $\phio$ (see Lemma \ref{lem:convergence} for the exact value constant). The same line of arguments applies to $\phi_{\Sc_\eps^c}(\g,\h)$, showing that it converges to another constant $\phio_{\Sc_\eps^c}$. Again from Lemma \ref{lem:convergence}(iii): $\phio_{\Sc_\eps^c}>\phio$. Thus, the conditions of Lemma \ref{lem:CGMT_M} are satisfied, and, it implies that the magnitude of any optimal minimizer (say) $\hat\w^{(PO)}$ of the (PO) problem in \eqref{eq:PO3} satisfies $\hat\w^{(PO)}\in\Sc$ in probability, in the limit of $n\rightarrow\infty$. 
%Then, in view of Lemma \ref{lem:K_b} (recall $K_\beta$ was chosen arbitrarily large), $\hat\w^B\in\Sc_{\eps}$, where $\hat\w^B$ is the optimal minimizer of the minimization in \eqref{eq:PO_B}. Since this holds for any $\eps>0$, it follows that $\|\hat\w^B\|\rP\alpha_*$. The proof completes with an appeal to Lemma \ref{lem:PO_B}.
\end{proof}

\section{Proofs for Section \ref{sec:proof}}\label{sec:thm_app}

\subsection{Proof of Theorem \ref{thm:CGMT}(iii)}
%\begin{proof}[Proof of (iii)]

%We only prove the second claim; the first inequality follows the exact same argument. 
Consider the following event 
$$
\Ec = \{ \Phi_{\Sc^c}(\G) \geq \phio_\Scc - \eta ~,~  \Phi(\G) \leq \phio + \eta \}.
$$
In this event, it is not hard to check using assumption (a) that $\Phi_\Scc>\Phi$, or equivalently $\w_\Phi\in\Sc$. Thus, it suffices to show that $\Ec$ occurs with probability at least $1-4p$.

Indeed, from statement (i) of the theorem and assumption (c),
\begin{align}%\label{eq:COLT1_gen}
\Pro(\Phi_{\Sc^c}(\G)< \phio_\Scc - \eta) \leq 2\Pro( \phi_{\Sc^c}(\g,\h) \leq \phio_\Scc - \eta ) \leq 2p.\nn
\end{align}
Also, from statement (ii) of the theorem and assumption (b),
\begin{align}%\label{eq:COLT1_gen}
\Pro(\Phi(\G) > \phio + \eta) \leq 2\Pro( \phi(\g,\h) \geq \phio + \eta ) \leq 2p.\nn
\end{align}
Combining the above displays the claim follows from a union bound.
%\end{proof}

%===================================================== %

\subsection{Proof of Corollary \ref{cor:asy_CGMT}}
Call $\eta:=(\phio_\Scc-\phio)/3>0$. By assumption, for any $p>0$ there exists $N:=N(\eta,p)$ such that the events $\{\phi<\phio+\eta\}$ and $\{\phi_\Scc>\phio_\Scc-\eta\}$ occur with probability at least $1-p$ each, for all $n>N$. Then, for all $n>N$, we can apply Theorem \ref{thm:CGMT}(iii) to conclude that $\w_{\Phi}(\G)\in\Sc$ with probably at least $1-4p$. Since this holds for all $p>0$, the proof is complete.

%===================================================== %

\subsection{Proof of Lemma \ref{lem:PO_B}}
%\begin{proof}
For convenience, denote with $M(\w)$ the objective function in \eqref{eq:PO}. For some $\eps>0$ such that $\alpha+\eps<K_\alpha$ (e.g. $\eps=\zeta/2$ in \eqref{eq:K_a}), denote $\Dc:=\{\w~|~\alpha-\eps\leq \|\w\|_2\leq \alpha+\eps\}$. By assumption, with probability approaching 1 (w.p.a. 1).
\begin{align}\label{eq:wBass}
\hat\w^B\in\Dc.
\end{align}
For the shake of a contradiction, assume that there exists optimal solution $\hat\w$ of \eqref{eq:PO} such that $\hat\w\not\in\Dc$ w.p.a. 1. Clearly, 
\begin{align}\label{eq:clear121}
M(\hat\w)\leq M(\hat\w^B).
\end{align}
Suppose $\hat\w\in\Sc_\w$, then $\hat\w$ is optimal for \eqref{eq:PO_B} and satisfies \eqref{eq:wBass}, which contradicts our assumption. Thus, $\hat\w\not\in\Sc_\w$. 
Next, let $\w_\theta:=\theta \hat\w + (1-\theta) \hat\w^B$ for $\theta\in(0,1)$ such that $\w_\theta\not\in\Dc$ and $\w_\theta\in\Sc_w$ (always possible, by definition of $\Dc$). By the convexity of $F$ and \eqref{eq:clear121}, it follows that
$M(\hat\w_\theta) \leq M(\hat\w^B)$. Hence, $\hat\w_\theta$ is optimal for \eqref{eq:PO_B} and satisfies \eqref{eq:wBass}, which, again, is a contradiction. This completes the proof.
%\end{proof}

%===================================================== %

\subsection{Proof of Lemma \ref{lem:K_b}}
It suffices to prove the equivalence of the optimization \eqref{eq:PO2a} and \eqref{eq:PO2}. Let $\w_*,\vb_*,\ub_*$ be optimal in \eqref{eq:PO2a}. To prove the claim, we show that $\ub_*\in\Sc_\ub\left(\Leftrightarrow\|\ub_*\|_2\leq K_\beta\right)$ w.p.a. 1. From the first order optimality conditions in \eqref{eq:PO2a}, we find that 
\begin{align}
\ub_* \in \frac{1}{\sqrt{n}}\partial\loss(\vb_*)\label{eq:opt11}\\
%\qquad\text{ and } \qquad
\vb_* = {\z} - \sqrt{n}\A\w_*.\label{eq:opt22}
\end{align}
Recall Assumption \ref{ass:tech}(b) and consider two cases.
First, if $\sup_{\vb\in\R^m}\sup_{\s\in\partial\loss(\vb)}\|\s\|_2<\infty$, the claim follows directly by \eqref{eq:opt11}. 
Next, assume that w.h.p., $\|\z\|_2\leq C_1\sqrt{n}$ for constant $C_1>0$. Also, a standard high probability bound on the spectral norm of Gaussian matrices gives $\|\A\|_2\leq C_2$, e.g. \cite{Ver}. Using these, boundedness of $\w_*$ and \eqref{eq:opt22}, we find that $\|\vb_*\|_2\leq C_3\sqrt{n}$ w.h.p.. Then, the normalization condition $\frac{1}{\sqrt{n}}\sup_{\s\in\partial\loss(\vb)}\|\s\|_2\leq C$ for all $\|\vb\|_2\leq c\sqrt{n}$ and all $n\in\mathbb{N}$, yields the desired, i.e. $\|\ub_*\|_2\leq C$ holds with probability approaching 1 as $n\rightarrow\infty$.

%===================================================== %

\subsection{ Proof of Lemma \ref{lem:CGMT_M} }
Let $\w_*$ denote an optimal solution of the ``bounded" optimization in \eqref{eq:PO3}. It will suffice to prove that $\w_*\in\Sc$ in probability.  To see this, recall from Lemma \ref{lem:K_b} that \eqref{eq:PO3} is asymptotically equivalent to \eqref{eq:PO_B}. Then, Lemma \ref{lem:PO_B} and the assumption $\alpha_*<K_\alpha$ guarantee that $\hat\w(\A)\in\Sc$ in probability, as desired.

Denote $\Phi:=\Phi(\A)$ the optimal cost of the minimization in \eqref{eq:PO3}
 and $\Phi_{\Sc^c}:=\Phi_{\Sc^c}(\A)$ the optimal cost of the same problem when the minimization is further restricted to be over the set $\w\in\Sc^c$. Note that $\w_*\in\Sc$ iff $\Phi_{\Sc^c}(\A)>\Phi(\A)$; hence, it will suffice to prove that the latter event occurs in probability. 

We do so by relating the (PO) in \eqref{eq:PO3} to the Auxiliary Optimization (AO) in \eqref{eq:AO_t} using Theorem \ref{thm:CGMT}.
For concreteness, denote the objective function in \eqref{eq:AO_t} with $A(\w,\vb,\ub,\s)$, and, recall $\Sc_\w:=\{\w~|~\|\w\|_2\leq K_\alpha\}$, $\Sc_\ub:=\{\ub~|~\|\ub\|_2\leq K_\beta\}$. With these, define
\begin{align}\label{eq:AO_phi}
{\phi^P}:={\phi^P}(\g,\h):=\min_{\w\in\Sc_\w,\vb}\max_{\ub\in\Sc_{\ub},\s} A(\w,\vb,\ub,\s)\quad\text{and}\quad
{\phi^D}:={\phi^D}(\g,\h):=\max_{\ub\in\Sc_{\ub},\s}\min_{\w\in\Sc_\w,\vb} A(\w,\vb,\ub,\s).
\end{align}
Observe here that the order of min-max in $\phi^P$ is exactly as in the original formulation of the CGMT, cf. \eqref{eq:AO_gen}; $\phi^D$ is the dual of it, and $\phi$ in \eqref{eq:AO_t} involves yet another change in the order of the optimizations. The reason we prefer to work with the later problem, is that this particular order allows for a number of simplifications performed in  Section \ref{sec:scalarize}.

 As done before, denote with ${\phi^P}_{\Sc^c},{\phi^D}_{\Sc^c}$ the optimal cost of the optimizations in \eqref{eq:AO_phi} under the additional constraint $\w \in\Sc^c$. The two problems in \eqref{eq:AO_phi} are related to the one in \eqref{eq:AO_t} as follows:
\begin{align}
{\phi^P}_{\Sc^c} &= \min_{\substack{\w\in\Sc_\w,\vb \\ \w\in\Sc^c}}\max_{\ub\in\Sc_{\ub},\s} A(\w,\vb,\ub,\s) = \min_{\substack{\w\in\Sc_\w,\vb \\ \w\in\Sc^c}}\max_{\beta,\s}\max_{\|\ub\|_2=\beta}A(\w,\vb,\ub,\s) \nn
\\ &\geq  \max_{\beta,\s}\min_{\substack{\w\in\Sc_\w,\vb \\ \w\in\Sc^c}}\max_{\|\ub\|_2=\beta}A(\w,\vb,\ub,\s) = \phi_{\Sc^c}, \label{eq:ineq1}
\end{align}
where the inequality follows from the min-max inequality \cite[Lem.~36.1]{Roc70}.
Similarly, 
\begin{align}
{\phi^D} &= \max_{\ub\in\Sc_{\ub},\s}\min_{\w\in\Sc_\w,\vb} A(\w,\vb,\ub,\s) = \max_{\beta,\s}\max_{\|\ub\|_2=\beta}\min_{\w\in\Sc_\w,\vb}A(\w,\vb,\ub,\s)\nn \\&\leq  \max_{\beta,\s}\min_{\w\in\Sc_\w,\vb}\max_{\|\ub\|_2=\beta}A(\w,\vb,\ub,\s) = \phi.\label{eq:ineq2}
\end{align}

Furthermore, they are related to the (PO) via the CGMT. 
From Theorem \ref{thm:CGMT}(i), for all $c\in\R$:
\begin{align}\label{eq:COLT1}
\Pro({\Phi}_{\Sc^c}< c) \leq 2\Pro( {\phi^P}_{\Sc^c} \leq c ).
\end{align}
Also, from Theorem \ref{thm:CGMT}(ii)\footnote{more precisely, please refer to equation (32) in \cite{COLT15}.}:
\begin{align}\label{eq:COLT2}
\Pro({\Phi} > c) \leq 2\Pro( {\phi^D} \geq c ).
\end{align}

The remaining of the proof is in the same lines as the proof of \ref{thm:CGMT}(iii), but is included for clarity. 
Let $\eta:=(\phi_{\Scc}-\phi)/3>0$. We may apply \eqref{eq:COLT1} for $c=\xio_{\Sc^c}-\eta$ and combine with \eqref{eq:ineq1} to find that 
\begin{align}\label{eq:conc1}
\Pro( {\Phi}_{\Sc^c} < \xio_{\Sc^c} - \eta ) \leq 2 \Pro( {\phi^P}_{\Sc^c} \leq \xio_{\Sc^c} - \eta ) \leq 2 \Pro( \phi_{\Sc^c} \leq \xio_{\Sc^c} - \eta ).
\end{align}
From assumption (b) the last term above tends to zero as $n\rightarrow\infty$. In a similar way, combining \eqref{eq:COLT2}, \eqref{eq:ineq2} and assumption (a), we find that 
\begin{align}\label{eq:conc2}
\Pro( {\Phi} > \xio + \eta ) \leq 2 \Pro( {\phi^D} \geq \xio + \eta ) \leq 2 \Pro( \phi \geq \xio + \eta ) ,
\end{align}
goes to zero with $n\rightarrow\infty$.
Denote the event $\Ec = \{{\Phi}_{\Sc^c}\geq \xio_{\Sc^c} - \eta \text{~and~} {\Phi}\leq \xio + \eta$\}. From \eqref{eq:conc1} and \eqref{eq:conc2} the event occurs with probability approaching 1. Furthermore, in this event, after using assumption (a), we have ${\Phi^b}_{\Sc^c}\geq \xio_{\Sc^c} - \eta > \xio + \eta \geq {\Phi^b}$; equivalently, the optimal minimizer satisfies $\w_*\in\Sc$, which completes the proof.

%===================================================== %
\subsection{Proof of Lemma \ref{lem:scalar}}\label{app:scalar}
%\begin{proof}

(i) We start by showing how the vector optimization in \eqref{eq:AO_t} can be reduced to the scalar one that appears in \eqref{eq:AO_scal}. This requires the following steps.

\vp
\noindent{\emph{Optimizing over the direction of $\ub$:}} Performing the inner maximization is easy. In particular, using the fact that $\max_{\|\ub\|_2=\beta}\ub^T\tb = \beta\|\tb\|_2$ for all $\beta\geq0$ the problem simplifies to a max-min one:
\begin{align}\nn
\max_{0\leq \beta\leq K_\beta,\s}~ \min_{\|\w\|\leq K_\alpha,\vb} ~ \frac{\beta}{\sqrt{n}}\|~\|\w\|_2\g+\z-\vb~\|_2 - \frac{\beta}{\sqrt{n}}\h^T\w + \frac{1}{{n}}\loss(\vb) +\frac{\la}{{n}}\s^T{\x_0}+\frac{\la}{\sqrt{n}}\s^T\w-\frac{\la}{{n}} f^*(\s),
\end{align}

\vp
\noindent{\emph{Optimizing over the direction of $\w$:}} Next, we fix  $\|\w\|_2=\alpha$, and, similar to what was done above, minimize over its direction:
\begin{align}
\label{eq:ii} \max_{0\leq \beta\leq K_\beta,\s}~ \min_{0\leq \alpha\leq K_\alpha,\vb} ~ \frac{\beta}{\sqrt{n}}\|~\alpha\g+\z-\vb~\|_2 + \frac{1}{{n}}\loss(\vb)  - \frac{\alpha}{\sqrt{n}} \| \beta\h-\la\s\|_2 + \frac{\la}{{n}}\s^T\x_0 -\frac{\la}{{n}} f^*(\s).
\end{align}

\vp
\noindent{\emph{Changing the orders of min-max:}} Denote with $M(\alpha,\beta,\vb,\s)$ the objective function above. It can be checked that $M$ is jointly convex in $(\alpha,\vb)$ and jointly concave in $(\beta,\s)$ (cf. Lemma \ref{lem:jointConv}).
Thus, $\min_{\vb}M$ is convex in $\alpha$ and jointly concave in $(\beta,\s)$. Furthermore, the constraint sets are all convex and the one over which minimization over $\alpha$ occurs is bounded. Hence, as  in \cite[Cor.~3.3]{Sion} %\cite[Cor.~]{Roc70} \chris{change this reference to Sion's minmax theorem as in wikipedia, which only requires one of the sets to closed+bounded and the other can be whatever}
 we can flip the order of $\max_{\beta,\s}\min_{\alpha}$, to conclude with
$$
\min_{0\leq\alpha\leq K_\alpha}\max_{0\leq\beta\leq K_\beta} \max_{\s}\min_{\vb} M(\alpha,\beta,\vb,\s).
$$
Also, observe that the order of optimization among $\vb$ and $\s$ does not affect the outcome.

% As a next step, a very similar argument shows that we can perform the minimization over $(\alpha,\vb,\tau)$ before the maximization over $\s$. To conclude, the order of optimization in \eqref{eq:to flip} is now $$\max_{\beta,\tauh}\min_{\alpha,\taug,\vb}\max_\s ~\cdots .$$

\vp
\noindent{\emph{The square-root trick:}} We apply the fact that $\sqrt{\chi}=\inf_{\tau>0} \{ \frac{\tau}{2} + \frac{\chi}{2\tau}\}$ to both the terms $\frac{1}{\sqrt{n}}\|\alpha\g+\z-\vb\|_2$ and $\frac{1}{\sqrt{n}}\|\beta\h-\la\s\|_2$:
\begin{align}\nn
\min_{0\leq \alpha\leq K_\alpha} \max_{0\leq \beta\leq K_\beta}~ \inf_{\taug>0} ~\sup_{\tauh> 0}~ \frac{\beta\taug}{2}&+\frac{1}{n}\min_{\vb} \left\{\frac{\beta}{2\taug}\|~\alpha\g+\z-\vb~\|_2^2 + \loss(\vb)\right\} \\ & - \frac{\alpha\tauh}{2} - \frac{1}{n}\min_{\s}\left\{\frac{\alpha}{2\tauh} \| \beta\h-\la\s\|_2^2 - \la\s^T\x_0 +\la f^*(\s)\right\}.\label{eq:to flip}
\end{align}

%\vp
%\noindent{\emph{Changing the orders of min-max:}} 

\vp
\noindent{\emph{Identifying the Moreau envelope:}} Arguing as before, we can change the order of optimization between $\beta$ and $\taug$. Also, it takes only a few algebra steps and  using basic properties of Moreau envelope functions (in particular, Lemma \ref{lem:Moreau}(ii)) in order to rewrite the last summand in  \eqref{eq:to flip} as below. If $\alpha>0$, then,
\begin{align}
\min_{\s}\left\{\frac{\alpha}{2\tauh} \| \beta\h-\la\s\|_2^2 - \la{\s^T\x_0} +\la f^*\left({\s}\right)\right\} 
&=  -\frac{\tauh}{2\alpha}{\|\x_0\|_2^2} - \beta{\h^T\x_0} +\la\cdot \env{f^*}{\frac{\beta}{\la}\h+\frac{\tauh}{\alpha\la}\x_0}{\frac{\tauh}{\alpha\la}} 
\label{eq:beforelasthere}
\\
%\frac{\beta^2\alpha}{2\tau_h}\|\h\|^2 - \frac{\tauh}{2\alpha}\left(\|\x_0\| + \frac{\alpha\beta}{\tauh}\h\right)^2 - \la\cdot \env{f^*}{\frac{\beta}{\la}\h+\frac{\tauh}{\alpha\la}\x_0}{\frac{\tauh}{\alpha\la}} \nn \\
&= \frac{\beta^2\alpha}{2\tau_h}{\|\h\|^2} - \la\cdot\env{f}{\frac{\beta\alpha}{\tauh}\h+\x_0}{\frac{\alpha\la}{\tauh}}. \label{eq:lasthere}
\end{align}
Otherwise, if $\alpha=0$, then the same term equals $-\la f(\x_0)$ since $\max_{\s}\s^T\x_0- f^*(\s) =  f(\x_0)$.

\vp
\noindent~(ii) The continuity of the objective function in \eqref{eq:AO_scal} follows directly from the continuity of the Moreau envelope functions, cf. \cite[Lem.~1.25, 2.26]{RocVar}. In particular, regarding the two branches of the objective: 
  it can be checked, using the continuity of the Moreau envelope, that the limit of the RHS in \eqref{eq:lasthere} as $\alpha\rightarrow 0$ evaluates to $-\la f(\x_0)$. (In fact, this is the unique extension of the upper branch to a continuous finite convex function on the whole $\alpha\geq0, \tau>0$, as per \cite[Thm.~10.3]{Roc70}).

 Convexity of \eqref{eq:AO_scal} can be checked from \eqref{eq:to flip}. By applying Lemma \ref{lem:jointConv}, after minimization over $\vb$ the Moreau Envelope remains jointly convex with respect to $\al$ and $\taug$ and concave in $\beta$. The same argument (and similar lemma) holds for the last term of \eqref{eq:to flip} in which after minimization over $\s$ it remains jointly convex in $\beta$ and $\tauh$ and concave in $\al$. Then the negative sign before this term makes it jointly concave in $\beta$ and $\tauh$ and convex over $\al$.

%\end{proof}

%===================================================== %

\subsection{Proof of Lemma \ref{lem:convergence}}\label{app:convergence}
%\begin{proof}
%\begin{enumerate}[(a)]
%\item 

\noindent{\textbf{(a)}}~
 By Assumption \ref{ass:tech}(a) the normalized Moreau envelope functions in \eqref{eq:AO_randopt} converge in probability to $L$ and $F$, respectively. Also, $\|\h\|_2^2/n\rP 1$ by the WLLN. This proves the convergence part.
%  from the assumptions on $\x_0$ in Section \ref{sec:model}, we have $\|\x_0\|_2^2/n\rP\sigma_x^2$ and $\h^T\x_0/n\rP 0$.

Lemma \ref{lem:scalar}(i) showed $\Rc_n$ to be convex-concave. Then, the same holds for $\Dc$ by point-wise convergence and the fact that convexity is preserved by point wise limits.

\vspace{5pt}
\noindent{\textbf{(b)}}~
Call 
\begin{align}\label{eq:M_a}
M_n(\alpha) = \sup_{\substack{0\leq \beta\leq K_\beta\\\tauh>0}}\inf_{\taug>0} \Rc_n(\alpha,\taug,\beta,\tauh)
\quad\text{ and }\quad
M(\alpha) = \sup_{\substack{0\leq \beta \\\tauh>0}}\inf_{\taug>0} \Dc_n(\alpha,\taug,\beta,\tauh).
\end{align}

The bulk of the proof consists of showing that the following two statements hold 
\begin{align}\label{eq:alpha>0}
\forall~\text{compact susets $\Ac\subset(0,\infty)$ and sufficiently large $K_\beta:=K_\beta(\Ac)>0$}:~ 
 \inf_{\alpha\in\Ac} M_n(\alpha) \rP \inf_{\alpha\in\Ac}M(\alpha)
\end{align}
and, 
\begin{align}\label{eq:alpha=0}
\forall\eps>0, \text{w.p.a.1}:~ M_n(0) < M(0) + \eps.
\end{align}
Before proceeding with the proof of those, let us show how the conclusion of the lemma is reached once \eqref{eq:alpha>0} and \eqref{eq:alpha=0} are established.

\vspace{5pt}
\noindent{\underline{Using \eqref{eq:alpha>0} and \eqref{eq:alpha=0} to prove the lemma}~}:
Fix $K_\alpha> \alpha_*$, any $\delta>0$ such that $\Ac:=[\alpha_*-2\delta,\alpha_*+2\delta]\subset(0,K_\alpha]$ and $K_\beta>0$ large enough such that \eqref{eq:alpha>0} and \eqref{eq:alpha=0} both hold. Then, for all $\eps>0$, w.p.a.1:
\begin{align}\label{eq:bhma_a}
\min_{\substack{0\leq\alpha\leq K_\alpha}}M_n(\alpha) \leq \min_{\substack{\alpha\in \Ac}}M_n(\alpha) \leq 
 M_n(\alpha_*) < M(\alpha_*) + \eps.
\end{align}
For the last inequality above: if $\alpha_*=0$, it follows from \eqref{eq:alpha=0}, or otherwise from \eqref{eq:alpha>0}.
% when applied for a set $\Ac$ containing $\alpha_*$.

Next, consider the compact set $\Ac_{l}=\{\alpha>0~|~\alpha\in[\alpha_*-2\delta,\alpha-\delta]~\}$ and $\Ac_{r}=\{\alpha>0~|~\alpha\in[\alpha_*+\delta,\alpha+2\delta]~\}$. (Note that if $\alpha_*=0$, then $\Ac_l$ is empty.) From \eqref{eq:alpha>0}, we know that for all $\eps>0$, w.p.a.1
$$
\min_{\alpha\in\Ac_l} M_n(\alpha) > \min_{\alpha\in\Ac_l}M(\alpha) - \eps \quad\text{ and }\quad  \min_{\alpha\in\Ac_u} M_n(\alpha) > \min_{\alpha\in\Ac_u}M(\alpha) - \eps.
$$
Let $\Ac_{lu}=\Ac_l\cup\Ac_u$ and combine the above to find
\begin{align}\label{eq:bhma_b}
\min_{\alpha\Ac_{lu}} M_n(\alpha) > \min_{\alpha\in\Ac_{lu}}M(\alpha) - \eps.
\end{align}

By assumption on uniqueness of $\alpha_*$ and on convexity of $M$, we have  \begin{align}\label{eq:bhma_c}
M(\alpha_*)<\min_{\alpha\in\Ac_{lu}}M(\alpha)
\end{align}
and $M(\alpha_*)=\min_{\alpha\in\Ac }M(\alpha)$. Thus,
Applying  \eqref{eq:bhma_a} and \eqref{eq:bhma_b} for $\eps= (\min_{\alpha\in\Ac_{lu}}M(\alpha) - M(\alpha_*))/3$ yields w.p.a.1 :
\begin{align}\label{eq:bhma_d}
\min_{\alpha\in\Ac_{lu}} M_n(\alpha) > \min_{\alpha\in\Ac_{lu}}M(\alpha) - \eps > M(\alpha_*) + \eps >  \min_{\alpha\in[\alpha_*-2\delta,\alpha_*+2\delta]} M_n(\alpha).
\end{align}

Thus, w.p.a.1,
$$
\hat\alpha_n := \arg\min_{\alpha\in\Ac}M_n(\alpha) \in (\alpha_*-\delta,\alpha_*+\delta).
$$
In this event, for any $\alpha\not\in\Ac$, there is a convex combination $\alpha_\theta:=\theta\hat\alpha_n + (1-\theta)\alpha$, ($\theta<1$) that equals either $\alpha_*-2\delta$ or $\alpha_*+2\delta$. By convexity,
$$
M_n(\alpha_\theta)\leq \theta M_n(\hat\alpha_n) + (1-\theta)M_n(\alpha)
$$
Also, from \eqref{eq:bhma_d}, $M_n(\hat\alpha_n)<M_n(\alpha_\theta)$. Combining those, we find $M_n(\hat\alpha_n)<M_n(\alpha)$, implying that $\hat\alpha_n$ is the minimizer of $M_n$ over the entire $[0,K_\alpha]$ w.p.a.1. In other words, for all $\eps$ w.p.a. 1,
\begin{align}\label{eq:bhma_e}
\min_{\alpha\in[0,K_\alpha]\setminus(\alpha_*-\delta,\alpha_*+\delta)}M_n(\alpha)\geq 
\min_{\alpha\in\Ac_{lu}} M_n(\alpha) > \min_{\alpha\in\Ac_{lu}}M(\alpha) - \eps.
\end{align}
To establish a connection with the three statements (i)-(iii) of the lemma, observe that $\phio_{[0,\infty)}=M(\alpha_*)$. Also, $\phio_{[0,\infty)\setminus\Sc_{\delta}}=  \min_{\alpha\in\Ac_{lu}}M(\alpha) $ (by convexity). With these, (i) corresponds directly to \eqref{eq:bhma_a}, (ii) to \eqref{eq:bhma_e}, and, (iii) to \eqref{eq:bhma_c}.

%Identify $\phio:=M(\alpha_*)<\min_{\alpha\in\Ac_{lu}}M(\alpha)=:\phio_\Sc$ to see that the proof is complete with \eqref{eq:bhma_a} and \eqref{eq:bhma_d}.

%Let $\hat\alpha_n=\argmax_{\alpha\in[\alpha_*-2\delta,\alpha_*+2\delta]}M_n(\alpha)$. Then,
%\begin{align}
%M_n(\hat\alpha_
%\end{align}

\vspace{5pt}
\noindent{\underline{Proof of \eqref{eq:alpha>0} and \eqref{eq:alpha=0}}~}:
%Consider the optimization in \eqref{eq:AO_det} for $\Ac=[0,\infty)$. 
From the first statement of the lemma, the objective function $\Rc_n$ of the (AO) converges point-wise to $\Dc$. We will use this to show that the minimax value of $\Rc_n$ converges to the corresponding minimax of $\Dc$.  
The proof is based on a repeated use of Lemma 
%\ref{lem:argmin_compact} and 
\ref{lem:argmin_open} below, about convergence of the infimum of a sequence of convex converging stochastic processes. This fact is essentially a consequence of what is known in the literature as \emph{convexity lemma}, according to which point wise convergence of convex functions implies uniform convergence in 
compact subsets. Please refer to Section \ref{sec:aux_Proofs} for the proof.

%\begin{lem}[Min-convergence -- Compact Sets]\label{lem:argmin_compact}
%Let $\Sc\subset\R$ be a compact set. Consider a sequence of continuous, proper, and, convex stochastic functions  $M_n:\Sc\rightarrow\R$, and, a continuous convex function $M:\Sc\rightarrow\R$, such that $M_n(x)\rP M(x)$, for all $x\in\Sc$. Then,
%$\min_{x\in\Sc} M_n(x) \rP  \min_{x\in\Sc} M(x).$  
%\end{lem}

\begin{lem}[Min-convergence -- Open Sets]\label{lem:argmin_open}
Consider a sequence of proper, convex stochastic functions  $M_n:(0,\infty)\rightarrow\R$, and, a deterministic function $M:(0,\infty)\rightarrow\R$, such that:
\begin{enumerate}[(a)]
\item $M_n(x)\rP M(x)$, for all $x>0$,
\item there exists $z>0$ such that $M(x)>\inf_{x>0}M(x)$ for all $x\geq z$.
%\item $\inf_{x>0} F(x)$ is either attained at a bounded set or is approached in the limit of $x\rightarrow 0$.
% set of minima $\{ \hat{x} | F(\hat{x})=\inf_{x\in\Sc} M(x) \}$ is non-empty and bounded, 
%\item $\lim_{x\rightarrow0^+}M_n(x)<\infty$ almost surely, for all $n\in\mathbb{N}$,
%\item $\inf_{x\in\Sc}M_n(x)>-\infty$  almost surely for all $n$.
\end{enumerate}
%For some bounded convex set $\Sc\subset\R$ that is either open or compact, let $M_n:\Sc\rightarrow\R$ be continuous convex stochastic processes such that: (i)
%$M_n(x)\rP F(x)$ for all $x\in\Sc$ and a continuous function $F:\Sc\rightarrow\R$. Further assume that $\inf_{x>0}M_n(x)>-\infty$ for all $n$, and, $\inf_{x>0}F(x)>-\infty$. 
Then, $\inf_{x>0} M_n(x) \rP  \inf_{x>0} F(x).$  
\end{lem}
\emph{1)}~
Fix $\alpha\geq 0,\beta>0$, and, $\tauh>0$. Consider
\begin{align}
M_n^{\alpha,\beta,\tauh}(\taug)&:=R_n(\alpha,\taug,\beta,\tauh),
\end{align}
\begin{align}\label{eq:M_taug}
M^{\alpha,\beta,\tauh}(\taug):= \Dc(\alpha,\taug,\beta,\tauh).
%\frac{\beta\taug}{2} + 
%\delta\cdot\Lm{\alpha}{\frac{\taug}{\beta}} -
%\frac{\alpha\tauh}{2} - \frac{\alpha\beta^2}{2\tauh} + \la\cdot \FFm{\frac{\alpha\beta}{\tauh}}{\frac{\alpha\la}{\tauh}}\nn
\end{align}
The functions $\{M_n\}$ are convex. Furthermore, $M_n^{\alpha,\beta,\tauh}(\taug)\rP M^{\alpha,\beta,\tauh}(\taug)$ point wise in $\taug$. Next, we show that $M^{\alpha,\beta,\tauh}$ is level-bounded, i.e. it satisfies condition (b) of Lemma \ref{lem:argmin_open}. In view of Lemma \ref{lem:lv_cvx}, it suffices to show that 
$
\lim_{\taug\rightarrow\infty} M^{\alpha,\beta,\tauh}(\taug) = +\infty,
$
or 
$
\lim_{\taug\rightarrow\infty} \left( \frac{\beta}{2} + \delta\cdot\frac{L(\alpha,\taug/\beta)}{\taug} \right) > 0.
$
By assumption \ref{ass:prop}(c), $\lim_{\taug\rightarrow\infty}L(\alpha,\taug/\beta) = -L_0$. There is two cases to be considered. Either $L_0<\infty$, or else, Assumption \ref{ass:prop}(d) holds. Either way, $\lim_{\taug\rightarrow\infty} L(\alpha,\taug/\beta)/\taug = 0$ and we are done. Now, we can apply Lemma \ref{lem:argmin_open} to conclude that
\begin{align}\label{eq:conv_taug}
\inf_{\taug>0}M_n^{\alpha,\beta,\tauh}(\taug) \rP \inf_{\taug>0}M^{\alpha,\beta,\tauh}(\taug).
\end{align}

%%%% =========================== \beta
\emph{2)}~
\vp
Next, again for fixed $\alpha\geq 0,\tauh>0$, consider (we use some abuse of notation here, with the purpose of not overloading notation)
$$
M_n^{\alpha,\tauh}(\beta):= \inf_{\taug>0} M^{\alpha,\beta,\tauh}_n(\taug)
$$
$$
M^{\alpha,\tauh}(\beta):= \inf_{\taug>0} M^{\alpha,\beta,\tauh}(\taug)
$$
The functions $\{M_n^{\alpha,\tauh}\}$ are concave in $\beta$, as the point wise minima of concave functions. Furthermore, $M_n^{\alpha,\tauh}(\beta)\rP M^{\alpha,\tauh}(\beta)$ point wise in $\beta>0$, by \eqref{eq:conv_taug}. 

\noindent{\underline{$\alpha>0$:}}~ For now and until further notice, restrict attention to the case $\alpha>0$.
Also, consider first $\beta>0$. We show that $M^{\alpha,\tauh}$ is level-bounded, i.e. it satisfies condition (b) of Lemma \ref{lem:argmin_open}. In view of Lemma \ref{lem:lv_cvx}, it suffices to show that 
$
\lim_{\beta\rightarrow+\infty} M^{\alpha,\tauh}(\beta) = -\infty,
$
or 
$
\lim_{\beta\rightarrow+\infty} \inf_{\taug>0}M^{\alpha,\beta,\tauh}(\taug) = -\infty.
$ This condition is equivalent to the following
\begin{align}\label{eq:cond_taug_22}
(\forall M>0)(\exists B>0)~\left[\beta>B\Rightarrow (\exists \{\taug\}_{k})~ [D(\alpha,\taug,\beta,\tauh)<-M] ~\right].
\end{align}

First, we show that
\begin{align}\label{eq:1lim}
\lim_{\beta\rightarrow+\infty}\frac{\alpha\beta^2}{2\tauh} - \la \cdot F\left( \frac{\alpha\beta}{\tauh} , \frac{\alpha\la}{\tauh} \right) = +\infty
\end{align}
%In fact, we will prove a stronger statement that will be also useful later in the proof: the function $\Hc(\beta,\tauh):=\frac{\alpha\tauh}{2} + \frac{\alpha\beta^2}{2\tauh} + \la\cdot F(\frac{\alpha\beta}{\tauh},\frac{\alpha\la}{\tauh})$ is level-bounded, i.e. the level sets $\{ (\beta,\tauh)~|~ \Hc(\beta,\tauh)\leq \gamma \}$ are bounded for all $\gamma$. Since, $\beta,\tauh>0$ and $\Hc$ is convex, it will suffice to prove that 
%\begin{align}\label{eq:2lim}
%\lim_{\substack{\tauh\rightarrow+\infty\\\beta\rightarrow+\infty}}\Hc(\beta,\tauh) = +\infty.
%\end{align}
%
%It is convenient to express $\Hc$ by introducing new variables $c=\frac{\alpha\beta}{\tauh}$ and $\tau=\frac{\alpha\la}{\tauh}$, as follows
%$$
%\Hc(c,\tau) := \la\left( \frac{\alpha^2}{\tau} + \frac{c^2}{2\tau} -  F(c,\tau) \right).
%$$
%Observe that $\tauh\rightarrow+\infty\implies \tau \rightarrow 0$. On the other hand, $c$ is allowed to take any value in $[0,+\infty]$. First, consider the case $c>0$. By assumption \ref{ass:prop}(?), it holds $\frac{2\tau}{c^2}F(c,\tau)< 1$. Thus, $\Hc(c,\tau)\rightarrow +\infty$. Next, let $c\rightarrow 0$. This implies $\beta\rightarrow0$, and so, $\frac{c}{\tau} = \frac{\beta}{\la}\rightarrow 0$. This, when combined with assumption \ref{ass:prop}(?) \chris{$\lim_{c,\tau,\rightarrow 0,\frac{c}{\tau}<\infty}F(c,\tau) = 0$.}, shows that $c\cdot \frac{c}{2\tau} - F(c,\tau)\rightarrow 0$. Thus, $\Hc(c,\tau)\rightarrow +\infty$.
%With this, we have shown \eqref{eq:2lim}.
This follows by Assumption \ref{ass:prop}(a) 
%% (or actually the inequality holds for all $c>0$ rather than in the limit!!!)
%%$$
%%\lim_{c\rightarrow+\infty}\frac{2\tau}{c^2}F(c,\tau)< 1,
%%$$
%$$
%F(c,\tau) < \frac{c^2}{2\tau}, 	\quad\forall c>0,
%$$
when applied for $c=\alpha\beta/\tauh$ and $\tau=\alpha\la/\tauh$ (recall here that $\alpha>0$). 

%Returning to \eqref{eq:1lim}, this follows from \eqref{eq:2lim} 
Next, choose $\{\taug\}_k\rightarrow 0$. For that choice, $\frac{\beta\taug}{2}+L(\alpha,\taug/\beta)\rightarrow \lim_{\tau\rightarrow 0}L(\alpha,\tau)<\infty$, where boundedness follows by Assumption \ref{ass:prop}(b). Thus, \eqref{eq:cond_taug_22} is correct and we may apply Lemma \ref{lem:argmin_open} to conclude that
 \begin{align}\label{eq:conv_beta0}
\sup_{\beta>0} M_n^{\alpha,\tauh}(\beta) \rP \sup_{\beta>0} M^{\alpha,\tauh}(\beta).
\end{align}
%By point-wise convergence assumption \ref{ass:tech}, if $L_0<\infty$, then  $\Rc_n(\alpha,\taug,\beta=0,\tauh)\rP\Dc(\alpha,\taug,\beta=0,\tauh)$

Now, we investigate the case $\beta=0$. We have,
$M_n^{\alpha,\tauh}(0)=-\frac{1}{n}\loss(\z)-\frac{\alpha\tauh}{2}+\frac{\la}{n}\left(\env{f}{\x_0}{\frac{\alpha\la}{\tauh}}-f(\x_0)\right)$ 
 and $M_n^{\alpha,\tauh}(0) = -\delta L_0 - \frac{\alpha\tauh}{2} + F(0,\frac{\alpha\la}{\tauh})$. 
 
 If $L_0<\infty$, then by assumption, $M_n^{\alpha,\tauh}(0)\rP M^{\alpha,\tauh}(0)$. Combined with \eqref{eq:conv_beta0}, we find
 \begin{align}\label{eq:conv_beta}
\sup_{\beta\geq0} M_n^{\alpha,\tauh}(\beta) \rP \sup_{\beta\geq0} M^{\alpha,\tauh}(\beta).
\end{align} 

Now, consider the case $L_0=+\infty$. Clearly, the optimal $\beta$ for $M^{\alpha,\tauh}$ is not at zero; thus, $\sup_{\beta\geq0} M^{\alpha,\tauh}(\beta) = \sup_{\beta>0} M^{\alpha,\tauh}(\beta)$. Also, by assumption,  for all $M$, $\lim_{n\rightarrow\infty}\Pro\left( \frac{1}{n}\loss(\z) > M \right) = 1$. Letting, $\eps>0$ and $M:=-\sup_{\beta>0}M^{\alpha,\tauh}(\beta) + \eps + \frac{\alpha\tauh}{2} - F(0,\frac{\al\la}{\tauh})$, then w.p.a.1, $M_n^{\alpha,\tauh}(0)<\sup_{\beta>0}M^{\alpha,\tauh}(\beta)-\eps\leq \sup_{\beta>0}M_n^{\alpha,\tauh}(\beta)$, where  the last inequality follows because of \eqref{eq:conv_beta0}. Again, this leads to \eqref{eq:conv_beta}. To sum up, \eqref{eq:conv_beta} holds for all $\alpha>0$.

\vp
\noindent{\underline{$\alpha=0$:}}~ We show that for all $\eps>0$, the following holds w.p.a.1:
\begin{align}\label{eq:conv_beta_a=0}
\sup_{\beta\geq 0} M_n^{\alpha=0,\tauh}(\beta) < \sup_{\beta\geq 0} M^{\alpha=0,\tauh}(\beta) + \eps.
 \end{align}
To begin with, note that for all $n$,
\begin{align}\label{eq:Fn_0}
\sup_{\beta\geq 0} M_n^{\alpha=0,\tauh}(\beta) \leq \sup_{\beta>0}\lim_{\taug\rightarrow 0}\frac{\beta\taug}{2} + \frac{1}{n}\min_{\vb}\left\{\frac{\beta}{2\taug}\|\z-\vb\|_2^2+\loss(\vb)-\loss(\z) \right\} = 0,
 \end{align}
 where we have used  Lemma 
 %the facts that $\loss(\z)\geq 0$ and 
 \ref{lem:Mconvex}(ix).
 Next, we show that 
 \begin{align}\label{eq:M_0}
 M^{\alpha=0,\tauh}(\beta)=0.
 \end{align}
  Using Assumption \ref{ass:prop}(c) on the non-negativity of $L_0$ and Assumption \ref{ass:prop}(b) that $\lim_{\tau\rightarrow0}L(c,\tau)=0$, it follows that  $M^{\alpha=0,\tauh}(\beta)\leq \sup_{\beta>0}\lim_{\taug\rightarrow0}{\frac{\beta\taug}{2}+L(0,\taug/\beta)} = 0$. Thus, it will suffice for the claim if we prove
\begin{align}\label{eq:lim_s}
 \lim_{\beta\rightarrow\infty}\inf_{\taug>0}{\frac{\beta\taug}{2}+L(0,\taug/\beta)} = 0,
\end{align}
 or equivalently,
$$
 \lim_{\beta\rightarrow\infty}\inf_{\kappa>0}{\kappa\left(\frac{\beta^2}{2}+\frac{L(0,\kappa)}{\kappa}\right)} = 0.
$$
Fix some $\beta>0$. Note that $\lim_{\kappa\rightarrow0}\frac{\kappa\beta^2}{2}+{L(0,\kappa)} = 0$, where we have used Assumption \ref{ass:prop}(b) that  $\lim_{\tau\rightarrow0}L(0,\tau) = 0.$ Also, $\lim_{\kappa\rightarrow\infty}\kappa\left(\frac{\beta^2}{2}+\frac{L(0,\kappa)}{\kappa}\right) = +\infty$, using Assumption \ref{ass:prop}(d) this time. Now, consider only $\beta>\sqrt{-L_{2,+}(0,0)}$ (see Assumption \ref{ass:prop}(b)). Then, 
 the function $\frac{\kappa\beta^2}{2}+{L(0,\kappa)}$ has a positive derivative at $\kappa\rightarrow0^+$. From this and convexity, it follows that for all $\kappa>0$, 
 $$\frac{\kappa\beta^2}{2}+{L(0,\kappa)}\geq \lim_{\kappa\rightarrow0}\frac{\kappa\beta^2}{2}+{L(0,\kappa)} = 0.$$ This proves \eqref{eq:lim_s} as desired. 

To complete the argument, \eqref{eq:conv_beta_a=0} follows by  \eqref{eq:Fn_0} and \eqref{eq:M_0}, and with this we have completed the proof of \eqref{eq:alpha=0}.

%
%
%We can write the argument above as $\frac{\taug}{\beta}\left(\frac{\beta^2}{2}+\frac{L(0,\taug/\beta)}{\taug/\beta}\right)$. This is no less than $\frac{\taug\beta}{2}$, where we have combined Assumptions \ref{ass:prop}(?) and (?) to the following: $\frac{L(0,\tau)}{\tau}\geq\lim_{\tau\rightarrow\infty}\frac{L(0,\tau)}{\tau}=0$. From this, we find,
% $$
% \lim_{\beta\rightarrow\infty}\inf_{\taug>0}{\frac{\beta\taug}{2}+L(0,\taug/\beta)} \geq  \lim_{\beta\rightarrow\infty}\inf_{\taug>0}{\frac{\beta\taug}{2}}.
% $$

%%%% =========================== \taug

\vp
\emph{3)}~
Keep $\alpha>0$ fixed and consider
$$
M_n^{\alpha}(\tauh):= \sup_{\beta\geq 0} M_n^{\alpha,\tauh}(\beta),
$$
$$
M^{\alpha}(\tauh):= \sup_{\beta\geq 0} M^{\alpha,\tauh}(\beta),
$$
The functions $\{M_n^{\alpha}\}$ and $F$ are all concave  in $\tauh$, as the point wise maxima of jointly concave functions. Furthermore, $M_n^{\alpha}(\tauh)\rP M^{\alpha}(\tauh)$ point wise in $\tauh$, by \eqref{eq:conv_beta}. Next, we show that $M^\tauh$ is level-bounded, i.e. it satisfies condition (b) of Lemma \ref{lem:argmin_open}. In view of Lemma \ref{lem:lv_cvx}, it suffices to show that 
$
\lim_{\tauh\rightarrow\infty} M^{\alpha}(\tauh) = +\infty,
$
or 
$\lim_{\tauh\rightarrow\infty}\sup_{\beta>0} \inf_{\taug>0}\Dc(\alpha,\taug,\beta,\tauh) = -\infty.
$
This is equivalent to the following
\begin{align}\label{eq:cond_taug_2}
(\forall M>0)(\exists T>0)~\left[\tauh>T\Rightarrow (\forall \{\beta\}_k)(\exists\{\taug\}_k)~ [D(\alpha,\taug,\beta,\tauh)<-M] ~\right].
\end{align}

Consider the function
$$\Hc(\beta,\tauh):=\frac{\alpha\tauh}{2} + \frac{\alpha\beta^2}{2\tauh} - \la\cdot F(\frac{\alpha\beta}{\tauh},\frac{\alpha\la}{\tauh}).$$
We show that \begin{align}\nn
 \Hc(\beta,\tauh)\geq \frac{\alpha\tauh}{2} .
\end{align}
To see this note that 
$
\env{f}{c\h+\x_0}{\tau} \leq \frac{c^2\|\h\|^2}{2\tau} +  f(\x_0).
$
Thus, $\frac{1}{n}\left\{\env{f}{c\h+\x_0}{\tau} - f(\x_0)\right\}\leq \frac{c^2\|\h\|^2}{2\tau n}$. The LHS converges to  $F(c,\tau)$ by Assumption \ref{ass:tech}(a) and the RHS converges to $\frac{c^2}{2\tau }$. Therefore,
$
F(c,\tau)\leq \frac{c^2}{2\tau}.
$
Applying this for $c=\frac{\alpha\beta}{\tauh}$ and $\tau=\frac{\alpha\la}{\tauh}$, we have that $\frac{\alpha\beta^2}{2\tauh} - \la\cdot F(\frac{\alpha\beta}{\tauh},\frac{\alpha\la}{\tauh})\geq 0$, as desired.

Then, 
\begin{align}\nn
\Dc(\alpha,\taug,\beta,\tauh) \leq \frac{\beta\taug}{2} + \delta\cdot L\left(\alpha,\frac{\taug}{\beta}\right) - \frac{\alpha\tauh}{2} .
\end{align} 
Also, note that for all $\beta>0$, we can choose (sequence) of $\taug$, such that $\beta\taug,\frac{\taug}{\beta}\rightarrow0$. Then, $\frac{\beta\taug}{2} + \delta\cdot L\left(\alpha,\frac{\taug}{\beta}\right)\rightarrow \lim_{\tau\rightarrow0}L(\alpha,\tau)=:A<\infty$. 
It can then be seen that \eqref{eq:cond_taug_2} holds for (say) $T:=T(M)=4(A+M)/\alpha$.

We can apply Lemma \ref{lem:argmin_open} to conclude that
\begin{align}\label{eq:conv_tauh}
\sup_{\tauh>0}M_n^{\alpha}(\tauh) \rP \sup_{\tauh>0}M^{\alpha}(\tauh).
\end{align}

%%%% =========================== \taug

\vp
\emph{4)}~ Finally, consider
$$
M_n(\alpha):= \sup_{\tauh>0} M_n^{\alpha}(\tauh),
$$
\begin{align}\label{eq:F}
M(\alpha):= \sup_{\tauh>0} M^{\alpha}(\tauh).
\end{align}
The functions $\{M_n\}$ and $F$ are all convex in $\tauh$, as the point wise maxima of convex functions. Furthermore, $M_n(\alpha)\rP M(\alpha)$ point wise in $\alpha$, by \eqref{eq:conv_tauh}. By assumption of the lemma, $F$ has a unique minimizer $\alpha_*$, which of course implies level boundedness. Thus, we can apply Lemma \ref{lem:argmin_open} to conclude that
\begin{align}\label{eq:conv_alpha}
\inf_{\alpha>0}M_n(\alpha) \rP \inf_{\alpha>0}M(\alpha).
\end{align}
Besides, pointwise convergence $M_n(\alpha)\rP M(\alpha)$ translates to uniform convergence over any compact subset $\Ac\subset(0,\infty)$ by the Convexity lemma \uniform,\cite[Lem.~7.75]{liese2008statistical}. Hence,
$$
 \inf_{\alpha\in\Ac} M_n(\alpha) \rP \inf_{\alpha\in\Ac}M(\alpha).
 $$
This is of course same as the desired in \eqref{eq:alpha>0}. Recall, \eqref{eq:alpha=0} was established in \eqref{eq:conv_beta_a=0}. The only thing remaining is showing  that there exists an optimal $\beta_*$ in $\sup_{\beta\geq0}M^{\alpha,\tauh}(\beta)$ that is bounded by some sufficiently large $K_\beta(\Ac)$. This follows from the level-boundedness arguments above as detailed immediately next. 

\vspace{5pt}
\noindent{\underline{Boundedness of solutions}~}: 
For a compact subset $\Ac\subset(0,\infty)$, we argue that there exists \emph{bounded} $\beta_*$ and sequences $\{\taug_*\}_{k},\{\tauh_*\}_{k}$ such that $(\alpha_*,\{\taug_*\}_{k},\beta_*,\{\tauh_*\}_{k})$  approaches \\$\min_{\alpha\in\Ac}\sup_{\substack{\tauh>0 \\\beta\geq0}}\inf_{\taug>0} \Dc(\alpha,\taug,\beta,\tauh)$. This follows from the work above. In particular, at each step in the proof of \eqref{eq:alpha>0} above, we showed level-boundedness of the corresponding functions. For example, \eqref{eq:cond_taug_2}  shows that there exists (sufficiently large) $T_h(\alpha)>0$ such that $\sup_{\tauh>0}M^{\alpha}(\tauh)$ is equal to $\sup_{T_h(\alpha) \geq \tauh>0}M^{\alpha}(\tauh)$. This holds for all $\alpha$; so, in particular, is true for $T_h:=\max_{\alpha\in\Ac}T_h(\alpha)$. Next, from \eqref{eq:1lim}  there exists $K_\beta(\alpha,T_h)$, such that $\sup_{\beta\geq0}M^{\alpha,\tauh}(\beta)$ is equal to $\sup_{K_\beta(\alpha,T_h) \geq\beta\geq 0}M^{\alpha,\tauh}(\beta)$.  Again, this holds for all $\alpha\in\Ac$, thus there exists sufficiently large $K_\beta>0$ such that (see also Lemma \ref{lem:sp_compact})
$$
\min_{\alpha\in\Ac}\sup_{\substack{\tauh>0 \\\beta\geq0}}\inf_{\taug>0} \Dc(\alpha,\taug,\beta,\tauh) = \inf_{\alpha\in\Ac}\sup_{\substack{\tauh>0 \\K_\beta\geq \beta\geq0}}\inf_{\taug>0} \Dc(\alpha,\taug,\beta,\tauh)
$$
%
%% with respect for all variables $\tauh,\beta,\taug$, we have shown that the optimal is 
%%
%% if there exists an optimal solution $\alpha_*$, then, there also exist optimal $\beta_*,\taug_*,\tauh_*$. In other words, there exists a saddle point $(\alpha_*,\taug_*,\beta_*,\tauh_*)$  of the min-max (DO) problem. This will follow from the work above. 
%
%
%By assumption \chris{Assumption!!!}, there exists a saddle-point $(\alpha_*,\beta_*)$. Let $K_\alpha, K_\beta>0$ be such that $(\alpha_*,\beta_*)\in[0,K_\alpha]\times [0,K_\beta]$. 
%
%Then, the minimization over $\alpha$ and maximization over $\beta$ can be constrained over $[0,K_\alpha]$ and $[0,K_\beta]$ respectively, without changing the overall optimization (see Lemma \ref{lem:sp_compact}). Thus, \eqref{eq:AO_det} is equivalent to the following
%\begin{align}\label{eq:AO_eqv2}
%\inf_{\substack{0\leq\alpha\leq K_\alpha\\ \taug>0}}~\sup_{\substack{0\leq\beta\leq K_\beta \\ \tauh>0}}~ \Dc(\alpha,\taug,\beta,\tauh). 
%\end{align}
The objective function $\Dc$ above is convex-concave. Also, the constraint sets over $\alpha$ and $\beta$ are compact. Furthermore, the optimization of $\Dc$ over $\taug$ and $\tauh$ is separable. With these and an application of Sion's minimax theorem, the order of inf--sup between the four optimization variables can be flipped arbitrarily without affecting the outcome. Thus, for example,
$$\inf_{\alpha\in\Ac}\sup_{\substack{\tauh>0 \\ \beta\geq0}}\inf_{\taug>0} \Dc(\alpha,\taug,\beta,\tauh) = \inf_{\substack{\alpha\in\Ac\\\taug>0}}\sup_{\substack{\tauh>0 \\\geq \beta\geq0}}\Dc(\alpha,\taug,\beta,\tauh)$$
The same is of course true for the corresponding random optimizations (also, Lemma \ref{lem:scalar}(iii)).
\subsection{Auxiliary Lemmas}\label{sec:aux_Proofs}
% ==================================================
% ==================================================

%

%%
%
%\begin{proof}[Proof of Lemma \ref{lem:argmin_compact}]
%Fix any $\eps>0$. Let $x_*\in\Sc$ be such that $F(x_*)=\min_{x\in\Sc}F(x)$. $\{F_n\}$ are convex and they converge point wise to $F$ in $\Sc$. This implies uniform convergence in compact sets by the Convexity lemma \cite{?}. That is, for all $\eta>0$, there exists $N:=N(\eta,\eps)$ such that for all $n>N$, the event $\sup_{x\in K}|F_n(x)-F(x)|<\eps$ occurs w.p. at least $1-\eta$. For all $n\geq N$, conditioning on this event, we have
%$$
%\min_{x\in\Sc} F_n(x)\leq F_n(x_*) < F(x_*) + \eps.
%$$
%Also, 
%$$
%\min_{x\in\Sc} F_n(x) =  F_n(x_n) \geq F(x_n) - \eps \geq F(x_*) - \eps .
%$$
%Since this holds for all $\eps>0$, it completes the proof.
%\end{proof}
%
%%

%
\vp

\begin{proof}[Proof of Lemma \ref{lem:argmin_open}]
%Convexity 
%\noindent{(i)~\underline{$\inf_{x>0}F(x)>-\infty$}}. 
First,  convexity is preserved by point wise limits, so that $F(x)$ is also convex. Using this and level-boundedness condition (b) of the lemma, it is easy to show that $\inf_{x>0}F(x)>-\infty$. Since $F$ is proper and (lower) level-bounded, the only way $\inf_{x>0}F(x)=-\infty$ is if  $\lim_{x\rightarrow 0}F(x)=-\infty$. But, this is not possible as follows: Fix  $0<x_1<x_2<x_3$. Then, for any $0<x<x_1$ and $\theta:=\frac{x_3-x_2}{x_3-x}$, convexity gives
$$
F(x) \geq \frac{1}{\theta} F(x_2) - \left(1-\frac{1}{\theta}\right) F(x_3)\geq -\frac{x_3-x_1}{x_3-x_2} | F(x_2) | -  \frac{x_2-x_1}{x_3-x_2} | F(x_3) |.
$$

\vp
Next, we show that for sufficiently small $\eps>0$, there exist $ x_0>x_\eps>0$:  
\begin{align}\label{eq:bhma_1} 
\inf_{x>0}F(x) \leq F(x_\eps) \leq \inf_{x>0}F(x) + \eps \quad\text{ and }\quad F(x_\eps)<F(x_0).
\end{align}
%\noindent{(ii)~\underline{For sufficiently small $\eps>0$, $\exists x_0>x_\eps>0$:  
%%\begin{align}\label{eq:bhma_1}
%$F(x_\eps) \leq \inf_{x>0}F(x) + \eps$ and $F(x_\eps)<F(x_0)$.
%%\end{align}
%}}
We show the claim for all $0<\eps<\eps_1:= ({F(z)-\inf_{x>0}F(x)})$. Since $\inf_{x>0}F(x)$ is finite, there exists $x_\eps>0$ such that $F(x_\eps)-\eps\leq \inf_{x>0} F(x)$. Without loss of generality, $x_\eps<z$. Pick any $x_0>z$. For the shake of contradiction, assume $F(x_0)\leq F(x_\eps)$. Then, by convexity, for some $\theta\in(0,1)$
$$
F(z) \leq \theta F(x_\eps) + (1-\theta) F(x_0) \leq F(x_\eps)\leq \inf_{x>0} F(x) + \eps  < F(z).
$$
Thus, $F(x_\eps)< F(x_0)$.

\vp
In  order to establish the desired, it suffices that for all arbitrarily small $\delta>0$, w.p.a. 1,
\begin{align}\label{eq:goal}
|\inf_{x>0} F_n(x) -  \inf_{x>0} F(x)|<\delta.
\end{align}
%with probability at least $1-\eta$ for sufficiently large $n$.

%Fix any $\eta>0$ and
Fix some $0<\eps<\min\{\eps_1,\delta\}$ such that \eqref{eq:bhma_1} holds, and, also some
\begin{align}\label{eq:eps}
0<\eps'<\min\{(F(x_0)-F(x_\eps))/4, \delta/4,\delta-\eps\}.
\end{align} 
Let $K=[a,b]\subset(0,\infty)$ be compact subset such that $a<x_\eps<x_0\leq b$
and $a=\frac{\delta-2\eps'}{2\delta-\eps'}x_\eps$
. The functions $\{F_n\}$ are convex and they converge point wise to $F$ in the open set $(0,\infty)$. This implies uniform convergence in compact sets by the Convexity lemma \uniform,\cite[Lem.~7.75]{liese2008statistical}. That is, there exists sufficiently large $N_1$ such that the event 
\begin{align}\label{eq:uni_K}
\sup_{x\in K}|F_n(x)-F(x)|<\eps'
\end{align}
 occurs w.p.a. 1, for all $n>N_1$. In this event,
\begin{align}%\label{eq:upper}
\inf_{x>0} F_n(x) \leq F_n(x_\eps) < F(x_\eps) + {\eps'} \leq \inf_{x>0} F(x) + \eps + \eps'\leq\inf_{x>0} F(x)  + \delta\nn
\end{align}
It remains to prove the other side of \eqref{eq:goal}. In what follows, take $n\geq N_1$ and condition  on the high probability event in \eqref{eq:uni_K}.

Let us first show level-boundedness of $F_n$. Consider the event 
%\begin{align}\label{eq:low_prob}
$\inf_{x>x_0} F_n(x) < \inf_{x\leq x_0} F_n(x).$
%\end{align}
If this happens, then, 
$
\inf_{x>x_0} F_n(x) < F_n(x_\eps),
$
in which case there exists (by continuity of $F_n$), $x_n>x_0$ such that 
%\begin{align}\label{eq:big1}
$F_n(x_n) < F_n(x_\eps).$
%\end{align}
But then, convexity implies that for some $0<\theta_n<1$,
\begin{align}
F_n(x_0)\leq \theta_n F_n(x_n) + (1-\theta_n) F_n(x_\eps)< F_n(x_\eps) \leq F(x_\eps) + \eps' < F(x_0)-\eps',
\end{align}
where we also used \eqref{eq:uni_K} and \eqref{eq:eps}. Of course, this contradicts \eqref{eq:uni_K}. Thus, 
\begin{align}\label{eq:bound}
\inf_{x>0} F_n(x) = \inf_{x\leq x_0} F_n(x).
\end{align}
%$\inf_{x>x_0}F_n(x)\leq \inf_{x\leq x_0}F_n(x)$

Using \eqref{eq:bound}, convexity and properness of $\{F_n\}$, it can be shown that $\inf_{x>0}F_n(x)>-\infty$. The argument is the same as the one used in the beginning of the proof for $F$, thus is omitted for brevity.

Overall, for all $n>N_1$, conditioned on \eqref{eq:uni_K}, there is some $0<x_n\leq x_0$ such that 
\begin{align}\label{eq:dew_1}
\inf_{x>0}F_n(x) \geq F_n(x_n) -  \eps'.
\end{align}

If $a\leq x_n\leq b$, then a direct application of \eqref{eq:uni_K} gives the desired
$$
 F_n(x_n) \geq F(x_n) - \eps' \geq \inf_{x>0}F(x) - \eps' \Rightarrow \inf_{x>0} F_n(x) \geq \inf_{x>0}F(x) - 2\eps' \geq  \inf_{x>0}F(x) - \delta.
$$

Next, assume that $0< x_n< a$.
% Easily, 
% \begin{align}\label{eq:case2_1}
%\inf_{x>0} F_n(x) \leq F_n(x_\eps) < F(x_\eps) + {\eps}.
%\end{align}
  There exists $\theta_n\in(0,1)$ such that $\theta_nx_n+(1-\theta_n)x_\eps = a$. In fact, 
\begin{align}\label{eq:dew_2}
 \theta_n = \frac{x_\eps-a}{x_\eps-x_n}\geq (1-a/x_\eps)=\frac{\delta-2\eps'}{2\delta-\eps'}.
\end{align}
 Then, by convexity and \eqref{eq:uni_K}, $F_n(a) \leq \theta_n F_n(x_n) + (1-\theta_n) F_n(x_\eps) $. Rearranging and using \eqref{eq:uni_K}
\begin{align*}
F_n(x_n) &\geq \frac{1}{\theta_n}F_n(a) - \frac{1-\theta_n}{\theta_n} F_n(x_\eps) \\ &\geq  \frac{1}{\theta_n}(F(a) - \eps') - \frac{1-\theta_n}{\theta_n} (F(x_\eps) + \eps')
\\ &\geq  \frac{1}{\theta_n}\left(\inf_{x>0}F(x) - \eps\right) - \frac{1-\theta_n}{\theta_n} \left(\inf_{x>0}F(x) + \delta\right)
%\\ &\geq  \inf_{x>0}F(x)  - \frac{3-\theta_n}{\theta_n} \eps \geq F(x_\eps) -\left( \frac{3}{\zeta} - 1\right) \eps 
\end{align*}
Combining this with \eqref{eq:dew_1} and \eqref{eq:dew_2}, yields the desired
$
\inf_{x>0}F_n(x_n) \geq \inf_{x>0}F(x) - \delta .
$
\end{proof}

\begin{lem}(Level-bounded convex fcns)\label{lem:lv_cvx}
Let $F:(0,\infty)\rightarrow\R$ be convex. Then, the following two statements are equivalent:
\begin{enumerate}[(a)]
\item There exists $z>0$ such that $F(x)>\inf_{x>0}F(x)$ for all $x\geq z$.
\item $\lim_{x\rightarrow\infty}F(x) = +\infty$.
\end{enumerate}
\end{lem}
\begin{proof}
\noindent{\underline{(a)$\Rightarrow$(b):}}~
Clearly, there exists $0<x_0<z$, such that $F(z)>F(x_0)$. Then, by convexity, for all $x>z$ it holds
$$
F(x) \geq F(z) + \underbrace{\frac{F(z)-F(x_0)}{z-x_0}}_{>0}(x-z).
$$
Taking limits of $x\rightarrow\infty$ on both sides above, proves the claim.

\noindent{\underline{(a)$\Leftarrow$(b):}}~ A a proper functions, $F$ has a nonempty domain in $(0,\infty)$. Hence, $\inf_{x>0}F(x)<\infty$ and can choose some $M>\inf_{x>0}F(x)$. From (b), there exists $z>0$ such that $F(x)\geq M$ for all $x\geq z$, as desired.
\end{proof}

\begin{lem}[Saddle-points]\label{lem:sp_compact}
For a convex-concave function $F:\R\times\R\rightarrow\R$, consider the minimax optimization $\inf_{x}\sup_{y}F(x,y)$. Let $C,D$ be compact subsets such that there exists at least one saddle point $(x_*,y_*)\in C\times D$. Then, 
$$
\inf_{x}\sup_{y}F(x,y) = \inf_{x\in C}\sup_{y\in D}F(x,y).
$$
\end{lem}
\begin{proof}
First observe that, 
\begin{align*}%\label{eq:kata_1}
\inf_{x}\sup_{y}F(x,y) = \inf_{x\in C}\sup_{y} F(x,y)
\end{align*}
Since $F$ has a saddle-point, the LHS above is equal to $\sup_{y}\inf_{x}F(x,y)$ \cite[Lem.~36.2]{Roc70}. Also, from Sion's minimax theorem, the RHS is equal to $\sup_{y}\inf_{x\in C} F(x,y)$. Thus, it suffices to prove that
\begin{align*}%\label{eq:kata_2}
\sup_{y}\inf_{x\in C} F(x,y) = \sup_{y\in D}\inf_{x\in C} F(x,y).
\end{align*}
Clearly, this holds with a ``$\geq$" sign. To prove equality, let $(x_*,y_*)$ be a saddle point. Then,
$$
\sup_{y}\inf_{x\in C} F(x,y) = \inf_{x\in C}\sup_{y}f(x,y)\leq  \sup_{y} f(x_*,y) \leq f(x_*,y_*) = \sup_{y\in D}\inf_{x\in C} F(x,y).
$$
\end{proof}

%\begin{lem}[Continuity of min]\label{lem:cont_comp}
%Let $f:\R\times\R\rightarrow\R$ be continuous and $Y$ a compact set. Then, the function $g(x)=\min_{y\in Y}f(x,y)$ is continuous over $x$.
%\end{lem}
%\begin{proof}
%The proof is standard and is thus omitted for brevity.
%\end{proof}

\begin{lem}\label{lem:jointConv}
The function $h(\alpha,\tau,\vb)=\frac{1}{2\tau}\| \alpha\x + \z - \vb \|_2^2$ is jointly convex in its arguments.
\end{lem}
\begin{proof}
The function $\|\al \x-\vb\|_2^2$ is trivially jointly convex in $\al$ and $\vb$. So its perspective function which is  $\frac{1}{\tau}\| \alpha\x - \vb \|_2^2$ is also jointly convex in all its arguments, same as its shifted version which is $h(\alpha,\tau,\vb)$.
\end{proof}
\begin{lem}\label{lem:Moreau}
Let $f:\R^n\rightarrow\R$ be convex. Then,
\begin{enumerate}[(i)]
\item $\prox{f}{\x}{\tau} + \tau\cdot\prox{f^*}{\x/\tau}{\tau^{-1}}=\x$,
\item $\env{f}{\x}{\tau}+\env{f^*}{\x/\tau}{1/\tau} = \frac{\|\x\|^2}{2\tau}$.
\end{enumerate}
\end{lem}

%\begin{lem}[Convexity Lemma]\label{lem:convexity}
%Let $\Oc\in\R^n$ be open and convex, $F_n(\x)$ be ...
%\end{lem}
%\begin{proof}
%See for example \cite[Lem.~7.75]{} or \cite[]{AG1982}.
%\end{proof}

\section{Proofs for Separable M-Estimators}\label{app:sep}

\subsection{Satisfying Assumptions \ref{ass:tech} and \ref{ass:prop}}\label{app:ass_sep}
%\subsection{Proof of Lemma \ref{lem:ass_sep}}

\subsubsection{Proof of Lemma \ref{lem:ass_sep}}
%For convenience, throughout the proof, we denote
Recall,
$$
\ell_+^\prime(v) = \max_{s\in\partial\ell(v)} |s|.
$$
and that \eqref{eq:ass_main} gives for all $c\in\R$,
\begin{align}\label{eq:we have}
\Exp|\ell_+^\prime(c G + Z)|<\infty\quad\text{ and }\quad \Exp|\ell_+^\prime(c G + Z)|^2<\infty.
\end{align}

We make repeated use of Lemma \ref{lem:Mconvex} on properties of the Moreau envelope function.

\vp
\noindent$\bullet$~
First, we show that 
\begin{align}\label{eq:1a_prepre}
\Exp\left[\Big| \frac{\partial \env{\ell}{\alpha G + Z}{\tau}}{\partial\tau} \Big| \right]<\infty, \quad \text{ for all } \alpha\in\R, \tau>0.
\end{align}
From \eqref{eq:e_2}
$
\Big|\frac{\partial \env{\ell}{\alpha G + Z}{\tau}}{\partial\tau}\Big| \leq |\ell_+^\prime( \prox{\ell}{\alpha G + Z}{\tau})|^2.
$
Lemma \ref{lem:Mconvex}(viii) shows that this is no larger than $ |\ell_+^\prime( {\alpha G + Z})|^2$. Then, \eqref{eq:1a_prepre} follows from \eqref{eq:we have}.

\vp
\noindent$\bullet$~
It is also useful to prove 
\begin{align}\label{eq:1a_pre}
\Exp\left[| \ell({\alpha G + Z}) - \ell(Z) | \right]<\infty, \quad \text{ for all } \alpha\in\R.
\end{align}
From convexity of $\ell$,
\begin{align*}
| \ell({\alpha G + Z}) - \ell(Z) | \leq \max\{|\ell_+^\prime({\alpha G + Z})|,|\ell_+^\prime(Z)|\}\cdot |\alpha G| \leq  \left(|\ell_+^\prime({\alpha G + Z})|+|\ell_+^\prime(Z)|\right)\cdot |\alpha G|,
\end{align*}
and the desired follows by taking expectations and applying \eqref{eq:we have} for $c=\alpha$ and $c=0$.

\vp 
\noindent$\bullet$~Let us now show 
\begin{align}\label{eq:1a}
\Exp\left[ |\env{\ell}{\alpha G + Z}{\tau} - \ell(Z) | \right]<\infty, \quad \text{ for all } \alpha\in\R, \tau>0
\end{align}
We have,
$
|\env{\ell}{\alpha G + Z}{\tau} - \ell(Z) | \leq |\env{\ell}{\alpha G + Z}{\tau} - \ell(\alpha G + Z) |  + |\ell(\alpha G + Z)- \ell(Z) |.
$
In view of \eqref{eq:1a_pre}, it suffices for \eqref{eq:1a} to show integrability of the first term. We argue as follows 
\begin{align*}
|\env{\ell}{\alpha G + Z}{\tau} - \ell(\alpha G + Z) | &= \lim_{\rho\rightarrow0}|\env{\ell}{\alpha G + Z}{\tau} - \env{\ell}{\alpha G + Z}{\rho} | \\&=
 \lim_{\rho\rightarrow0} \Big|\frac{\partial\env{\ell}{\alpha G + Z}{\tau}}{\partial\tau}\Huge|_{\tau=\xi(\rho)} \Big | \cdot |\tau - \rho |.
\end{align*}
It remains to take expectations of both sides and apply the argument below \eqref{eq:1a_prepre} to yield \eqref{eq:1a}.

\vp 
\noindent$\bullet$~ \underline{Assumption \ref{ass:tech}(a).} 
We have
$
\frac{1}{m}\{\env{\loss}{\alpha\g+\z}{\tau}-\loss(\z)\} = 
\frac{1}{m}\sum_{j=1}^m\left(\env{\ell}{\alpha\g_j+\z_j}{\tau}-\ell(\z_j)\right).
$
Then rom the WLLN (e.g. \cite[Thm.~2.2.9]{dukeProb}) the expression above converges in probability to 
\begin{align}\label{eq:Lm_sep}
\Lm{\alpha}{\tau} = \E\left[\env{\ell}{\alpha G+Z}{\tau}-\ell(Z)\right],
\end{align}
where we have also used \eqref{eq:1a} to verify integrability. 

\vp 
\noindent$\bullet$~ \underline{Continuity and convexity of $L$.} The Moreau envelope function is convex in its arguments (see Lemma \ref{lem:Mconvex}(ii)). Convexity is preserved under affine transformations and nonnegative weighted sums; thus, $L(\alpha,\tau)$ is jointly convex in $\alpha,\tau$. Continuity then follows as a consequence of convexity \cite[Thm.~10.1]{Roc70}.

\vp 
\noindent$\bullet$~ \underline{Assumption \ref{ass:prop}(c).} 
To compute $\lim_{\tau\rightarrow+\infty}\Exp\left[\env{\ell}{\alpha G+Z}{\tau}-\ell(Z)\right]$, we first apply the Dominated Convergence Theorem to pass the limit inside the expectation. This is justified since \eqref{eq:1a} shows integrability, and the limit exists as follows (see Lemma \ref{lem:Mconvex}(vii))
\begin{align}
\lim_{\tau\rightarrow+\infty}\env{\ell}{\alpha G+Z}{\tau} = \min_v \ell(v) = 0,\nn
\end{align}
for all $\alpha,\tau>0$. Taking expectation of this shows $\lim_{\tau\rightarrow+\infty}L(\alpha,\tau)=-L_0$, where $L_0=\Exp[\ell(Z)]$ by the WLLN. 
%%
%\vp 
%\noindent$\bullet$~ \underline{Assumption \ref{ass:prop}(b).} 
Also, we need to show that if $\Exp[\ell(Z)]<\infty$, then $\Exp\left[\env{\ell}{\alpha G+Z}{\tau}\right]\geq 0$. This follows easily since $\env{\ell}{\alpha G+Z}{\tau}\geq \min_{v}\ell(v)=0$. Finally, the property $L(\alpha,\tau)\geq\lim_{\tau\rightarrow\infty}L(\alpha,\tau)$ follows by the non increasing nature of $e_\ell$ with respect to $\tau$ (cf.  \ref{lem:Mconvex}(v)).

\vp 
\noindent$\bullet$~ \underline{Assumption \ref{ass:prop}(d).} 
%Suppose $\Exp[\ell(Z)]=\infty$. 
%Lemma \ref{lem:4} shows that 
%\begin{align}\label{eq:4}
%\lim_{\tau\rightarrow+\infty}\frac{L(\alpha,\tau)}{\tau} = 0.
%\end{align}
If $\lim_{\tau\rightarrow+\infty}L(\alpha,\tau)<\infty$, the claim is immediate. Otherwise, we apply de l'hospital rule and \eqref{eq:der_L} to get
\begin{equation}
\lim_{\tau\rightarrow\infty}\frac{\E[\env{\ell}{\GZ}{\tau}-\ell(Z)]}{\tau}=\lim_{\tau\rightarrow\infty}\frac{\partial}{\partial \tau}\E[\env{\ell}{\GZ}{\tau}-\ell(Z)]\nn
\end{equation}
An application of the Dominated Convergence Theorem in Lemma \ref{lem:strict}(i) ,shows that we can interchange the order of differentiation and expectation above. 
We will prove that
\begin{align}\label{eq:will we}
\lim_{\tau\rightarrow\infty}\frac{\partial}{\partial \tau}(\env{\ell}{\GZ}{\tau}-\ell(Z))=0
\end{align}
for all $G$ and $Z$. Then, we can also utilize dominated convergence theorem to pass the limit in the expectation and conclude with the desired.
% as follows:
%\begin{align}
%\E\left[\lim_{\tau\rightarrow\infty}\frac{\partial}{\partial \tau}(\env{\ell}{\GZ}{\tau}-\ell(Z))\right]&=\int_x \E\left[\big(\lim_{\tau\rightarrow\infty}\frac{\partial}{\partial \tau}(\env{\ell}{\GZ}{\tau}-\ell(Z))\big)\big|\GZ=x\right]\nonumber\\
%&=\int_x 0\cdot p_{\GZ}(x)dx=0\nn
%\end{align}
%Where $p_{\GZ}(x)$ denotes the probability density function of $\GZ$.\\

From standard properties of the Moreau envelopes (cf. \eqref{eq:e_2}),
\begin{equation}
\frac{\partial}{\partial \tau}(\env{\ell}{\GZ}{\tau}-\ell(Z))=\frac{1}{-2\tau^2}(\GZ-\prox{\ell}{\GZ}{\tau})^2\nn
\end{equation}
%
%Note that due to the assumption \ref{ass:sep}, the expected value exists for any value of $\tau$ and $\GZ$. So as we discussed in \ref{lem:strict}, we can exchange the expected value and the derivation. Then if $\lim_{\tau\rightarrow\infty}\frac{\partial}{\partial \tau}(|\env{\ell}{\GZ}{\tau}-|Z||)=0$ for all $\GZ$, we can also utilize dominated convergence theorem to exchange the expected value and the limit to get the result. From sections \ref{lem:strict} and \ref{lem:Mconvex},
%Thus it suffices to show $\lim_{\tau\rightarrow\infty}\frac{1}{\tau}(\GZ-\prox{\ell}{\GZ}{\tau})=0$ which implies $\lim_{\tau\rightarrow\infty}\frac{\partial}{\partial \tau}(|\env{\ell}{\GZ}{\tau}-|Z||)=0$ for fixed $\GZ$ and we will have
%\begin{align}
%\E[\lim_{\tau\rightarrow\infty}\frac{\partial}{\partial \tau}(|\env{\ell}{\GZ}{\tau}-|Z||)]&=\int_x \E[\big(\lim_{\tau\rightarrow\infty}\frac{\partial}{\partial \tau}(|\env{\ell}{\GZ}{\tau}-|Z||)\big)\big|\GZ=x]\nonumber\\
%&=\int_x 0.f_{\GZ}(x)dx=0
%\end{align}
%Where $f_{\GZ}(x)$ is the probability density function of $\GZ$.\\
Thus, it suffices to prove $\lim_{\tau\rightarrow\infty}\frac{1}{\tau}(x-\prox{\ell}{x}{\tau})=0$ for all $x$. This is shown in Lemma \ref{lem:Mconvex}(vii).

\vp 
\noindent$\bullet$~ \underline{Assumption \ref{ass:prop}(b).}
We apply the Dominated Convergence Theorem to compute $\lim_{\tau\rightarrow0^+}\E[\env{\ell}{\al G+Z}{\tau}-\ell(Z)]$ and exchange limit and expectation. Then, because $\lim_{\tau\rightarrow0^+}\env{\ell}{\al G+Z}{\tau}=\ell(\al G+Z)$ we have
$$
\lim_{\tau\rightarrow0^+}\E\left[\env{\ell}{\al G+Z}{\tau}-\ell(Z)\right]=\E\left[\lim_{\tau\rightarrow0^+}\env{\ell}{\al G+Z}{\tau}-\ell(Z) \right]=\E\left[ \ell(\al G+Z)-\ell(Z) \right]<\infty,
$$
Boundedness follows from \eqref{eq:1a}. The same argument shows that 
$$
\lim_{\tau\rightarrow0^+}L(0,\tau)=\lim_{\tau\rightarrow0^+}\E\left[\env{\ell}{Z}{\tau}-\ell(Z)\right]=\E\left[\lim_{\tau\rightarrow0^+}\env{\ell}{Z}{\tau}-\ell(Z) \right]=\E\left[ \ell(Z)-\ell(Z) \right]=0.
$$
Finally, to compute $\lim_{\tau\rightarrow0^+}L_2(0,\tau)$, we apply Dominated Convergence Theorem twice as was done for the proof of Assumption \ref{ass:prop}(d).  With this we have,
$$
\lim_{\tau\rightarrow0^+}L_2(0,\tau)=\E\left[ \lim_{\tau\rightarrow 0^+}\frac{\partial}{\partial\tau}\left(\env{\ell}{\al G +Z}{\tau}-\ell(Z) \right)\big|_{\al=0}\right]=-\frac{1}{2}\E\left[\lim_{\tau\rightarrow0^+}\left(\ell^\prime_{ \prox{\ell}{Z}{\tau},\tau}\right)^2 \right]\leq0.
$$
The second equality above follows by Lemma \ref{lem:Mconvex}(iii) (please see \eqref{eq:opt} for the notation $\ell^\prime_{\chi,\tau}$). Besides, due to lemma \ref{lem:Mconvex}(viii), $(\ell^\prime_{\prox{\ell}{Z}{\tau},\tau})^2\leq(\ell^\prime_+(Z))^2 $ which implies
$$
-\E\left[\lim_{\tau\rightarrow0^+}\left(\ell^\prime_{ \prox{\ell}{Z}{\tau},\tau}\right)^2 \right]\geq-\E\left[\lim_{\tau\rightarrow0^+} (\ell_+^\prime(Z))^2 \right]=-\E\left[ (\ell_+^\prime(Z))^2 \right]>-\infty.
$$
Boundedness follows by \eqref{eq:we have}.
%\end{proof}

%If $\lim_{\tau\rightarrow+\infty}L(\alpha,\tau)<\infty$, the claim is immediate. Otherwise, we apply de l'hospital rule and \eqref{eq:der_L} to find
%\begin{align}
%\lim_{\tau\rightarrow+\infty}\frac{L(\alpha,\tau)}{\tau} =
%\lim_{\tau\rightarrow+\infty}\frac{\partial L(\alpha,\tau)}{\partial\tau} = 
%\Exp\left[\lim_{\tau\rightarrow+\infty} \right]
%\end{align}

\vp
\subsubsection{Proof of Lemma \ref{lem:ass_sep2}}

\vp 
\noindent$\bullet$~ \underline{Assumption \ref{ass:tech}(a).}~
Assumption \ref{ass:tech}(a) is satisfied for $F(c,\tau)=\E\left[\env{f}{c H+X_0)}{\tau}-f(X_0)\right].$ The proof is exact same as in Lemma  \ref{lem:ass_sep}.

\vp 
\noindent$\bullet$~ \underline{$\lim_{\tau\rightarrow0^+}F(\tau,\tau)=0$.}~ This will follow from continuity of the Moreau envelope. In particular, 
%we use the fact that 
%\begin{align*}%\label{eq:Mco}
%\env{f}{x_n}{\tau_n}\rightarrow f(x)
%\end{align*}
% whenever $x_n\rightarrow x$ while $\tau_n\rightarrow 0^+$ in such a way  that the sequence $\{{|x_n-x|}/{\tau_n}\}_{n\in\mathbb{N}}$ is bounded \cite[Thm.~1.25]{RocVar}. In our case, this gives 
using Lemma \ref{lem:Mconvex}(ix), we find that 
 for all $H,X_0$:
 $$
 \lim_{\tau\rightarrow0}\env{f}{\tau H + X_0}{\tau} = f(X_0).
 $$
Then, the desired claim follows from this and an application of the Dominated Convergence Theorem.\\

\vp 
\noindent$\bullet$~ \underline{$\lim_{c\rightarrow\infty}\frac{c^2}{2\tau}- \E[\env{f}{cH+X_0}{\tau}-f(X_0)]=\infty$.}~
%Finally, it remains to prove 
%\begin{equation*}
%\lim_{c\rightarrow\infty}\frac{c^2}{2\tau}- \E[\env{f}{cH+X_0}{\tau}-f(X_0)]=\infty
%\end{equation*}
We have
\begin{align}\label{eq:limit_fs_1}
&\frac{c^2}{2\tau}-\E[\env{f}{cH+X_0}{\tau}-f(X_0)]=\E[\frac{(c H+X_0)^2}{2\tau}- \env{f}{cH+X_0}{\tau}]+\E[f(X_0)-X_0^2]\nonumber\\
&=\E[\env{f^*}{(c H+X_0)/\tau}{1/\tau}]+\E[f(X_0)-X_0^2]\nonumber\\
&=\frac{1}{2}\E[\env{f^*}{(c H+X_0)/\tau}{1/\tau}\big|H>0]+\frac{1}{2}\E[\env{f^*}{(c H+X_0)/\tau}{1/\tau}\big|H<0]
+\E[f(X_0)-X_0^2]\nonumber\\
&\geq\frac{1}{2}\env{f^*}{\E[(c H+X_0)/\tau\big| H>0]}{1/\tau}+\frac{1}{2}\env{f^*}{\E[(c H+X_0)/\tau\big|H<0]}{1/\tau}
+\E[f(X_0)-X_0^2]\nonumber\\
&=\frac{1}{2}\env{f^*}{c/\tau\sqrt\frac{2}{\pi}+\E[X_0]}{1/\tau}+\frac{1}{2}\env{f^*}{-c/\tau\sqrt\frac{2}{\pi}+\E[X_0]}{1/\tau}+\E[f(X_0)-X_0^2].
\end{align}
The second equality above follows from Lemma \ref{lem:Moreau}. For the inequality, $\env{f^*}{c}{\tau}$ is convex in $c$, thus it follows from Jensen's inequality.
From \eqref{eq:limit_fs_1}, observing that $\E[f(X_0)-X_0^2]>-\infty$ and $|\E[X_0]|<\infty$ by boundedness of $\Exp[X_0^2]$ and non-negativity of $f$, it suffices to show that 
$$
\lim_{|c|\rightarrow\infty}\env{f^*}{c/\tau}{1/\tau}=\infty.
$$
First, assume that $f(x)$ is defined for some positive value $a>0$ and $f(a)<\infty$, then
\begin{equation}\label{eq:limit_fs_2}
\forall M,\quad \forall x>X_M=\frac{f(a)}{a}+\frac{M}{a}:\quad f^*(x)=\max_{y}xy-f(y)\geq ax-f(a)>M.
\end{equation}
Which means that $\lim_{x\rightarrow\infty}f^*(x)=\infty$. Now in order to show that the limit in \eqref{eq:limit_fs_1} goes to infinity we prove that
\begin{align}
\forall M\quad \forall x>\tau(X_M+\sqrt{2 M/\tau})&\quad \forall v, \quad \frac{\tau}{2}(x/\tau-v)^2+f^*(v)>M\nonumber.\\ 
&\Longleftrightarrow\forall u\quad \frac{\tau}{2}u^2+f^*(u+x/\tau)>M
\end{align}
This is easy to show. For the cases that $|u|>\sqrt{2M/\tau}$ we have
$$
 \frac{\tau}{2}u^2+f^*(u+x/\tau)>M+f^*(u+x/\tau)\geq M.
$$
Note that $f(0)=0$ implies $f^*(x)\geq0$ for all $x$. On the other hand, for the cases that $|u|\leq\sqrt{2M/\tau}$, 
$$
x/\tau+u\geq x/\tau-|u|\geq x/\tau-\sqrt{2M/\tau}>X_M.
$$
Thus due to \eqref{eq:limit_fs_2},
$$
 \frac{\tau}{2}u^2+f^*(u+x/\tau)> \frac{\tau}{2}u^2+M\geq M,
$$
which shows that $\lim_{c\rightarrow\infty}\env{f^*}{c/\tau}{1/\tau}=\infty$.\\
On the other hand, if $f(x)$ is also defined for some negative value $a<0$ and $f(a)<\infty$, the same set of arguments proves that $\lim_{c\rightarrow-\infty}f^*(c)=\infty$ and also $\lim_{c\rightarrow-\infty}\env{f^*}{c/\tau}{1/\tau}=\infty$.
\subsection{Strict Convexity of the Expected Moreau Envelope}\label{sec:Moreau}

In this section, we prove Lemmas \ref{lem:smooth_ell} and \ref{lem:strict_ell}. We have combined the statements in Lemma \ref{lem:strict} below.

\begin{lem}[Lemmas \ref{lem:smooth_ell} and \ref{lem:strict_ell}]
\label{lem:strict}
Let $\ell:\R\rightarrow\R$ be a proper, closed, convex function, $G~\sim\Nn(0,1)$ and $Z\sim p_z$ such that \eqref{eq:ass_main} holds. The function $L:\R\times\R_{>0}\rightarrow\R$:
\begin{align*}
L(\al,\tau)&:=
%\E_{G,Z}\left[\min_{v}  \frac{1}{2\tau}(\al G+ Z -v)^2+\ell(v)\right]= 
\E_{G,Z}\left[\env{\ell}{\al G+ Z}{\tau} - \ell(Z)\right].% \\&= \E\left[\frac{1}{2\tau}(\al\G+\Z -\prox{\ell}{\al G +Z}{\tau})^2+\ell(\prox{\ell}{\al G +Z}{\tau})\right]\nn
\end{align*}
has the following properties:
\begin{enumerate}[(i)]
\item It is differentiable with
\begin{align}\label{eq:der_L}
\frac{\partial L}{\partial\alpha} = \Exp\left[ \frac{\partial \env{\ell}{\al G + Z}{\tau}}{\partial \alpha} \right]\quad\text{ and }\quad
\frac{\partial L}{\partial\tau} = \Exp\left[ \frac{\partial \env{\ell}{\al G + Z}{\tau}}{\partial \tau} \right],
\end{align}
\item If the conditions (a) and (b) of Lemma \ref{lem:strict_ell} also hold, then it is jointly strictly convex in $\R_{>0}\times\R_{>0}$.
\item If $\ell(x)\geq\ell(0)=0$ and $\ell(x_+)>0$ for some $x_+>0$, then, the function $F(\alpha):=\lim_{\tau\rightarrow0^+}L(\alpha,\tau)=\Exp\left[ \ell( \alpha G + Z ) - \ell(Z) \right]$ is strictly convex in $\alpha>0$.

\end{enumerate}
\end{lem}
\begin{proof}

We make repeated use of the properties of the Moreau envelope function as listed in Lemma \ref{lem:Mconvex}. Also, we use the same notation as in that lemma; in particular, recall   \eqref{eq:opt}, \eqref{eq:e_1} and \eqref{eq:e_2}. For ease of reference we summarize the notation used throughout this section below:
%\begin{subequations}
\begin{align}\nn
\vh_{\chii,\tau} := \prox{\ell}{\chii}{\tau}, \quad \ell'_{\chii,\tau} := \frac{1}{\tau}(\chii-\vh_{\chii,\tau}), \quad
\end{align}
\begin{align}\nn
E_1(\alpha,\tau):=\frac{\partial \env{\ell}{\al G + Z}{\tau}}{\partial \alpha}, \quad E_2(\alpha,\tau):=\frac{\partial \env{\ell}{\al G + Z}{\tau}}{\partial \tau}.
\end{align}

\vp
\noindent{\emph{\textbf{(i):}}}
The claim follows by the Dominated Convergence Theorem, since the following hold:
\begin{itemize}
\item  $\env{\ell}{\al G+Z}{\tau}$ is
continuously differentiable with respect to both $\alpha$ and $\tau$ (cf. Lemma \ref{lem:Mconvex}(iii)),
\item In Section \ref{app:ass_sep} (see \eqref{eq:1a}) we use \eqref{eq:ass_main} to show  that $\Exp[|\env{\ell}{\al G + Z}{\tau}-\ell(Z)|]<\infty$ for all $\al$ and $\tau>0$.
\item for all $\alpha\in\R$ and $\tau>0$:
\begin{align*}
\E\left[\left| E_1(\al,\tau)\right|\right] &=
\frac{1}{\tau}\E\left[\left| \al G + Z -  \prox{\ell}{\al G + Z}{\tau} \right| \cdot \left| G \right|\right] \\
&\leq \frac{1}{\tau}\sqrt{\E\left[\left| \al G + Z -  \prox{\ell}{\al G + Z}{\tau} \right|^2\right] } = 
\sqrt{\E\left[\left|E_2(\al,\tau)\right|\right]},
\end{align*} 
where we have used Lemma \ref{lem:Mconvex}(iii), the \Cauchy inequality. In Section \ref{app:ass_sep} (see \eqref{eq:1a_prepre}) we use \eqref{eq:ass_main} to show  that 
% and Assumption \ref{ass:sep}(ii). Similarly, it can be shown that under Assumption \ref{ass:sep}(ii), it holds
$\E\left[\left|E_2(\al,\tau)\right|\right]<\infty$.
\end{itemize}

\vp
\noindent{\emph{\textbf{(ii):}}}
For any $\al>0,\tau>0$, it suffices to show that
\begin{align}\label{eq:2show_sc}
\Gamma(x,y):=L(\al+x,\tau+y) - L(\al,\tau) - L_1(\al,\tau)x - L_2(\al,\tau)y > 0, \quad \text{for all } x\in\R,y>-\tau,
\end{align}
where we use numerical subscript notation to denote derivation with respect to the corresponding argument, i.e. $L_1=\partial L/\partial\al$ and $L_2=\partial L/\partial\tau$.

Observe that $\Gamma(x,y)$ defined in \eqref{eq:2show_sc} is differentiable; denote its partial derivatives with respect to $x$ and $y$ as $\Gamma_1$ and $\Gamma_2$, respectvely.  Furthermore, $\Gamma$ is jointly convex in $(x,y)$ (see Lemma \ref{lem:Mconvex}(ii)) and $\Gamma(0,0) = 0$. Thus, it suffices for \eqref{eq:2show_sc} to prove strict positivity of the following expression
\begin{align}
\Gamma_1(x,y) x + \Gamma_2(x,y) y &= (L_1(\al+x,\tau+y) - L_1(\al,\tau) )x + (L_2(\al+x,\tau+y) - L_2(\al,\tau) )y\nn\\
&= \E\left[(E_1(\al+x,\tau+y) - E_1(\al,\tau) )x + (E_2(\al+x,\tau+y) - E_2(\al,\tau) )y\right].\nn
\end{align}
In the last equality above we have interchanged the order of expectation and differentiation. Lemma \ref{lem:Mconvex}(iv) lower bounds the expression inside the expectation above. To be specific, using \eqref{eq:lowerb_M}, we find that
$$
\Gamma_1(x,y) x + \Gamma_2(x,y) y \geq \left(\tau+\frac{y}{2}\right)\E\left[(\ell'_{\al+x,\tau+y}-\ell'_{\al,\tau})^2\right].
$$
Therefore, it will suffice for our purposes to show that for any fixed $x,y$,
\begin{align}\label{eq:E>0}
\E\left[\left(\ellaG-\ellaxG\right)^2\right] > 0.
\end{align}
For this it is enough to prove the existence of $(G_*,Z_*)$ with $p(Z_*)>0$ such that
\beq
%\ell'_{\al+x,\tau+y}\neq\ell'_{\al,\tau}
\ellaGar{G_*}{Z_*}\neq\ellaxGar{G_*}{Z_*}\label{eq:desire},
\eeq
Indeed, if this is the case, by continuity of the mapping $(G,Z)\rightarrow \al G + Z$ and of the prox operator (cf. Lemma \ref{lem:Mconvex}(i)) there exists an open neighborhood $\Nc$ around $(G_*,Z_*)$ such that
$\ellaG\neq\ellaxG$ for all $(G,Z)\in\Nc$. Furthermore, there exists subset $\Jc_1\times\Jc_2\subseteq\Nc$ of nonzero measure such that: (i) $\Jc_1$ is a closed interval with $p(G)>0$ for all $G\in\Jc_1$, (ii) if $Z$ has a point mass at $Z_*$, then $\Jc_2=Z_*$; otherwise, $\Jc_2$ is a closed interval with $p(Z)>0$ for all $Z\in\Jc_2$. In all cases, $\Jc_1\times\Jc_2$ is a set of nonzero measure, with which we conclude \eqref{eq:E>0} as desired. In what follows, we prove \eqref{eq:desire}.

% closed set $S\subset\R^2$  such that for all $(\G,\Z)\in S$, 
%\beq
%%\ell'_{\al+x,\tau+y}\neq\ell'_{\al,\tau}
%\ellaG\neq\ellaxG\label{eq:desire},
%\eeq
%and $p(Z)>0$.% $\Pc\big((\G,\Z)\in S\big)>0$.

\vp
\noindent{\textit{Case 1}}: Assume that there exists an open interval $\Ic$ on which $\ell$ is  differentiable with strictly increasing derivative:
\begin{align}\label{eq:inc}
\ell'(v_1)<\ell'(v_2),\quad \text{ for all } v_1<v_2 \in\Ic.
\end{align}
In particular, since $\ell$ is convex in its entire domain it further holds that 
\begin{align}\label{eq:further}
v_1\in\Ic, v_2\neq v_1 \Rightarrow \ell'(v_1)\neq \ell'(v_2).
\end{align}
%Denote, $\Jc= (\Id + \tau\ell')^{-1}(\Ic)$. By continuity of the prox operator (Lemma \ref{lem:Mconvex}(i)), $\Jc$ is a nonempty open set. Next, l
Consider the set
\begin{align}
\label{eq:SI}
 \Sc:=\{ (G,Z) ~|~ \vaG\in\Ic\}.
 \end{align}
Clearly, $\Sc$ is a nonempty open set (by continuity of the prox operator). 
%Thus, $\Sc$ is a non-empty set. % that is relatively open to the set $\Zc=\{Z~|~p(Z)>0\}$.
% By definition,
%\begin{align}\label{eq:SI}
%(G,Z)\in\Sc \Leftrightarrow\vaG\in\Ic.
%\end{align}
Next, we show that there exists $(G_*,Z_*)\in\Sc$, such that
\begin{align}\label{eq:kala8atan}
\vaGar{G_*}{Z_*}\neq\vaxGar{G_*}{Z_*}
\end{align}
and $p(Z_*)>0$.
This suffices for proving \eqref{eq:desire}, since when combined with $\vaGar{G_*}{Z_*}\in\Ic$ and \eqref{eq:further} it implies that $\ellaGar{G_*}{Z_*}\neq\ellaxGar{G_*}{Z_*}$. 

Choose any two distinct $Z_i, i=0,1$ with $p(Z_i)>0$. This is possible since by assumption $\mathrm{Var}[Z]\neq 0$. Denote, $\Sc_{\Z_i}:=\{G~|~(G,Z_i)\in\Sc\}$. Clearly, $\Sc_{\Z_i}$ are nonempty open sets. If there exists $G_i\in\Sc_{\Z_i}$ such that $(G_*,Z_*)=(G_i,Z_i)$ satisfies \eqref{eq:kala8atan}, there is nothing else to prove.

 Otherwise, we would have $\vaGar{G}{Z_i}=\vaxGar{G}{Z_i}\in\Ic$ and consequently $\ellaGar{G}{Z_i}=\ellaxGar{G}{Z_i}$, for all $G\in\Sc_{Z_i}$ and $i=0,1$.  But, Lemma \ref{lem:impossible} below proves that this cannot happen under our assumptions on the sets $\Sc_{Z_i}$.

%This is of course a contradiction.

% Now that there exists $(G_*,Z_*)\in\Sc$ satisfying \eqref{eq:kala8atan}, consider the function
%$$
%f(G,Z) = |\vaG - \vaxG|.
%$$
%Clearly, $f(G_*,Z_*)>0$. Also, $f$ is continuous because of continuity of both the mapping $(G,Z)\rightarrow \al G + Z$ and of the prox operator (Lemma \ref{lem:Mconvex}(i)). Thus, there exists an open neighborhood $\Nc$ around $(G_*,Z_*)$ such that $f(G,Z)>0$ for all $(G,Z)\in\Nc$. Also, since $(G_*,Z_*)\in\Sc$ and $\Sc$ is relatively open to $\Zc$, the set $\Sc:=\Sc\cap\Nc$ is non-empty and relatively open to $\Zc$. For all $(G,Z)\in\Sc\subset \Sc$, \eqref{eq:kala8atan} is satisfied and $\vaG\in\Ic$. Hence, from \eqref{eq:further},
%$$
%\ellaG \neq \ellaxG, \quad\text{ for all } (G,Z)\in\Sc.
%$$

\vp
\noindent{\textit{Case 2}}: Let $v_0$ be a point where $\ell$ is not differentiable and consider $\Ic\subset\partial\ell(v_0)$ a non-empty open subset of the subdifferential of $\ell$ at $v_0$. Further consider the nonempty open sets
\begin{align}
\label{eq:SI}
 \Sc:=\{ (G,Z) ~|~ \ellaG\in\Ic\} \quad \text{ and } \quad  \tilde\Sc:=\{ (G,Z) ~|~ \ellaxG\in\Ic\}.
 \end{align}
Clearly, $\vaG=v_0$ for all $(G,Z)\in\Sc$ and similar for $\tilde\Sc$. Choose any two distinct $Z_i, i=1,2$ with $p(Z_i)>0$. This is possible since by assumption $\mathrm{Var}[Z]\neq 0$. Denote, $\Sc_{\Z_i}:=\{G~|~(G,Z_i)\in\Sc\}$ and $\tilde\Sc_{\Z_i}:=\{G~|~(G,Z_i)\in\tilde\Sc\}$, which are all nonempty sets. Consider,
$$
\Nc_{Z_i}=\Sc_{Z_i}\setminus \tilde\Sc_{Z_i}, i=0,1.
$$ 

\noindent If (say) $\Nc_{Z_0}\neq\emptyset$, then for any $G_0\in\Nc_{Z_0}$, it holds
 $$\vaxGar{G_0}{Z_0}\neq\vaGar{G_0}{Z_0}\in\Ic \Rightarrow \ellaxGar{G_0}{Z_0}\neq\ellaGar{G_0}{Z_0},$$
 where the last implication follows because of monotonicity of the subdifferential.  This shows \eqref{eq:desire} as desired.
 
Otherwise, $\Nc_i=\emptyset\Rightarrow$ for $i=0,1$. In case there exists $G_i\in\Sc_{Z_i}$ such that $(G_*,Z_*)=(G_i,Z_i)$ satisfies \eqref{eq:desire}, there is nothing else to prove.
If this was not the case, then we would have $\ellaGar{G_i}{Z_i}=\ellaxGar{G_i}{Z_i}\in\Ic$ and $\vaGar{G_i}{Z_i}=\vaxGar{G_i}{Z_i}=v_0$, for all $G\in\Sc_{Z_i}$. But, Lemma \ref{lem:impossible} below proves that this cannot happen under our assumptions on the sets $\Sc_{Z_i}$.
%When combined with optimality conditions (cf. \eqref{eq:opt}), these give

\vp
\noindent{\emph{\textbf{(iii):}}}
Suppose that the statement of the lemma is false. Then, there exist $\alpha_1\neq \alpha_2>0$, and, $\alpha_\theta:=\theta\alpha_1 + (1-\theta)\alpha_2$ for $\theta\in(0,1)$ such that $F(\theta\alpha_1 + (1-\theta)\alpha_2) = \theta F(\alpha_1) + (1-\theta)F(\alpha_2)$, or, 
\begin{align}\label{eq:wouldBe33}
\Exp\left[ \theta\ell( \alpha_1 G + Z ) + (1-\theta)\ell( \alpha_2 G + Z ) - \ell( \alpha_\theta G + Z ) \right]= 0
\end{align}
The convexity of $\ell$ ensures that, for each $\al G + Z$ , the argument in the expectation is nonnegative. Therefore, the relation above holds if and only if the argument under the expectation is zero almost surely with respect to the distribution of $\al G + Z$. Next, we prove that this leads to a contradiction.

Let $x_+$ as in the statement of the lemma, and  $x_0=\max\{ x\in[0,x_+] ~|~ \ell(x)=0 \}<x_+$. For some $\eps>0$ to be specified later in the proof, let $x_1 = x_0 + \eps$. Note that $\ell(x_1)>0$ by definition of $x_0$ and by convexity.  Without loss of generality assume $\alpha_1>\alpha_2$. Fix $Z_0$ such that $p(Z_0)>0$ and $x_0\neq Z_0$ (always possible since $\mathrm{Var}[Z]\neq 0$).  Consider two cases based on the sign of  $x_0-Z_0$.

\vp
\noindent{\underline{ $x_0>Z_0$:}}~ Define $G_0=(x_1-Z_0)/\alpha_1>0$. Note that $\alpha_1 G_0 + Z_0=x_1$ and call $x_2:=\alpha_2 G_0 + Z_0$.  Choose $\eps=(\frac{\alpha_1}{\alpha_2}-1)(x_0-Z_0)/2>0$. Then, it is not hard to check that $x_2<x_0<x_1$; thus, for some $\theta\in(0,1)$, $\alpha_\theta G_0+Z_0=x_0$. 
But, $\theta \ell(x_1) + (1-\theta)\ell(x_2) > 0 = \ell(x_0)$, or,
\begin{align}
\ell(\alpha_\theta G_0+Z_0) < \theta \ell(\alpha_1 G_0 + Z_0)  + (1-\theta) \ell(\alpha_2 G_0 + Z_0).
\end{align}
There exists an open ball (of non-zero measure) around $\alpha G_0 + Z_0$, where the same relation as above holds. This contradicts \eqref{eq:wouldBe33} and concludes the proof.

\vp
\noindent{\underline{ $x_0<Z_0$:}}~
Define $G_0:={x_1-Z_0}/{\alpha_2}>0$.  Note that $\alpha_2 G_0 + Z_0=x_1$ and call $x_2:=\alpha_1 G_0 + Z_0$.  Choose $\eps=(\frac{\alpha_2}{\alpha_1}-1)(x_0-Z_0)/2>0$. Then, it is not hard to check that $x_2<x_0<x_1$ and the same argument as above leads to a contradiction of \eqref{eq:wouldBe33}.

\end{proof}

\vspace{5pt}
 \begin{lem}[Auxiliary]\label{lem:impossible}
 Suppose $Z_0\neq Z_1$. For some nonempty set $\Jc\subset\R$ assume that the sets
\begin{align}\label{eq:lemass1}
 \Gc_{Z_i} := \{ G~|~ \vaGar{G}{Z_i}\in\Jc \}, \quad i=0,1
\end{align}
 are non-empty and have at least two elements each. Further suppose that for all $G,G'\in\Gc_i, i=0,1$ the following holds
\begin{align}\label{eq:lemass2}
\ellaGar{G}{Z_i}=\ellaGar{G'}{Z_i} \Rightarrow \vaGar{G}{Z_i}=\vaGar{G'}{Z_i}.
\end{align}
 Then, it cannot be true that  for all $G\in\Gc_i$ and $i=0,1$:
 \begin{align}\label{eq:do they?}
 \vaGar{G}{Z_i} =  \vaxGar{G}{Z_i} \quad\text{ and }\quad  \ellaGar{G}{Z_i}=\ellaxGar{G}{Z_i}.
 \end{align}
 \end{lem}
 \begin{proof}
Assume to the contrary of the lemma that the sets $\Gc_0$ and $\Gc_1$ satisfy \eqref{eq:do they?}. When combined with optimality conditions (cf. \eqref{eq:opt}), the properties of the sets give
\begin{align}\label{eq:cantbetrue}
y\ellaGar{G}{Z_i} = xG, \quad\text{for all } G\in\Gc_i, i=0,1.
\end{align}
Consider separately two cases on the possible values of $x$ and $y$:

\noindent$\bullet$~{$x=0,y\neq 0$}: Let $G\neq G'$ both belonging in $\Gc_0$ (such a pair exists since $\Gc_0$ is open). Starting from \eqref{eq:cantbetrue} and using  \eqref{eq:lemass2}, we have:
$$
\ellaGar{G}{Z_0} = \ellaGar{G'}{Z_0} = 0 \Rightarrow \vaGar{G}{Z_0} = \vaGar{G'}{Z_0} .
$$
Those equalities, when combined with optimality conditions of the prox (cf. \eqref{eq:opt}) they yield a contradiction: $G=G'.$

\noindent$\bullet$~{$x\neq 0$}:
Let any $G_0\in\Gc_0$, and, consider $G_1:=G_0+\frac{Z_0-Z_1}{\al}\neq G_0$. Note that $\al G_1 + Z_1 = \al G_0+Z_0$. Also, by uniqueness of the prox operator, $\vaGar{G_1}{Z_1}=\vaGar{G_0}{Z_0}\in\Ic$ and $G_1\in\Gc_1$. Furthermore,  $ \ellaGar{G_0}{Z_0}=\ellaGar{G_1}{Z_1}.$
%where the last implication above follows since $\ell$ is by assumption differentiable at $\vaGar{G_0}{Z_0}$. 
Then, combining with \eqref{eq:cantbetrue} we reach the following contradiction:
$$
x G_0 = x G_1 \Rightarrow G_0 = G_1.
$$
\end{proof}

\subsection{Strict convexity $\implies$ uniqueness of $\alpha_*$}\label{app:uni}

\begin{lem}[]\label{lem:uni}
Suppose all assumptions of Theorem \ref{thm:sep} are satisfied. Then, \eqref{eq:AO_det_thm} has a unique optimal minimizer $\alpha_*$.
\end{lem}

\begin{proof}
During the proof, we borrow notation and results from the proof of Lemma \ref{lem:convergence} in Section \ref{app:convergence}. 
Under the assumption of the theorem, $L(\alpha,\tau)$ is jointly strictly convex in $\R_{>0}\times\R_{>0}$, by Lemma \ref{lem:strict}.
Also, by assumptions, the set of minimizers of $F$ in \eqref{eq:F} is bounded. With these, we will show that the set of optima actually consists of a unique point. 
Consider $M^{\alpha,\beta,\tauh}(\taug)$ as in \eqref{eq:M_taug}. We have shown in Section \ref{app:convergence} that $M^{\alpha,\beta,\tauh}$ is level bounded. Thus, the minimum is either attained at some $\taug_*$ or is achieved in the limit of $\taug\rightarrow0$. Now, consider extending the function at $\taug=0$, by setting $L(\alpha,0) = \lim_{\taug\rightarrow0^+}L(\alpha,\taug)$. By assumption, this latter is a strictly convex function of $\alpha$. Hence, similarly extending $M^{\alpha,\beta,\tauh}$ at $\taug=0$, the function is jointly strictly convex in $(\alpha,\taug)$ and the minimum over $\taug$ is now attained (can be $\taug_*=0$). Using those two, Lemma \ref{lem:ptwstrict_inf} shows that $\inf_{\taug>0}M^{\alpha,\beta,\tauh}(\taug)$ is strictly convex in $\alpha>0$. Next, consider taking the supremum over $\beta\geq 0$. From the results of Section \ref{app:convergence}, the optimal $\beta$ is attained at some value $\beta_*\geq 0$ (in other words, it does not approach infinity). Suppose $\beta_*=0$, then the optimal $\alpha$ solves
\begin{align*}
\inf_{\alpha\geq0}\sup_{\tauh>0} -\frac{\alpha\tauh}{2} + \la F(0,\alpha\la/\tauh).
\end{align*}
In Lemma \ref{lem:beta=0} we show that the set of minimizers of this optimization is unbounded. This contradicts our assumption on the boundedness of $\alpha_*$. Hence, $\beta_*\neq 0$, and we can apply Lemma \ref{lem:ptwstrict} to find that $M^{\alpha}(\tauh):=\sup_{\beta}\inf_{\taug>0}M^{\alpha,\beta,\tauh}(\taug)$, remains a strictly convex function of $\alpha>0$. Lastly, maximizing over $\tauh$ does not affect strict convexity since it is not involved in the term $\frac{\beta\taug}{2}+\delta L(\alpha,\taug/\beta)$. Overall,  $F(\alpha) = \sup_{\beta,\tauh}\inf_{\taug}\Dc(\alpha,\taug,\beta,\tauh)$ is strictly convex in $\alpha>0$. Using this it is straightfowrard to show that its minimizer over $\alpha\geq0$ is unique, thus, completing the proof. 
\end{proof}

%\begin{lem}\label{lem:ptwstrict} 
%Let $\Xc,\Yc$ be convex sets and $F(\x,\y):\Xc\times\Yc\rightarrow\R$ be  strictly convex in its first argument. If $F(\x,\cdot)$ attains its maximum value in $\Yc$ for all $\x\in\Xc$, then, $G(\x):=\sup_{\y\in\Yc} F(\x,\y)$ is strictly convex.
%\end{lem}
%\begin{proof}
%For all $\theta\in[0,1]$ and $\x_1,\x_2\in\Xc$, $\y\in\Yc$,
%\begin{align*}
%F( (1-\theta)\x_1+\theta\x_2 , \y ) &< (1-\theta)F(\x_1,\y) + \theta F(\x_2,\y)\\
%&\leq (1-\theta) G(\x_1) + \theta G(\x_2),
%\end{align*}
%where the first inequality follows from strict convexity of $F$ and the second by definition of $G$. Since the above holds for all $\y$, we can apply it for the $\y$ which maximizes $F( (1-\theta)\x_1+\theta\x_2 , \y )$ to conclude with the desired.
%%$$
%%G((1-\theta)\x_1+\theta\x_2 ) < (1-\theta) G(\x_1) + \theta G(\x_2).
%%$$
%\end{proof}

\begin{lem}\label{lem:ptwstrict_inf} 
Let $\Xc,\Yc$ be convex sets and $F(\x,\y):\Xc\times\Yc\rightarrow\R$ be jointly strictly convex. If $F(\x,\cdot)$ attains its minimum value in $\Yc$ for all $\x\in\Xc$, then, $G(\x):=\inf_{\y\in\Yc} F(\x,\y)$ is strictly convex.
\end{lem}
\begin{proof}
For $\theta\in(0,1)$, $\x_1,\x_2\in\Xc$, denote
$
\x_\theta=  (1-\theta)\x_1+\theta\x_2
$, 
$
\y_{\theta}:=\arg\inf_{\y\in\Yc}{F( (1-\theta)\x_1+\theta\x_2 , \y )}
$
and 
$\y_i:=\arg\inf_{\y\in\Yc}{F( \x_i , \y )}, i=1,2$. With these
\begin{align*}
G(\x_\theta)= F( \x_\theta, \y_\theta )\leq F( \x_\theta, \theta \y_1 + (1-\theta) \y_2 )  &< (1-\theta)F(\x_1,\y_1) + \theta F(\x_2,\y_2)\\
&= (1-\theta) G(\x_1) + \theta G(\x_2),
\end{align*}
where the first inequality follows from definition of $\y_\theta$, the second from the joint strict convexity of $F$, and, the third by definition of $\y_1,\y_2$.
% Since the above holds for all $\y$, we can apply it for the $\y$ which maximizes $F( (1-\theta)\x_1+\theta\x_2 , \y )$ to conclude with the desired.
%$$
%G((1-\theta)\x_1+\theta\x_2 ) < (1-\theta) G(\x_1) + \theta G(\x_2).
%$$
\end{proof}

\begin{lem}\label{lem:ptwstrict} 
Consider $F:\R_{>0}\times\R_{\geq 0}\rightarrow\R$ and $G(x):=\min_{0\leq y\leq K}F(x,y)$. If $F$ is jointly strictly convex in $\R_{>0}\times\R_{>0}$ and $F(x,0)$ is also strictly convex, then $G$ is strictly convex.
\end{lem}
\begin{proof}
Consider $x_1,x_2>0$ and $x_\theta=\theta x_1 + (1-\theta) x_2$ for some $\theta\in(0,1)$. Let $y_1,y_2$ and $y_\theta$ be defined such that $G(x_1)=F(x_1,y_1)$, $G(x_2)=F(x_2,y_2)$, and, $G(x_\theta)=F(x_\theta,y_\theta)$. We distinguish four cases. For each one we prove that $G(x_\theta)<\theta G(x_1) + (1-\theta) G(x_2)$, as desired.

\vp
\noindent{\underline{$y_1,y_2>0$:}}~ 
$$ G(x_\theta) \leq F(x_\theta,\theta y_1 + (1-\theta) y_2 ) < \theta F(x_1,y_1) + (1-\theta) F(x_2,y_2) =  \theta G(x_1) + (1-\theta) G(x_2). $$
The strict inequality follows from the joint convexity of $F$ in $\R_{>0}$.

\vp
\noindent{\underline{$y_1=0,y_2=0$:}}~ 
$$ G(x_\theta) \leq F(x_\theta, 0 ) < \theta F(x_1,0) + (1-\theta) F(x_2,0) =  \theta G(x_1) + (1-\theta) G(x_2). $$
The strict inequality follows from the joint convexity of $F(\cdot,0)$.

\vp
\noindent{\underline{$y_1>0,y_2=0,y_\theta\neq \theta y_1$:}}~ From the strict convexity of $F(x_\theta,\cdot)$ in $\R_{\geq 0}$ it follows that $G(x_\theta)<F(x_\theta,\theta y_1)$. But, from convexity $F(x_\theta,\theta y_1)\leq \theta G(x_1) + (1-\theta)G(x_2)$. 

\vp
\noindent{\underline{$y_1>0,y_2=0,y_\theta= \theta y_1$:}}~ Consider the restriction of $F$ on the line segment passing through points $(x_1,y_1)$, $(x_\theta,y_\theta)$ and $(x_2,0)$. Call it $H(\rho)$ and let be $H(0) = G(x_1)$ and $H(1) = G(x_2)$. Clearly, $H(1-\theta)=G(x_\theta)$. By strict convexity of $F$, it follows that $H(\rho)$ is strictly convex for $0\leq \rho<1$.  Hence,
\begin{align*}
H(1-\theta) = H\left(  \frac{2(1-\theta)}{2-\theta} \left(1-\frac{\theta}{2}\right) \right) &< \frac{\theta}{2-\theta}H(0) + \frac{2(1-\theta)}{2-\theta} H\left(1-\frac{\theta}{2}\right) \\
&\leq \frac{\theta}{2-\theta}H(0) + \frac{2(1-\theta)}{2-\theta} \left( \frac{\theta}{2} H(0) + \frac{2-\theta}{2}H(1)\right)\\
&=\theta H(0) + (1-\theta)H(1).
\end{align*}
The strict inequality follows from strict convexity of $H$ in $(0,1]$. The last inequality is a consequence of convexity of $H$ in $[0,1]$.
%
%Here, $G(x_\theta)$ We show that this case cannot occur. For any $x'\in(x_\theta,x_2)$, it follows as in \eqref{eq:1here} that $G(x')<\theta G(x_\theta) + (1-\theta) G(x_2)$

\end{proof}

\begin{lem}\label{lem:beta=0}
Consider the following optimization
\begin{align}
\inf_{\alpha\geq0}\sup_{\tauh>0} -\frac{\alpha\tauh}{2} + \la \E\left[\env{f}{X_0}{\alpha\la/\tauh}-f(X_0)\right].
\end{align}
The set of minimizers over $\alpha$ is unbounded.
\end{lem}
\begin{proof}
For convenience, denote the objective function as $O(\alpha,\tauh)$ and its optimal value as $O_*$. Let us first perform the optimization over $\tauh$ for fixed $\alpha$. 
We have 
$$
\lim_{\tauh\rightarrow0^+}O(\alpha,\tauh) = \la \Exp\left[\min_{x}f(x) - f(X_0)\right],
$$
where we have used \ref{lem:Mconvex}(vi). Also,
$$
\lim_{\tauh\rightarrow+\infty}O(\alpha,\tauh) = -\infty,
$$
with an appeal to \ref{lem:Mconvex}(ix).
What we learn from these is that 
\begin{align}\label{eq:opt_O}
O_*\geq \la \Exp\left[\min_{x}f(x) - f(X_0)\right]
\end{align}
and that the optimal $\tauh$ either approaches $0$ or is attained. In the latter case, the optimal $\tauh_*$ satisfies the first-order optimality condition:
$$
\frac{1}{\alpha^2}\Exp\left[  (X_0 - \prox{f}{X_0}{{\alpha\la}/{\tauh_*}})^2 \right] = 1.
$$
But for any $\tauh_*>0$, by \ref{lem:Mconvex}(vii), the left hand-side above tends to 0 as $\alpha\rightarrow\infty$. Thus, in the limit $\alpha\rightarrow\infty$, the optimal $\tauh$ approaches 0, giving
$$
\lim_{\alpha\rightarrow+\infty}\sup_{\tauh>0}O(\alpha,\tauh) = \la \Exp\left[\min_{x}f(x) - f(X_0)\right].
$$
When combined with \eqref{eq:opt_O}, this completes the proof of the lemma.
\end{proof}

\section{Useful Properties of Moreau Envelopes}\label{sec:M_prop}

% ------------------------------------------------------------------------------------------- %

In this section we have gathered some very useful properties of Moreau envelopes of convex functions. We have made heavy use of those results for the proofs in Appendix \ref{app:sep}. Some of the results are standard, while others are more tailored towards our interests. 

\begin{lem}[Properties of the Moreau envelope]\label{lem:Mconvex}
Let $\ell:\R\rightarrow\R$ be a proper, closed, convex function. For $\tau>0$, consider its Moreau envelope function and its proximal operator:
\begin{subequations}
\begin{align}
\env{\ell}{\chii}{\tau}&:= \min_{v} ~\frac{1}{2\tau}(\chii -v)^2+\ell(v)\label{eq:env},\\
\prox{\ell}{\chii}{\tau}&:=\arg\min~\frac{1}{2\tau}(\chii -v)^2+\ell(v)\label{eq:prox}
\end{align}
\end{subequations}
The following statements are true:
\begin{enumerate}[(i)]
\item $\prox{\ell}{\chii}{\tau}$ is single valued and continuous. Furthermore,
\begin{align}\label{eq:opt}
\ell'_{\chii,\tau}:=\frac{1}{\tau}(\chii-\prox{\ell}{\chii}{\tau})\in \partial\ell(\prox{\ell}{\chii}{\tau}).
\end{align}
\item $\env{\ell}{\chii}{\tau}$ is jointly convex in $(\chii,\tau)$.
\item $\env{\ell}{\chii}{\tau}$ is continuously differentiable with respect to both $x$ and $\tau$. The gradients are given by:
\begin{align}%\label{eq:gradients}
E_1(\chii,\tau) :=\frac{\partial e_{\ell}}{\partial \chii}&=\frac{1}{\tau}(\chii-\prox{\ell}{\chii}{\tau}) = \ell'_{\chii,\tau},  \label{eq:e_1}\\
E_2(\chii,\tau) := \frac{\partial e_{\ell}}{\partial \tau}&=-\frac{1}{2\tau^2}(\chi-\prox{\ell}{\chii}{\tau})^2=-\frac{1}{2}\left(\ell'_{\chii,\tau}\right)^2\label{eq:e_2}.
\end{align}
\item Fix $\chii$ and $\tau>0$. Consider the function $\Delta:\R\times(-\tau,\infty)\rightarrow\R$:
$$
\Delta(x,y) := ( E_1(\chii+x,\tau+y) - E_1(\chii,\tau) )x + ( E_2(\chii+x,\tau+y) - E_2(\chii,\tau) )y
%\env{\ell}{\chii+x}{\tau+y} - \env{\ell}{\chii}{\tau} - \frac{\partial e_\ell}{\partial{\chii}}\Big|_ {(\chii,\tau)} x - \frac{\partial e_\ell}{\partial{\tau}}\Big|_{(\chii,\tau)} y.
$$
Then,
\begin{align}\label{eq:lowerb_M}
\Delta(x,y) \geq \left(\tau+\frac{y}{2}\right)(\ell'_{\chii+x,\tau+y}-\ell'_{\chii,\tau})^2.
\end{align}

\item $\env{\ell}{x}{\tau}$ is non-increasing in $\tau$.

\item $\lim_{\tau\rightarrow\infty}\env{\ell}{x}{\tau} = \min_{v}\ell(v)$.

\item $\lim_{\tau\rightarrow\infty}\frac{1}{\tau}|x-\prox{\ell}{x}{\tau}|=0$.

\item If $0\in\arg\min_v\ell(v)$, then $\prox{\ell}{x}{\tau}x\geq 0$, $|\prox{\ell}{x}{\tau}|\leq |x|$ and $|\ell^\prime_{\prox{\ell}{x}{\tau},\tau}| \leq |\ell^\prime_{x,\tau}|$.

%
%\item \begin{align}\label{eq:Mco}
\item $\env{\ell}{x_n}{\tau_n}\rightarrow \ell(x)$
%\end{align}
 whenever $x_n\rightarrow x$ while $\tau_n\rightarrow 0^+$ in such a way  that the sequence $\{{|x_n-x|}/{\tau_n}\}_{n\in\mathbb{N}}$ is bounded.

%%
%\item $\env{\ell}{x}{\tau}\leq \frac{x^2}{2\tau}$ 

\end{enumerate}
\end{lem}
\begin{proof}
~
\noindent\textit{(i)} From \cite[Thm.~2.26(a)]{RocVar}, $\prox{\ell}{\chii}{\tau}$ is known to be continuous single valued mapping. Besides, from standard optimality conditions:
\beq
\frac{1}{\tau}(\chii-\prox{\ell}{\chii}{\tau})\in\partial\ell(\prox{\ell}{\chii}{\tau})\nn%\label{eq:eq3_2}
\eeq
For convenience, we have define $\ell'_{\chi,\tau}:=\frac{1}{\tau}(\chii-\prox{\ell}{\chii}{\tau}\in\partial\ell(\prox{\ell}{\chii}{\tau})$. Note that if $\ell$ is differentiable at $\prox{\ell}{\chii}{\tau}$, then $\ell'_{\chii,\tau}$ is the derivative of $\ell$ at that point.

\vp
\noindent\textit{(ii)} Trivially, $h(\chii,v):=(\chii-v)^2$ is a jointly convex function of $v$ and $x$. Thus, its perspective function $\tau h(\frac{\chii}{\tau},\frac{v}{\tau})=\frac{1}{\tau}(\chii-v)^2$ is also jointly convex over $\tau$, $x$ and $v$ and so after minimization over $v$, the function remains jointly convex over $x$ and $\tau$ (cf. \cite[Prop.~2.22]{RocVar}).

\vp
\noindent\textit{(iii)} See \cite[Thm.~2.26(b)]{RocVar} for differentiability with respect to $x$. Next, we mimic the argument to conclude about differentiability with respect to $\tau$. \
It suffices to show that $h(y):=\env{\ell}{\chii}{\tau+y}-\env{\ell}{\chii}{\tau}+\frac{y}{2\tau^2}(\chii-\prox{\ell}{\chii}{\tau})^2$ is differentiable at $y=0$ with $\frac{\partial h}{\partial y}=0$. We know $\env{\ell}{\chii}{\tau}=\frac{1}{2\tau}(\chii -\prox{\ell}{\chii}{\tau})^2+\ell(\prox{\ell}{\chii}{\tau})$, whereas $\env{\ell}{\chii}{\tau+y}\leq\frac{1}{2(\tau+y)}(\chii -\prox{\ell}{\chii}{\tau})^2+\ell(\prox{\ell}{\chii}{\tau})$. Thus,
\begin{align}
h(y)&\leq \frac{1}{2(\tau+y)}(\chii-\prox{\ell}{\chii}{\tau})^2- \frac{1}{2\tau}(\chii -\prox{\ell}{\chii}{\tau})^2 +\frac{y}{2\tau^2}(\chii-\prox{\ell}{\chii}{\tau})^2\nonumber\\
&=\frac{y^2}{2\tau^2(\tau+y)}(\chii-\prox{\ell}{\chii}{\tau})^2\label{eq:eq4}.
\end{align}
Besides because of convexity of $h(y)$, $0=h(0)\leq \frac{1}{2}h(y)+\frac{1}{2}h(-y)$ or equivalently $h(y)\geq -h(-y)$. Thus, \eqref{eq:eq4} gives:
\beq
h(y)\geq \frac{y^2}{2\tau^2(\tau-y)}(\chii-\prox{\ell}{\chii}{\tau})^2\label{eq:eq5}.
\eeq 
Combining \eqref{eq:eq4} and \eqref{eq:eq5} leads to the following
\beq
 \frac{y^2}{2\tau^2(\tau-y)}(\chii-\prox{\ell}{\chii}{\tau})^2\leq h(y)\leq \frac{y^2}{2\tau^2(\tau+y)}(\chii-\prox{\ell}{\chii}{\tau})^2
\eeq
Here, $h(y)$ is sandwiched between two continuously differentiable functions at 0 with zero derivatives. This completes the proof.

\vp
\noindent\textit{(iv)}
From \eqref{eq:e_1} and \eqref{eq:e_2}, we have
\begin{align}
\Delta(x,y) &= (\ell'_{\chii+x,\tau+y}-\ell'_{\chii,\tau})x - \left(\ell'^2(\chii+x,\tau+y)-\ell'^2(\chii,\tau)\right)\frac{y}{2} \nn\\
&=(\ell'_{\chii+x,\tau+y}-\ell'_{\chii,\tau})\left(x - \frac{y}{2}\left(\ell'_{\chii+x,\tau+y}+\ell'_{\chii,\tau}\right)\right).\nn%\label{eq:simple1}.
\end{align}
On the other hand, due to optimality conditions in \eqref{eq:opt},
\begin{align}
\prox{\ell}{\chii+x}{\tau+y}-\prox{\ell}{\chii}{\tau}&=x-(\tau+y)\ell'_{\chii+x,\tau+y}+\tau\ell'_{\chii,\tau}\nonumber\\
&=\left(x-\frac{y}{2}(\ell'_{\chii+x,\tau+y}+\ell'_{\chii,\tau})\right)-(\tau+\frac{y}{2})(\ell'_{\chii+x,\tau+y}-\ell'_{\chii,\tau})\nn.%\label{eq:eq9}.
%\\
%&-\frac{\tau}{2}(\ell'_{\chii+x,\tau+y}-\ell'_{\chii,\tau})\nonumber
\end{align}
%Thus,
%\begin{align}
%x-\frac{y}{2}(\ell'_{\chii+x,\tau+y}+\ell'_{\chii,\tau})&=(\prox{\ell}{\chii+x}{\tau+y}-\prox{\ell}{\chii}{\tau})+(\tau+\frac{y}{2})(\ell'_{\chii+x,\tau+y}-\ell'_{\chii,\tau})\nonumber
%%\\
%%&+\frac{\tau}{2}(\ell'_{\chii+x,\tau+y}-\ell'_{\chii,\tau})\label{eq:eq9}
%\end{align}
Finally, from convexity of $\ell$, it follows from the monotonicity property of the subdifferential that
$$
(\ell'_{\chii+x,\tau+y}-\ell'_{\chii,\tau})(\prox{\ell}{\chii+x}{\tau+y}-\prox{\ell}{\chii}{\tau})\geq 0.
$$
Combining the three displays above gives the desired inequality.

\vp
\noindent\textit{(v)}~
This follows directly by non-positivity of the derivative as in \eqref{eq:e_2}.

\vp
\noindent\textit{(vi)}~
Using the decreasing nature of $\env{\ell}{x}{\tau}$ w.r.t. $\tau$, we have
\begin{equation}
\lim_{\tau\rightarrow\infty}\env{\ell}{x}{\tau}=\text{inf}_{\tau>0}\min_{v}\frac{1}{2\tau}(x-v)^2+\ell(v)= \min_{v}\text{inf}_{\tau>0}\frac{1}{2\tau}(x-v)^2+\ell(v)= \min_{v}\ell(v).\nn
\end{equation}

\vp
\noindent\textit{(vii)}~
Fix an $\epsilon>0$. Since $\lim_{\tau\rightarrow\infty}\env{\ell}{x}{\tau}=\min_v\ell(v)$, there exist $T^\prime_\eps$ such that for all $\tau\geq T_\eps:=\max\{2,T^\prime_\eps\}$, 
\begin{equation}\nn%\label{l1_levelbounded_1}
|\env{\ell}{x}{\tau}-\min_v\ell(v)|=\frac{1}{2\tau}(x-\prox{\ell}{x}{\tau})^2+(\ell(\prox{\ell}{x}{\tau})-\min_v\ell(v))<\epsilon^2
\end{equation}
Then, $\frac{1}{2\tau}(x-\prox{\ell}{x}{\tau})^2<\epsilon^2$, which gives
% $|x-\prox{\ell}{x}{\tau}|<\epsilon\sqrt{2\tau}$, since otherwise $\frac{1}{2\tau}(x-\prox{\ell}{x}{\tau})^2>\epsilon^2$, which would contradict \eqref{l1_levelbounded_1}. Thus,
% for any $\tau>\sqrt{2T}$ we have
\begin{equation}
\frac{1}{\tau}|x-\prox{\ell}{x}{\tau}|<\eps\sqrt{\frac{2}{\tau}}\leq\eps\sqrt{\frac{2}{T_\eps}}\leq \eps\nn
%~~~\frac{\sqrt{2T}\epsilon}{\sqrt{2T}}=\epsilon
\end{equation}
Therefore, $\lim_{\tau\rightarrow\infty}\frac{1}{\tau}|x-\prox{\ell}{x}{\tau}|=0$.

\vp
\noindent\textit{(viii)}~
By \eqref{eq:opt} and the assumption $0\in\arg\min_v{\ell(v)}$, we find $\prox{\ell}{0}{\tau}=0$.
Monotonicity of the prox operator \cite[Prop.~12.19]{RocVar}, gives
$
(\prox{\ell}{x}{\tau} - \prox{\ell}{0}{\tau}) x\geq 0,
$
which then shows $\prox{\ell}{x}{\tau}x\geq 0$. Also, monotonicity of the subdifferential of $\ell$ gives $\ell^\prime_{x,\tau}x\geq 0$. Those two, when combined with optimality conditions in \eqref{eq:opt} give
$$
x-\prox{\ell}{x}{\tau} = \tau\ell^\prime_{x,\tau}\implies x^2\geq \prox{\ell}{x}{\tau}x \implies |x| \geq | \prox{\ell}{x}{\tau}x|.
$$
It remains to show that $\max_{s\in\partial\ell(\prox{\ell}{x}{\tau})}|s| \leq \max_{s\in\partial\ell(x)}|s|$.  Since $0\in\arg\min_v{\ell(v)}$, it follows by convexity that $$(0\leq x_1\leq x_2~\text{ or }~ x_2\leq x_1\leq 0)\implies \max_{s\in\partial\ell(x_1)}|s|\leq \max_{s\in\partial\ell(x_2)}|s|$$
Observe that the LHS of the implication above is equivalent to $(|x_2|\geq|x_1|~\text{ and }~x_1x_2\geq 0)$. Then apply it for $x_1=\prox{\ell}{x}{\tau}$ and $x_2=x$, to conclude.

\vp
\noindent\textit{(ix)}~ Please see \cite[Thm.~1.25]{RocVar}.

\end{proof}

% ------------------------------------------------------------------------------------------- %

%\input{strict}
%\newpage

\section{Proofs for Section \ref{sec:sim}}\label{sec:proofs_4}

%\subsection{Proof of Corollary \ref{cor:ridge}}
%
%\begin{proof}
%We only need to show that Assumptions \ref{ass:tech} and \ref{ass:uni} are satisfied. The rest follows from writing the first-order optimality conditions of the corresponding (DO) and some algebra. 
%
%The part of Assumption \ref{ass:tech} regarding $\ell,p_z$ follows since Assumption \ref{ass:sep} is true for $\ell,p_z$. For $f,p_{\x_0}$, we have
%$$
%\frac{1}{n}\env{\frac{1}{2}\|\cdot\|_2^2}{c\h+\x_0}{\tau} =  \frac{1}{n}\sum_{i=1}^n \frac{1}{2}\frac{\|c\h+\x_0\|_2^2}{1+\tau}  = \frac{1}{2(\tau+1)}\left(c^2\frac{\|\h\|_2^2}{n}  + \frac{\|\x_0\|_2^2}{n} + 2c\frac{\h^T\x_0}{n} \right) \rP \frac{ c^2 + \sigma_x^2 }{2(\tau+1)}.
%$$
%Hence, Assumption \ref{ass:tech} holds with $F(c,\tau)=\frac{ c^2 + \sigma_x^2 }{2(\tau+1)}-\sigma_x^2$, and, $L(c,\tau)$ as in \eqref{eq:Lm_sep}. 
%
%We prove that Assumption \ref{ass:uni} holds by showing that the objective function of the corresponding (DO) is strictly convex in $(\alpha,\taug)$ and strictly concave in $(\beta,\tauh)$. This can be shown by taking second derivatives. The details are omitted for brevity. We only mention that this essentially translates to strict convexity of $L$ and $F$. For $L$ this is guaranteed by Lemma \ref{lem:strict}. For $F$, it can be easily checked via differentiation (also, $\sigma_x^2>0$) that $\frac{c^2+\sigma_x^2}{2(\tau+1)}$ is strictly convex in $(c,\tau)$.
%
%\end{proof}

\subsection{On Remark \ref{rem:LAD_zero}}\label{app:LAD_zero}
Substituting the envelope function of $|\cdot|$ in \eqref{eq:cone_eq} gives:
%\eqref{eq:cone_eq} becomes
\begin{subequations}\label{eq:cone_eq_l2}
\begin{align}
% \left\{
\frac{\beta}{2}+\bar{s}\E\begin{cases}
-\frac{\beta(\GZ)^2}{2\tau^2}&, |\GZ|\leq\frac{\tau}{\beta}\\
-\frac{1}{2\beta}&,\text{otherwise}
\end{cases}+(\delta-\bar{s})\E\begin{cases}
-\frac{\beta\al^2 G^2}{2\tau^2}&, |\al G|\leq\frac{\tau}{\beta}\\
-\frac{1}{2\beta}&,\text{otherwise}
\end{cases}\geq0, \label{eq:cone_4}\\
\bar{s}\E\begin{cases}
\frac{\beta G (\GZ)}{2} &,|\GZ|\leq\frac{\tau}{\beta}\\
G~\sign(\GZ) &, \text{otherwise}
\end{cases} +(\delta-\bar{s})\E\begin{cases}
\frac{\al G^2 \beta}{\tau} &, |\al G\|\leq\frac{\tau}{\beta}\\
G~\sign(G) &, \text{otherwise}
\end{cases}-\beta\sqrt{\DKPo}\geq0, \label{eq:cone_5}\\
\frac{\tau}{2}+\bar{s}\E\begin{cases}
\frac{(\GZ)^2}{2\tau}&, |\GZ|\leq\frac{\tau}{\beta}\\
-\frac{\tau}{2}&,\text{otherwise}
\end{cases}+(\delta-\bar{s})\E\begin{cases}
\frac{\al^2 G^2}{2\tau}&, |\al G|\leq\frac{\tau}{\beta}\\
-\frac{\tau}{2}&,\text{otherwise}
\end{cases}-\al\sqrt{\DKPo} \leq0.\label{eq:cone_6}
%\\
%\taug &= \alpha\sqrt{\DKPo}.
% \right.
 \end{align}
 \end{subequations}
Define $\kappa:=\frac{\tau}{\beta\al}$ and $\rho:=\frac{\tau}{\beta}$. In order to find a sufficient condition for $\al$ to be zero, we assume $\al\rightarrow 0$, $\tau\rightarrow 0$, $\rho\rightarrow 0$ and $\kappa\geq0$ and look for conditions under which the equations in \eqref{eq:cone_eq_l2} are consistent. Under these assumptions, one can check that  \eqref{eq:cone_6} is satisfied (the argument converges to zero), while, \eqref{eq:cone_5} and \eqref{eq:cone_4} become

\begin{subequations}\label{eq:cone_eq_int}
\begin{align}
% \left\{
2\frac{(\delta-\bar{s})}{\kappa}\int_0^\kappa G^2\phi(G)\mathrm{d}G+(\delta-\bar{s})\int_\kappa^\infty G\phi(G)\mathrm{d}G\geq\beta\sqrt{\DKPo},
\label{eq:cone_7}\\
\beta^2\geq\bar{s}+2\frac{\delta-\bar{s}}{\kappa^2}\int_0^\kappa G^2\phi(G)\mathrm{d}G+2(\delta-\bar{s})\int_\kappa^\infty \phi(G)\mathrm{d}G \label{eq:cone_8},
%\\
%\taug &= \alpha\sqrt{\DKPo}.
% \right.
 \end{align}
 \end{subequations}
where $\phi(G)=e^{-G^2/2}/\sqrt{2\pi}$ and we multiplied \eqref{eq:cone_4} by $\beta^2$ to get \eqref{eq:cone_8}. Observe that \eqref{eq:cone_7} upper bounds $\beta$ while \eqref{eq:cone_8} derives a lower bound on it. Thus, consistency of the set of equations \eqref{eq:cone_eq_int} is achieved if the following holds:
\begin{equation}
\frac{1}{\DKPo}(2\frac{(\delta-\bar{s})}{\kappa}\int_0^\kappa G^2\phi(G)\mathrm{d}G+(\delta-\bar{s})\int_\kappa^\infty G\phi(G)\mathrm{d}G)^2
\geq
\bar{s}+2\frac{\delta-\bar{s}}{\kappa2}\int_0^\kappa G^2\phi(G)\mathrm{d}G+2(\delta-\bar{s})\int_\kappa^\infty \phi(G)\mathrm{d}G\nonumber.
\end{equation}
Or, equivalently,
\begin{equation}
\DKPo\leq\frac{(2\frac{(\delta-\bar{s})}{\kappa}\int_0^\kappa G^2\phi(G)\mathrm{d}G+(\delta-\bar{s})\int_\kappa^\infty G\phi(G)\mathrm{d}G)^2}{\bar{s}+2\frac{\delta-\bar{s}}{\kappa^2}\int_0^\kappa G^2\phi(G)\mathrm{d}G+2(\delta-\bar{s})\int_\kappa^\infty \phi(G)\mathrm{d}G}\label{eq:cone_9}.
\end{equation}
Thus if maximum of the right side of \eqref{eq:cone_9} with respect to $\kappa$ is greater than $\DKPo$, all our variables satisfy \eqref{eq:cone_eq} and the optimal value in \eqref{eq:DO_cone} occurs when $\al\rightarrow0$, $\tau\rightarrow0$ and $\frac{\tau}{\al\beta}\rightarrow\kappa$ which means $\al_*=0$. We will show that
\begin{equation}
\max_{\kappa>0} \frac{(2\frac{(\delta-\bar{s})}{\kappa}\int_0^\kappa G^2\phi(G)\mathrm{d}G+(\delta-\bar{s})\int_\kappa^\infty G\phi(G)\mathrm{d}G)^2}{\bar{s}+2\frac{\delta-\bar{s}}{\kappa^2}\int_0^\kappa G^2\phi(G)\mathrm{d}G+2(\delta-\bar{s})\int_\kappa^\infty \phi(G)\mathrm{d}G} \geq
\delta-\min_{\kappa>0} \bar{s}(1+\kappa^2)+2(\delta-\bar{s})\int_\kappa^\infty (G-\kappa)^2\phi(G)\mathrm{d}G\label{eq:cone_10}
\end{equation}

If both this and \eqref{eq:cond_per} are true then,  there will be a $\kappa$ for which \eqref{eq:cone_9} holds and as we discussed, this implies $\al_*=0$.

For convenience, we define $A_\kappa=\int_\kappa^\infty G^2\phi(G)\mathrm{d}G$, $B_\kappa=\int_\kappa^\infty G\phi(G)\mathrm{d}G$ and $C_\kappa=\int_\kappa^\infty \phi(G)\mathrm{d}G$. The optimal $\kappa$ for the right side of \eqref{eq:cone_10} satisfies the following due to the first optimality condition
\begin{equation}
2(\delta-\bar{s})\hat{\kappa} B_{\hat{\kappa}}-2(\delta-\bar{s})\hat{\kappa}^2 C_{\hat{\kappa}}=\hat{\kappa}^2\bar{s}\label{eq:cone_11}
\end{equation}
For this value of $\kappa$, the left side of \eqref{eq:cone_10} becomes
\begin{align}
&\frac{(2\frac{(\delta-\bar{s})}{\hat{\kappa}}\int_0^{\hat{\kappa}} G^2\phi(G)\mathrm{d}G+(\delta-\bar{s})\int_{\hat{\kappa}}^\infty G\phi(G)\mathrm{d}G)^2}{\bar{s}+2\frac{\delta-\bar{s}}{\hat{\kappa}^2}\int_0^{\hat{\kappa}} G^2\phi(G)\mathrm{d}G+2(\delta-\bar{s})\int_{\hat{\kappa}}^\infty \phi(G)\mathrm{d}G}=(\delta-\bar{s})(1-2A_{\hat{\kappa}}+2\hat{\kappa}B_{\hat\kappa})\nonumber\\
&\quad=\delta-\bar{s}-2(\delta-\bar{s})\hat\kappa B_{\hat\kappa}+2(\delta-\bar{s})\hat\kappa^2 C_{\hat\kappa}-2(\delta-\bar{s})A_{\hat\kappa}+4(\delta-\bar{s})\hat\kappa B_{\hat\kappa}-2(\delta-\bar{s})\hat\kappa^2 C_{\hat\kappa}\nonumber\\
&\quad=\delta-\bar{s}(1+\hat\kappa^2)-2(\delta-\bar{s})(A_{\hat\kappa}-2B_{\hat\kappa}+\hat\kappa^2 C_{\hat\kappa})=\delta- \bar{s}(1+\hat\kappa^2)+2(\delta-\bar{s})\int_{\hat\kappa}^\infty (G-\hat\kappa)^2\phi(G)\mathrm{d}G,\nn
\end{align}
where the first and third equalities follow after substituting $\bar{s}$ using \eqref{eq:cone_11}. This proves \eqref{eq:cone_10} as desired to conclude the claim of the remark.

% ******************************************************************** %

\vp
\subsection{On Section \ref{sec:sq_LASSO}}\label{app:square}

\subsubsection{Satisfying Assumptions \ref{ass:tech}(a) and \ref{ass:prop}(b)-(d)}

It only takes a few calculations to show that \
\begin{align}\label{eq:square-root}
\frac{1}{m}\env{\sqrt{n}\|\cdot\|_2}{\al \g+\z}{\tau} = \begin{cases}
\frac{1}{\sqrt{m\delta}}\|\al \g+\z\|_2 - \frac{\tau}{2\delta} &,\text{if}~ \frac{\sqrt\delta\|\al \g+\z\|_2}{\sqrt{m}}\geq\tau,\\ 
\frac{1}{2\tau}\frac{\|\al \g+\z\|_2^2}{m} &,\text{otherwise}.
\end{cases}
\end{align}
Assume that $0< \E\frac{\|\z\|_2^2}{m} =:\sigma^2 <\infty$. From \eqref{eq:square-root}, it can be seen that $\frac{1}{m}\env{\sqrt{n}\|\cdot\|_2}{\al \g+\z}{\tau}$ is a Lipschitz convex function of $\frac{\|\al \g+\z\|}{\sqrt m}$. Also, $\frac{\|\al \g+\z\|}{\sqrt m}$ converges in probability to $\sqrt{\al^2+\sig^2}$, thus
\begin{align}
\frac{1}{m}(\env{\sqrt{n}\|\cdot\|_2}{\al \g+\z}{\tau}-\|\z\|_2\sqrt n) \rightarrow L(\al ,\tau)=\begin{cases}
\frac{\sqrt{\al^2+\sigma^2}-\sigma}{\sqrt\delta}- \frac{\tau}{2\delta} &,\text{if}~ \delta(\al^2+\sigma^2)\geq\tau^2,\\ 
\frac{1}{2\tau}(\al^2+\sigma^2)-\frac{\sigma}{\sqrt\delta} &,\text{otherwise}.
\end{cases}\nonumber
\end{align}
Finally, it remains to show this function satisfies assumption \ref{ass:prop}.
%\begin{itemize}

\vp
\noindent{\underline{Assumption \ref{ass:prop}(b):}}
$\lim_{\tau\rightarrow0}L(\al,\tau)=\frac{\sqrt{\al^2+\sig^2}-\sig}{\sqrt\delta}$ and $\lim_{\tau\rightarrow0}L(0,\tau)=0$. Besides
$$
L_{2,+}(\al,\tau)=\begin{cases}
-\frac{1}{2\delta}& ,\delta(\al^2+\sig^2)\geq\tau^2,\\
-\frac{\al^2+\sig^2}{2\tau^2}&,\text{otherwise}.
\end{cases}
$$
So, $L_{2,+}(0,0)=-\frac{1}{2\delta}$; thus, condition (b) is satisfied.

\vp
\noindent{\underline{Assumption \ref{ass:prop}(c):}}
 $\frac{1}{m}\mathcal{L}(z)\xrightarrow{P} \frac{\sig}{\sqrt\delta}=-\lim_{\tau\rightarrow\infty}L(\al,\tau)$. It is also easy to check that $L(\al,\tau)\geq-\frac{\sig}{\sqrt\delta}$ for all $\al$ and $\tau>0$ because
$$
L(\al,\tau)=\begin{cases}
\frac{\sqrt{\al^2+\sig^2}-\sig}{\sqrt\delta}-\frac{\tau}{2\delta}&,\delta(\al^2+\sig^2)\geq\tau^2\\
\frac{\al^2+\sig^2}{2\tau}-\frac{\sig}{\sqrt\delta}&,\text{otherwise}
\end{cases}\geq\begin{cases}
\frac{\tau}{2\delta}-\frac{\sig}{\sqrt\delta}&,\delta(\al^2+\sig^2)\geq\tau^2\\
\frac{\al^2+\sig^2}{2\tau}-\frac{\sig}{\sqrt\delta}&\text{otherwise}
\end{cases}\geq-\frac{\sig}{\sqrt\delta}.
$$
Therefore condition (c) is satisfied. Besides, since $L_0=\frac{\sig}{\sqrt\delta}<\infty$, there is nothing to check regarding condition \ref{ass:prop}(d).

\subsubsection{Proving \eqref{eq:square-root_4}$\Leftrightarrow$\eqref{eq:square-root_7}}

It suffices to show that $H(\beta):=-\frac{\al\beta^2}{2\tauh} +F(\frac{\al\beta}{\tauh},\frac{\al\la}{\tauh})$ is a non-increasing function of $\beta>0$. 
We prove this for the separable function $f$ satisfying the assumptions of Theorem \ref{thm:sep} where $F(c,\tau)=\E[\env{f}{cH+X_0}{\tau}-f(X_0)]$. Using Lemma \ref{lem:strict}(i) and \ref{lem:Mconvex}(iii), we find 
\begin{align}\nn%\label{eq:square-root_6}
\lim_{c\rightarrow0^+}\frac{\partial}{\partial c}\left(\E[\env{f}{cH+X_0}{\tau}-f(X_0)]\right)=\frac{1}{\tau }\E[H(X_0-\prox{f}{X_0}{\tau})].
\end{align}
Because of independence of $\h$ and $\x_0$, the RHS above is zero. Thus $\lim_{c\rightarrow0^+}\frac{\partial}{\partial c}~F(c,\tau)=0$ which when combined with concavity of $H$, it shows that it is non-increasing for $\beta>0$.

\end{document}